\documentclass{article}

\usepackage{arxiv}

\usepackage[colorlinks=true, linkcolor=darkblue, citecolor=darkblue, urlcolor=darkblue]{hyperref}
\usepackage[utf8]{inputenc} 
\usepackage[T1]{fontenc}    
\usepackage{url}            
\usepackage{booktabs}       
\usepackage{amsfonts}       
\usepackage{nicefrac}       
\usepackage{microtype}      
\usepackage{lipsum}         
\usepackage{graphicx}
\usepackage{doi}

\usepackage{amsmath, amssymb, amsthm}
\usepackage{graphicx}
\usepackage{enumitem}
\usepackage{subcaption}
\usepackage{algorithm}
\usepackage{algorithmic}
\usepackage{xcolor}
\usepackage{tabularx}
\usepackage{booktabs}
\usepackage{multirow}
\usepackage{comment}
\usepackage{float}
\usepackage[nomarkers,nolists,figuresonly]{endfloat}

\usepackage{threeparttable}
\usepackage{setspace}
\usepackage{dsfont}
\usepackage{appendix}
\usepackage{authblk}

\usepackage[english]{babel}
\newtheorem{theorem}{Theorem}[section]
\newtheorem{lemma}[theorem]{Lemma}
\newtheorem{assumption}[theorem]{Assumption}

\DeclareMathOperator*{\argmin}{argmin}

\newcommand{\EE}[0]{\mathbb{E}}

\definecolor{darkblue}{rgb}{0.0, 0.0, 0.5}
\usepackage{natbib}

\newcommand{\pl}{\parallel}
\newcommand{\openr}{\hbox{${\rm I\kern-.2em R}$}}
\newcommand{\openn}{\hbox{${\rm I\kern-.2em N}$}}

\title{Targeted Highly Adaptive Lasso Minimum Loss Estimation of Target Functions}

\author[1]{Vanessa Rodriguez}
\author[1]{Karla Diaz-Ordaz}
\author[1]{Brieuc Lehmann}
\author[2]{Mark J. van der Laan}

\affil[1]{Department of Statistical Science, University College London, UK}
\affil[2]{Division of Biostatistics, University of California, Berkeley, USA}

\date{\today}

\begin{document}
\maketitle

\begin{abstract}
We propose a Targeted Highly Adaptive Lasso for estimation of non-pathwise differentiable functional parameters such as the dose-response curve (DRC) for continuous exposure. 
We assume the target function lies in the $k$-th order smoothness class used to define the $k$-th order Highly Adaptive Lasso (HAL), which can be well approximated by linear spans of $k$-th order spline basis functions.
We construct a projection of the true target function onto a large finite dimensional working model spanned by an initial set of $k$-th order spline basis functions, which defines a pathwise differentiable approximation of the target functional parameter. A standard TMLE is then applied with a data-adaptive initial fit, replacing the MLE targeting step with a LASSO step over HAL spline basis functions that span the target function.
We prove that the resulting Targeted HAL-MLE is pointwise asymptotically normally distributed and achieves a convergence rate determined solely by the dimension and smoothness of the target function, giving dimension free rates up till $\log n$-factors.
Through a simulation study for the DRC, we show that the Targeted HAL outperforms a HAL plug-in estimator in terms of bias and mean squared error.
Targeted HAL offers a fully data-adaptive approach to inference on functional parameters without requiring sieve specification or parametric assumptions.
\end{abstract}

\section{Introduction}

Targeted minimum loss estimation (TMLE) is a general framework for statistical estimation in causal inference \citep{vanderLaan:Rose11}. 
Given $n$ i.i.d. copies of $O\sim P_0\in {\cal M}$ and a pathwise differentiable target parameter $\Psi:{\cal M}\rightarrow \mathbb{R}^d$, TMLE constructs an initial estimator of the relevant part $Q(P_0)$ of $P_0$ minimizing a loss-based empirical risk, and then updates it through a fluctuation step that solves the efficient influence curve estimating equation.
TMLE provides a plug-in efficient estimator of pathwise differentiable target features of the data distribution \citep{bkrw1997,vanderLaan:Rose11,vanderLaan&Rose18}. 
It can be used to target infinite dimensional pathwise differentiable functions with universal least favorable paths such as survival curves \citep{rytgaard_targeted_2022,rytgaard_estimation_2023,vanderLaan&Rose18,vanderLaan&Gruber15}.
The standard TMLE fluctuation step, however, is designed to solve a finite-dimensional score equation, and therefore does not directly extend to non-pathwise differentiable functional target parameters. 
The goal of this article is to construct estimators that provide statistical inference for such non-pathwise differentiable target functions.

We illustrate this using the confounder-adjusted Dose Response Curve (DRC), defined as
\[
\psi_0(a)=\Psi(P_0)(a)\equiv E_{P_0}E_{P_0}(Y\mid A=a,W),
\]
where $A$ is a univariate continuous treatment, such as a dose of a drug, $Y$ is a binary or continuous outcome, and $W$ is a $(\bar{d}-1)$-dimensional vector of baseline covariates. 
We can interpret $\psi_0(a)$ causally, $\psi_0(a)=E_0 Y(a)$, assuming conditional exchangeability given $W$, i.e. $A$ is independent of $Y(a)$, the potential outcome under treatment $A=a$, given $W$, and positivity  of the treatment mechanism $g_0(A\mid W)=P_0(A\mid W)$, that is 
 $g_0(a\mid W)>0$ almost everywhere.
More generally, $\psi_0$ represents a confounder adjusted response curve of interest establishing the association of $A$ with $Y$ after controlling for the other measured factors $W$ affecting the outcome.

We take the Highly-Adaptive LASSO (HAL) as our starting point. HAL approximates $\bar{Q}_0$, using a finite linear combination of indicator basis functions. The $\bar{k}$-th order HAL-MLE $\bar{Q}_n$ minimizes the empirical risk $P_n L(\bar{Q})$ over the smoothness class $D^{(\bar{k})}_M([0,1]^{\bar{d}})$ of c\`adl\`ag functions with $k$-th order sectional variation norm bounded by $M$, $k=0,\ldots$, which can be represented as the closure of the linear span of $\bar{k}$-th order spline basis functions with $L_1$-norm bounded by $M$\citep{van2023higher,vanderLaan17,benkeser_hal_2016}.  Going forward, this $\bar{k}$-th order sectional variation norm is implicitly assumed to be finite, though for simplicity of notation we drop the subscript $M$ unless necessary. Subspaces $D^{(\bar{k})}({\cal R}(\bar{d}))$ are spanned by selected subsets ${\cal R}(\bar{d})$ of these basis functions, enabling computation via standard LASSO software \citep{van2023higher}. The $\bar{k}$-th order HAL-MLE achieves dimension-free rates $O^+(n^{-\bar{k}^*/(2\bar{k}^*+1)})$ of convergence, with $\bar{k}^*=\bar{k}+1$, where we use notation $O^+(n^{-\delta})$ for a rate that is $O(n^{-\delta})$ up to a power of a $\log n$-factor \citep{van2023higher}. 
We can then obtain a plug-in HAL-MLE estimator of the DRC by marginalizing $\bar{Q}_n$, the estimated conditional outcome regression, over the empirical distribution of $W$, $\psi_n(a)=Q_{W,n}\bar{Q}_n(a,\cdot)$. This results in a pointwise asymptotically normal estimator (via a Delta-method argument \citep{van2023higher,shi2025halbasedpluginestimationpointwise}). However, $\bar{Q}_n$ is not targeted towards the lower-dimensional target function $\psi_0$, resulting in finite sample bias and motivating a targeted approach.

Our proposed approach is to introduce a sieve for the target function $\psi_0$: a finite-dimensional spline working model growing in dimension with sample size, that can uniformly approximate any element of the target function space at a rate controlled by its dimension. 
Projecting $\psi_0$ onto this working model defines a pathwise differentiable approximation $\Psi_{\lambda}(P_0)$ of $\Psi(P_0)$, to which a standard TMLE can be applied. 
The dimension, $n(\lambda)$, of the working model acts as a tuning parameter governing the bias-variance tradeoff.
Choosing it optimally achieves dimension-free rates of convergence determined solely by the smoothness of $\psi_0$, not the smoothness of the higher-dimensional true outcome regression $\bar{Q}_0$, which is particularly advantageous when $\psi_0$ is substantially smoother than $\bar{Q}_0$. 
This approach extends the working marginal structural model (MSM) literature \citep{Robins:Hernan:Brumback00,Neugebauer:vanderLaan07} from fixed working models to actual estimation of the target function. 

We first establish the pointwise asymptotic normality of a pre-specified fixed sieve TMLE with respect to the projection of the target function, and show that with some undersmoothing (choosing a larger sieve) asymptotic normality also holds with respect to the true target function.
We next let the working model be constructed data-adaptively, and show that the fixed-sieve asymptotic normality results still carry over to the data-adaptive TMLE.
Constructing such a working model with the required approximation rate is straightforward for the univariate DRC but becomes challenging for higher-dimensional target functions.
We therefore propose the Targeted HAL-MLE (T-HAL-MLE), which avoids explicit sieve construction by starting from a large initial working model and replacing the MLE targeting step with a LASSO step that selects the relevant basis functions.

Several related works address non-pathwise differentiable parameter estimation via a similar two-step structure: an initial outcome regression followed by a targeted update using a sieve of basis functions in the style of TMLE. 
An earlier targeted estimator for the DRC was proposed in \citep{Diaz13c}, but without the pointwise normality results we establish here. 
The more recent i-learner \citep{vansteelandt_orthogonal_2025} and EP-learner \citep{van_der_laan_combining_2024} both update the initial outcome regression estimate so that a plug-in loss becomes asymptotically Neyman-orthogonal, achieving the minimax convergence rate for smooth target functions. 
Both methods require the user to prescribe the sieve dimension and its growth rate with sample size, which is straightforward for a univariate target function but non-trivial for higher-dimensional multivariate functions. 
Moreover, they target the parameter directly rather than a pathwise differentiable projection of it, which complicates valid variance estimation. 
The T-HAL-MLE instead uses a LASSO at the targeting step to data-adaptively select the relevant basis functions from a large initial working model, avoiding explicit sieve construction. 
Crucially, because the T-HAL-MLE targets a pathwise differentiable projection $\Psi_\lambda(P_0)$ of $\psi_0$ rather than the non-pathwise differentiable parameter directly, we are able to provide a formal proof of pointwise asymptotic normality at dimension-free rates. 
To our knowledge, no analogous guarantee has been established for the EP-learner or the i-learner.

The remainder of this article is structured as follows. Section \ref{Chtargetfunction2} discusses the plug-in HAL-MLE as a benchmark and motivates the need for targeted estimation. 
Section \ref{Chtargetfunction3} develops the sieve-based TMLE for the projection of the DRC onto a fixed working model, where the sieve is  fixed. 
Section \ref{Chtargetfunction4} generalizes this method to data adaptively determined sieves, thereby paving the way for the proposed T-HAL-MLE.
Section \ref{section4b} introduces the T-HAL-MLE and establishes its pointwise asymptotic normality, with the general theorem for arbitrary non-pathwise differentiable target functions given in Section \ref{Chtargetfunction5}. 
Section \ref{sec:simulation} reports a simulation study comparing the HAL-MLE with the T-HAL-MLE under a range of challenging outcome models and treatment distributions.
We conclude with a discussion in Section \ref{Chtargetfunction7}. 
The Appendix covers methods for computing the canonical gradient of the projection parameter under more general HAL models, and provides a further example of estimating the conditional causal effect of a binary treatment on a survival outcome.

\begin{table}
\centering
\caption{Notations and their descriptions.}
\label{tab:notation}
\begin{tabular}{p{0.33\textwidth} p{0.62\textwidth}}
\toprule
\textbf{Notation} & \textbf{Description} \\
\midrule
$n$& Sample size\\[2pt]
$O=(W,A,Y)\sim P_0\in{\cal M}$ & Observed unit: $W$ covariates, $A$ dose, $Y$ outcome \\[2pt]
$P_0$, $P_n$ & True data-generating and empirical distributions \\[2pt]
$\Psi:{\cal M}\to D^{(k)}([0,1]^{d})$; $\psi_0=\Psi(P_0)$ & Target function parameter and its true value \\[2pt]
$\bar{Q}_P=E_P(Y\mid A,W)$& Outcome regression\\[2pt]
$\bar{Q}_n$ & Initial outcome estimator \\[2pt]
$g_0(a\mid W)$& True conditional density of exposure \\[2pt]
$Q_W$ & Marginal distribution over $W$ \\[2pt]
$D^{(k)}_M([0,1]^{d})$ & $k$-th order smoothness class, sectional variation norm $\leq M$ \\[2pt]
$D^{(k)}({\cal R}(d))$ & Subspace spanned by $k$-th order splines in ${\cal R}(d)$ \\[2pt]
$D^{(k)}({\cal R}_{\lambda})$; $n(\lambda)$; $\lambda$ & Fixed companion sieve for $\psi_0$; its dimension; tuning parameter \\[2pt]
$D^{(k)}({\cal R}_{\lambda,n})$; $n(\lambda,n)$ & Data-adaptive sieve used by the TMLE; its dimension ($n(\lambda,n)\geq n(\lambda)$) \\[2pt]
$\phi_j$, $\bar{\phi}_j$ & $k$-th / $\bar{k}$-th order spline basis function (target / outcome) \\[2pt]
$\omega(a)$ & Weight function for the $L^2$-projection defining $\Psi_\lambda$ \\[2pt]
$\Sigma_{\phi_\lambda} = \int \omega(a)\,{\bf\phi}_\lambda{\bf\phi}_\lambda^\top da$ & Gram matrix  \\[2pt]
$\alpha_\lambda(P)$, $\alpha_{\lambda,n}(P)$  & $L^2$-projection coefficients (fixed / data-adaptive)\\[2pt]
$\Psi_{\lambda}(P)=\alpha_{\lambda}(P)^\top{\bf\phi}_\lambda$ & Fixed-sieve projection of target function \\[2pt]
$\Psi_{\lambda,n}(P)=\alpha_{\lambda,n}(P)^\top{\bf\phi}_{\lambda,n}$ & Data-adaptive projection of target function \\[2pt]
$\bar{Q}_{n,\lambda}^*$ & TMLE-update of $\bar{Q}_{n}$ targeting $\alpha_\lambda(P_0)$ \\[2pt]
$D^*_{\alpha_\lambda(),P}$ & Canonical gradient of the projection coefficients $\alpha_{\lambda}$ at $P$ \\[2pt]
$D^*_{\Psi_\lambda(),a,P}={\bf \phi}_\lambda(a)^\top D^*_{\alpha_\lambda(),P}$ & Canonical gradient of $\Psi_{\lambda}(P)(a)$ at $P$ \\[2pt]
$\sigma^2_\lambda(a)=P_0\{D^*_{\Psi_\lambda(),a,P_0}\}^2$ & Variance of canonical gradient \\[2pt]
\bottomrule
\end{tabular}
\end{table}

\section{Plug-in HAL-MLE of the DRC}\label{Chtargetfunction2}

Suppose we observe $n$ i.i.d. copies $O = (W, A, Y) \sim P_0 \in \cal{M}$ and assume a $\cal{M}$ is a nonparametric model beyond possible assumptions on $g_0(A \mid W)$. 
Throughout we use notation $Pf \equiv \int f(o) \, dP(o)$. 
The DRC can be written as $\Psi(P) = \Psi^F(Q_P) = Q_W \Phi(\bar{Q}_P)= \int \Phi(\bar{Q}_P(A, w))dQ_W(w)$, where $Q_P = (Q_W, \bar{Q}_P)$ consists of the marginal distribution of $W$ and the conditional outcome regression $\bar{Q}_P(A,W) = E_P[Y \mid A, W]$, $\Psi^F$ is the functional mapping $Q_P$ to the estimand, and $\Phi(x) = x$ if $Y$  is continuous and $\Phi(x) = 1/(1 + \exp(-x))$ if $Y$ is binary.
In the continuous case, $\psi(a) = \Psi(P)(a) = Q_W \bar{Q}(a, \cdot)$ is bilinear in $(Q_W, \bar{Q})$, while in the binary case, $\psi(a) = \Psi(P)(a) = Q_W \{1 + \exp(-\bar{Q}(a, \cdot))\}^{-1}$ is non-linear in $\bar{Q}(a, \cdot)$. 

We consider the statistical model, 
\[
{\cal M}({\cal R}(\bar{d})) = \{P : \bar{Q}_P \in D^{(\bar{k})}({\cal R}(\bar{d})), \, \Psi(P) \in D^{(k)}([0,1]^{d})\},
\]
for a specified spline model $D^{(\bar{k})}({\cal R}(\bar{d})) \subset D^{(\bar{k})}([0,1]^{\bar{d}})$ for $\bar{Q}_0$. 
The smoothness orders $\bar{k}$ and $k$ refer separately to the outcome regression $\bar{Q}_0$ and the target function $\psi_0$. For simplicity one might set $k = \bar{k}$, but the target function is often substantially smoother than $\bar{Q}_0$, which may be smooth in $A$ but less so in $W$. In general we therefore have $k \geq \bar{k}$.

Let $L(\bar{Q})$ be the squared error loss if $Y$ is continuous, or the binary log-likelihood loss if $Y$ is binary, so that $\bar{Q}_0 = \arg\min_{\bar{Q} \in D^{(\bar{k})}({\cal R}(\bar{d}))} P_0 L(\bar{Q})$. 
Following the standard HAL-MLE construction, we select a finite rich subset ${\cal R}_N \subset {\cal R}(\bar{d})$ such that the approximation error of $D^{(\bar{k})}({\cal R}_N)$ relative to $D^{(\bar{k})}({\cal R}(\bar{d}))$ is negligible. 
The HAL-MLE is then defined by
\[
\beta_n = \arg\min_{\beta:\, \|\beta\|_1 \leq C_n} P_n L\Big( \sum_{j \in {\cal R}_N} \beta(j) \bar{\phi}_j \Big), 
\qquad \bar{Q}_n = \sum_{j \in {\cal R}_N} \beta_n(j) \bar{\phi}_j,
\]
where $C_n$ is an $L_1$-norm bound on the coefficient vector.
This yields the plug-in estimator $\psi_n = Q_{W,n} \Phi(\bar{Q}_n)$ of $\Psi_{{\cal R}(\bar{d})}(P_0)$. 
We write ${\cal R}_n = \{j \in {\cal R}_N : \beta_n(j) \neq 0\}$ for the active set, so that $\bar{Q}_n \in D^{(\bar{k})}_{C_n}({\cal R}_n)$ as well.

We now analyse the plug-in HAL-MLE $\Psi^F(Q_n)$ of the DRC $\Psi^F(Q_0)$, drawing on the pointwise asymptotic normality results for the HAL-MLE established in \citep{van2023higher}. 
Decomposing the estimation error gives
\begin{eqnarray*}
\psi_n - \psi_0 &=& (Q_{W,n} - Q_{W,0})\Phi(\bar{Q}_0) + Q_{W,0}\left\{\Phi(\bar{Q}_n) - \Phi(\bar{Q}_0)\right\} \\
&& + (Q_{W,n} - Q_{W,0})\left\{\Phi(\bar{Q}_n) - \Phi(\bar{Q}_0)\right\}.
\end{eqnarray*}
The first term is $O_P(n^{-1/2})$ and the third is $o_P(n^{-1/2})$ because $D^{(\bar{k})}_{C_n}([0,1]^{\bar{d}})$ is a Donsker class, assuming $C_n = O_P(1)$ or grows sufficiently slowly with sample size. 
The first-order behaviour of $\psi_n - \psi_0$ is therefore driven entirely by the second term, and we have
\[
\psi_n - \psi_0 = Q_{W,0}\left\{\Phi(\bar{Q}_n) - \Phi(\bar{Q}_0)\right\} 
+ O_P(n^{-1/2}).
\]
Since $\Phi$ is nicely differentiable, the directional derivative of $\Phi$ at $\bar{Q}$ acting on a perturbation $h$ is
\[
\dot{\Phi}_{\bar{Q}}(h)(a, w) = 
\begin{cases}
h(a, w) & \text{if } \Phi(x) = x, \\[2pt]
\dot{Q}(a, w) h(a, w) & \text{if } \Phi(x) = 1/(1 + \exp(-x)),
\end{cases}
\]
where $\dot{Q}(a, w) = \exp(-\bar{Q}(a, w)) / \{1 + \exp(-\bar{Q}(a, w))\}^2$. 
In the binary case the directional derivative is therefore a multiplication operator. 
Linearising around $\bar{Q}_0$ gives
\[
Q_{W,0}\{\Phi(\bar{Q}_n) - \Phi(\bar{Q}_0)\} \approx 
Q_{W,0} \dot{\Phi}_{\bar{Q}_0}(\bar{Q}_n - \bar{Q}_0),
\]
where the remainder behaves as the square of the rate of convergence of $Q_n-Q_0$ and is thus negligible.

The proof of the pointwise asymptotic normality of HAL involves a decomposition into a random part and a bias part, $(\bar{Q}_n - \bar{Q}_{0,n}) + (\bar{Q}_{0,n} - \bar{Q}_0)$, where $\bar{Q}_{0,n} = \arg\min_{\bar{Q} \in D^{(\bar{k})}({\cal R}_{0,n})} P_0 L(\bar{Q})$ is the oracle MLE over an independent sparse working model $D^{(\bar{k})}({\cal R}_{0,n})$ of size $J_n$ with uniform approximation error 
$O^+(1/J_n^{\bar{k}+1})$. 
Applying this decomposition to the linearised expression gives
\[
Q_{W,0}\{\Phi(\bar{Q}_n) - \Phi(\bar{Q}_0)\} \approx 
Q_{W,0} \dot{\Phi}_{\bar{Q}_0}(\bar{Q}_n - \bar{Q}_{0,n}) 
+ Q_{W,0} \dot{\Phi}_{\bar{Q}_0}(\bar{Q}_{0,n} - \bar{Q}_0),
\]
where the second term is the bias. The asymptotic linearity of HAL for $\bar{Q}_n - \bar{Q}_{0,n}$, applied to the linear functional $\bar{Q} \mapsto Q_{W,0} \dot{Q}_0 \bar{Q}$, gives pointwise asymptotic normality of $(n/J_n)^{1/2} Q_{W,0} \dot{Q}_{0,a}
(\bar{Q}_n - \bar{Q}_{0,n})$ at each $a$. 
To obtain inference for $\psi_0$ we also need to control the bias. 
Since $\bar{Q}_0 \in D^{(\bar{k})}({\cal R}(\bar{d}))$ implies $\|\bar{Q}_{0,n} - \bar{Q}_0\|_\infty = O^+(J_n^{-(\bar{k}+1)})$, choosing $J_n$ large enough, which may require undersmoothing relative to the cross-validation selector $C_{n,cv}$, renders the bias negligible relative 
to the standard error $(J_n/n)^{1/2}$. This results in a corresponding rate of convergence $(J_n/n)^{1/2}$ with respect to the target function. Selecting $J_n$ so that the bias and standard error are balanced shows that 
this plug-in HAL-MLE  achieves a rate of convergence that is dimension-free, up to $\log n$ factors,
\[
\|\psi_n - \psi_0\|_\infty = O^+(n^{-k^*/(2k^*+1)}), \quad k^* = k + 1.
\]

The above asymptotic results show that the plug-in HAL-MLE of lower-dimensional targets has excellent statistical properties given appropriate undersmoothing, and the same analysis extends to any target function estimation problem. 
However, the plug-in HAL-MLE has two important limitations.
First, in finite samples the bias term depends in an important way on the choice of submodel $D^{(\bar{k})}({\cal R}_N)$, and specifically on the bias of $\bar{Q}_{0,n} - \bar{Q}_0$ on the data-adaptively selected model $D^{(\bar{k})}({\cal R}_n)$. This arises because the fit $\bar{Q}_n$ is aimed at the high-dimensional function $\bar{Q}_0$ and is not targeted towards the lower-dimensional curve $\psi_0$.
For a fixed $L_1$-norm the HAL-MLE is concerned with fitting the whole function $\bar{Q}_{0}$ and not only the relevant features $\Psi(Q_{W,0},\bar{Q}_0)$.
Appropriate selection of ${\cal R}_N$ and undersmoothing of the HAL-MLE may be needed to ensure that enough basis functions involving $A$ are included for the marginalised $\Psi(Q_{W,n}, \bar{Q}_n)$ to be a flexible enough function of $A$. 
However, global undersmoothing of $\bar{Q}_n$ is an inefficient way to achieve the desired fit, in just the same way that global undersmoothing is an inefficient route to estimating a pathwise differentiable feature of $\Psi(P_0)$. 
It is this very inefficiency that motivated TMLE in the pathwise differentiable setting as a principled approach to undersmoothing with minimal overfitting. 
The same logic motivates a TMLE-style targeted undersmoothing for the DRC: starting from a well-tuned, non-undersmoothed $\bar{Q}_n$, the targeting step can economically tailor the fit towards the marginalised target.

Second, the plug-in HAL-MLE relies on $\bar{Q}_0 \in D^{(\bar{k})}({\cal R}(\bar{d}))$ to control bias at a rate determined by the smoothness $k$ of the outcome regression $\bar{Q}_0$. 
If we only care about the target function $\Phi(\bar{Q}_0)$, however, the bias control should depend only on the smoothness of $Q_{W,0} \Phi(\bar{Q}_0)$, not on the global smoothness of $\bar{Q}_0$. 
Our proposed TMLE achieves exactly this: it controls bias by assuming $\Psi^F(Q_0) \in D^{(k)}([0,1])$ with $\bar{k} \le k$, avoiding any additional smoothness assumptions on $\bar{Q}_0$ in $W$. 
We still require second-order terms in $\bar{Q}_n - \bar{Q}_0$ to converge to zero fast enough, but this can be arranged by assuming, for example, $\bar{Q}_0 \in D^{(0)}([0,1]^{\bar{d}})$ together with $\Psi^F(Q_0) \in D^{(k)}([0,1])$ for some $k \geq 1$. 
This also shows that, if $\Psi^F(Q_0) \in D^{(k)}([0,1]^{d})$ while $\bar{Q}_0$ is non-smooth in $W$ (so that $\bar{Q}_0 \in D^{(0)}([0,1]^{\bar{d}})$ but not $D^{(1)}([0,1]^{\bar{d}})$), the plug-in HAL-MLE achieves only the rate 
$O^+(n^{-1/3})$ while the proposed TMLE attains the minimax rate $O^+(n^{-k^*/(2k^*+1)})$.

\section{Sieve TMLE of projection of DRC on spline working model in target function space} \label{Chtargetfunction3}

The plug-in HAL-MLE presented above is fully nonparametric but its bias is determined
by the global smoothness of $\bar{Q}_0$, not the smoothness of the lower dimensional target function $\Psi^F(Q_0)$.
This motivates us to first consider a working MSM that will directly control the bias for the target function and then use TMLE to estimate the unknown coefficients in this working MSM. It addresses the lack of targeting of the plug-in HAL-MLE and provides theoretical strong results, at least as strong as the plug-in HAL-MLE. 

The first step is to define an $n(\lambda)$-dimensional parametric working model $D^{(k)}({\cal R}_{\lambda})\subset D^{(k)}([0,1]^{d})$ in terms of linear combination of $n(\lambda)$ $k$-th order splines for the target function, where $\lambda$ is a tuning parameter that controls the size $n(\lambda)$ of the model and its uniform approximation error of the total parameter space $D^{(k)}([0,1]^{d})$ for the target function. Due to our uniform approximation results for spline working models, we know that this $n(\lambda)$-dimensional working model can be chosen to uniformly approximate any target function in $D^{(k)}([0,1]^{d})$ at rate $O^+(n(\lambda)^{-(k+1)})$, thereby controlling the bias \citep{van2023higher}.

As a second step we need to decide how to define the coefficients in this working model as a parameter of the data distribution $P$, which then yields a definition of an approximation $\Psi_{\lambda}:{\cal M}\rightarrow D^{(k)}([0,1]^{d})$ of $\Psi:{\cal M}\rightarrow D^{(k)}([0,1]^{d})$ so that $\Psi_{\lambda}(P_0)$ is the approximation of target function $\Psi(P_0)$.
We define the coefficient vector in terms of a projection of the target function $\psi_0$ onto this working model using a particular distance such as $L^2$-norm. We refer to this as projecting in the target function space.
One could also define a working model ${\cal Q}_w$ for $\bar{Q}_P$ that implies that $\Psi({\cal Q}_w)$ equals or strongly approximates our $n(\lambda)$-dimensional working model $D^{(k)}({\cal R}_{\lambda})$ for the target function $\Psi^F(Q_0)$. The loss-based projection of $\bar{Q}_P$ onto this working model ${\cal Q}_w$ implies then also a projection on the working model $D^{(k)}({\cal R}_{\lambda})$ and thereby defines the coefficient vector as a function of $P$. We refer to this as projecting in model space. In this article we focus on projecting in target function space, but in future work we show that there may be important advantages about projecting in model space.

Either projection makes the coefficient vector $\alpha_{\lambda}(P)$ a pathwise differentiable parameter on the statistical model ${\cal M}$, so that we can apply TMLE $P_{n,\lambda}^*$ of the $n(\lambda)$-dimensional target parameter $\alpha_{\lambda}(P_0)$ solving the empirical mean of the $n(\lambda)$-dimensional canonical gradient $P_n D^*_{\alpha_{\lambda}(),P_{n,\lambda}^*}=0$. This results then in a TMLE $P_{n,\lambda}^*$ and corresponding plug-in estimator $\Psi(P_{n,\lambda}^*)$ of $\Psi_{\lambda}(P_0)$. Typically, $P_{n,\lambda}^*$ represents both the TMLE $Q_{n,\lambda}^*$ targeting $\Psi^F(Q_0)$ and a possible nuisance parameter $g_n$ the efficient influence curve $D^*_{\alpha_{\lambda}(),P}=D^*_{\alpha_{\lambda}(),Q_P,g_P}$ depends upon. We will then follow a standard TMLE analysis to prove asymptotic linearity of $\Psi(P_{n,\lambda}^*)-\Psi_{\lambda}(P_0)$ and thereby asymptotic normality at rate $(n(\lambda)/n)^{1/2}$. We will show that the choice of $n(\lambda)$ controls the bias-variance trade-off and can be optimized to obtain the desired optimal rate $n^{-k^*/(2k^*+1)}$ of convergence.

\subsection{Projection definition and its canonical gradient}
We will first demonstrate this approach in the DRC example. The general TMLE is a straightforward generalization of this. 
Given the working model $D^{(k)}({\cal R}_{\lambda})$ for $\psi_0$, we define the weighted $L^2$-projection parameter $\Psi_{\lambda}(P)= \alpha_{\lambda,P}^{\top}{\bf \phi}_\lambda(a)$, where 
\[
\alpha_{\lambda}(P)=\arg\min_{\alpha} \int_a \omega(a)\left(E_P\bar{Q}_P(a,W)-\sum_{j\in {\cal R}_\lambda}\alpha(j)\phi_j(a)\right )^2 da,\]
with weight $\omega(a)$. For example, $\omega(a)$ could be the marginal density of treatment for the creation of stabilized weights.

We now determine the canonical gradient of $\alpha_{\lambda}(P)$.
We have that $\alpha_{\lambda}(P)=\alpha_{\lambda}(Q_W,\bar{Q}_P)$ depends on $P$ through the marginal distribution of $W$ and the outcome regression $\bar{Q}_P=E_P(Y\mid A,W)$.
For the squared error and binary log-likelihood loss, the loss-based tangent space at a $\bar{Q}\in D^{(\bar{k})}({\cal R}(\bar{d}))$ implied by paths $\beta_h(\delta)=\beta+\delta h$, so that
$\bar{Q}_{\beta_h(\delta)}=\sum_{j\in {\cal R}(\bar{d})}\beta(j)\bar{\phi}_j+\delta \sum_{j\in {\cal R}(\bar{d})}h(j)\bar{\phi}_j$ through $\bar{Q}_P$ at $\delta=0$ (i.e., paths through $\bar{Q}_{\beta}$ that stay in the model space $D^{(\bar{k})}({\cal R}(\bar{d}))$) is given by
\[
T_{\cal M}(P)=\{ h(A,W)(Y-E_P(Y\mid A,W)): h\in D^{(\bar{k})}({\cal R}(\bar{d}))\}.\]
\begin{lemma}
\label{lem:canonical gradient}
Consider $\alpha_{\lambda}:{\cal M}\rightarrow\openr^{n(\lambda)}$ defined above. Let $\Sigma_{\phi_\lambda}=\int_a \omega(a) {\bf \phi}_\lambda(a){\bf \phi}_\lambda(a)^{\top}da$.
The canonical gradient of $\alpha_{\lambda}$ at $P$ is given by:
\begin{equation}
\label{canonicalgradient}
   D^*_{\alpha_{\lambda}(),P}=D^*_{\alpha_{\lambda}(),\bar{Q},P}+D^*_{\alpha_{\lambda}(),Q_W,P}, 
\end{equation}
where
\begin{eqnarray*}
D^*_{\alpha_{\lambda}(),Q_W,P}&=&\int_a \Sigma_{\phi_\lambda}^{-1}{\bf \phi}_\lambda(a)(\bar{Q}_P(a,W)-\Psi(P)(a) )\omega(a)da\\
D^*_{\alpha(),\bar{Q},P}&=&\omega(A)\Sigma_{\phi_\lambda}^{-1}{\bf \phi}_\lambda(A)\frac{1}{g(A\mid W)}(Y-\bar{Q}_P(A,W)).
\end{eqnarray*}
{\bf Canonical gradient of $\Psi_{\lambda,a}()$:}
By the delta-method this implies that 
\begin{eqnarray*}
D^*_{\Psi_{\lambda}(),a,P}&=&{\bf \phi}_\lambda^{\top}(a)D^*_{\alpha_{\lambda}(),P}\\
&=&D^*_{\Psi_{\lambda}(),a,\bar{Q},P}+D^*_{\Psi_{\lambda}(),a,Q_W,P}\\
D^*_{\Psi_{\lambda}(),a,\bar{Q},P}&=&\left\{\Sigma_{{\phi}_\lambda}^{-1}{\bf \phi}_\lambda(A)\right\}^{\top}{\bf \phi}_\lambda(a) \frac{\omega(A)}{g(A\mid W)}(Y-\bar{Q}_P(A,W))\\
D^*_{\Psi_{\lambda}(),a,Q_W,P}&=&{\bf \phi}_\lambda^{\top}(a) \int_{a'} \Sigma_{\phi_\lambda}^{-1}{\bf \phi}_\lambda(a')(\bar{Q}_P(a',W)-\Psi(P)(a') )h(a')da'.
\end{eqnarray*}
We could select an orthonormal basis ${\bf \phi}_\lambda^*$ w.r.t. inner product $\langle h_1,h_2\rangle =\int_a \omega(a) h_1(a)h_2(a)da $, in which case
$\Sigma_{\phi_\lambda^*}$ is the identity matrix so that we obtain
\begin{eqnarray*}
D^*_{\Psi_{\lambda}(),a,P}&=&\sum_{j\in {\cal R}_{\lambda}}\phi_j^*(a)\phi_j^*(A) \frac{\omega(A)}{g(A\mid W)}(Y-\bar{Q}_P(A,W))\\
&&+\sum_{j\in {\cal R}_{\lambda}}\phi_j^*(a)\int_{a'}\phi_j^*(a')(\bar{Q}_P(a',W)-\Psi(P)(a'))h(a')da'.
\end{eqnarray*}
\end{lemma}
{\bf Proof:}
To determine the pathwise derivative and canonical gradient we need to use the implicit function theorem. 
We have that $\alpha_{\lambda}(P)$ solves
\[
 0=\int_a \omega(a){\bf \phi}_\lambda(a)(E_P\bar{Q}_P(a,W)-\alpha^{\top}{\bf \phi}_\lambda(a) )da.\]
 This is an equation of the form $U(\alpha(P),\Psi(P))=0$, where our notation suppresses the dependence on $\lambda$. 
 To understand the derivative of $\alpha(P)$ as a functional of $P$, consider $\psi_n=E_n\bar{Q}_n$ and $\Psi(P)=E_P \bar{Q}_P$ and by definition of $\alpha_n=\alpha(\bar{Q}_n))$ and $\alpha(P)=\alpha(\bar{Q}_P)$, we have $U(\alpha_n,\psi_n)=U(\alpha(P),\Psi(P))=0$.
 Thus, 
 \[
 \alpha_n-\alpha(P)\approx \Sigma_{\phi_1}^{-1}\int_a \omega(a){\bf \phi}_\lambda(a)(\psi_n-\Psi(P))(a)da,\]
 where $\Sigma_{\phi_1}=\int_a \omega(a) {\bf \phi}_{1}(a){\bf \phi}_{1}(a)^{\top}da$.
 Moreover, \[
 \psi_n-\Psi(P)=E_n\bar{Q}_n-E_P\bar{Q}_P\approx (Q_{W,n}-Q_{W})\bar{Q}_P(a,\cdot)+Q_W (\bar{Q}_n-\bar{Q}_P).\]
 Thus,
 \[
 \begin{array}{l}
 \alpha_n-\alpha(P)\approx \Sigma_{\phi_1}^{-1}\int_a \omega(a){\bf \phi}_\lambda(a)Q_W(\bar{Q}_n-\bar{Q}_P)(a,\cdot)da\\
 +
 \Sigma_{\phi_1}^{-1}(Q_{W,n}-Q_{W})\int_a \omega(a){\bf \phi}_\lambda(a)\bar{Q}_P(a,\cdot)da .\end{array}
 \]
Viewing $\alpha_n=\alpha(P_{\epsilon})$, the above demonstrates 
that each of the two terms contribute a component of the canonical gradient $D^*_{\alpha_{\lambda}(),P}$ of the pathwise derivative $\frac{\partial}{\partial\epsilon} \alpha_{\lambda}(P_{\epsilon})$. The second term contributes the component due to the dependence on $Q_W$:
\[
D^*_{\alpha_{\lambda}(),Q_W,P}=\int_a \Sigma_{\phi_\lambda}^{-1}{\bf \phi}_\lambda(a)(\bar{Q}_P(a,W)-\Psi(P)(a) )\omega(a)da.
\]
 So it remains to obtain the canonical gradient of $P\rightarrow E_W\int_a \omega(a){\bf \phi}_\lambda(a)\bar{Q}_P(a,W)da$ treating $E_W$ as fixed.
 In a nonparametric model, and acting as if $W,A$ is discrete, the canonical gradient of $\bar{Q}_P(a,w)$ for a given $(a,w)$ is given by 
 $D_{\bar{Q}(),(a,w),P}=I(A=a,W=w)/(g(a\mid w)p(w))(Y-\bar{Q}_P(A,W))$.
 By the delta-method this yields the following canonical gradient
 \[
 \begin{array}{l}
\int_w p(w) \int_a \omega(a){\bf \phi}_\lambda(a) \frac{I(A=a,W=w)}{g(a\mid w)p(w)}(Y-\bar{Q}_P(A,W)) da\\
=\frac{\omega(A){\bf \phi}_\lambda(A)}{g(A\mid W)}(Y-\bar{Q}_P(A,W)).\end{array}
\]
Thus, the canonical gradient contribution from the dependence of $\alpha_{\lambda}$ on $\bar{Q}_P$ is given by:
\[
D^*_{\alpha(),\bar{Q},P}=\omega(A)\Sigma_{\phi_\lambda}^{-1}{\bf \phi}_\lambda(A)\frac{1}{g(A\mid W)}(Y-\bar{Q}_P(A,W)).\]
This is also the canonical gradient if $(W,A)$ is not discrete, which follows by using a finer and finer discretization and applying the above derivation.
$\Box$


\subsection{Representation of canonical gradient that proves its variance is $O(n(\lambda))$}
\begin{lemma}\label{lemma1}
Define $D^*_{\phi,P}\equiv \phi(A)\omega(A)/g(A\mid W)(Y-\bar{Q}_P(A,W))+\int_{a'}\phi(a')(\bar{Q}_P(a',W)-\Psi(P)(a'))h(a')da'$.
Define the following inner product on $D^{(k)}([0,1])$:
\[
\langle \phi_1,\phi_2 \rangle_{\psi,P}\equiv 
E_P D^*_{\phi_1,P}D^*_{\phi_2,P}.\]
Let $\phi_j^*$, $j\in {\cal R}_{\lambda}$, be an orthonormal basis of $D^{(k)}({\cal R}_{\lambda})$ w.r.t. this inner product.
Then, \[
 D^*_{\Psi_{\lambda}(),a,P}=\sum_{j\in {\cal R}_{\lambda}}\phi_j^*(a)D^*_{\phi_j^*,P}(O),\]
 which is now a sum of uncorrelated random variables.
 The variance is then represented as:
 \[
 \mbox{VAR}_P D^*_{\Psi_{\lambda}(),a,P}=\sum_{j\in {\cal R}_{\lambda}}\{\phi_j^*(a)\}^2.
 \]
 This inner product can be represented as 
\[
\begin{array}{l}
\langle \phi_1,\phi_2\rangle_{\psi,P}= \int a \phi_1\phi_2(a) d\mu_P(a)
+\int_{a_1,a_2}\phi_1(a_1)\phi_2(a_2)\rho_P(a_1,a_2) da_1 da_2,
\end{array}
\]
where
\begin{eqnarray*}
 d\mu_P(a)&=&
 \frac{h^2(a)}{g(a\mid W)}\sigma^2(a,W) da\\
\rho_P(a_1,a_2)&\equiv& E_W(\bar{Q}_P(a_1,W)-\Psi(P)(a_1))(\bar{Q}_P(a_2,W)-\Psi(P)(a_2)) .
\end{eqnarray*}
\end{lemma}
\begin{proof}
This inner product can be worked out as a more regular looking inner product as follows: 
\[
\begin{array}{l}
\langle \phi_1,\phi_2\rangle_{\psi}\equiv 
E_P \phi_1(A)\phi_2(A) \frac{h^2(A)}{g^2(A\mid W)}(Y-\bar{Q}_P(A,W))^2\\
+E_P\int_{a_1}\phi_1(a_1)(\bar{Q}_P(a_1,W)-\Psi(P)(a_1))h(a_1)da_1\int_{a_2}\phi_2(a_2)(\bar{Q}_P(a_2,W)-\Psi(P)(a_2))h(a_2)da_2\\
=E_W\int_a \phi_1\phi_2 \frac{h^2(a)}{g(a\mid W)}\sigma^2(a,W) da\\
+\int_{a_1,a_2}\phi_1(a_1)\phi_2(a_2) \left\{ E_W(\bar{Q}_P(a_1,W)-\Psi(P)(a_1))(\bar{Q}_P(a_2,W)-\Psi(P)(a_2)) \right\} da_1 da_2\\
\equiv \int_a \phi_1 \phi_2 h^2(a) E_W\frac{\sigma^2(a,W)}{g(a\mid W)} da
+\int_{a_1,a_2}\phi_1(a_1)\phi_2(a_2)\rho_P(a_1,a_2) da_1 da_2\\
\\
\equiv \int a \phi_1\phi_2(a) d\mu_P(a)
+\int_{a_1,a_2}\phi_1(a_1)\phi_2(a_2)\rho_P(a_1,a_2) da_1 da_2,
\end{array}
\]
where $d\mu_P(a)$ is a positive measure, and $\rho_P(a_1,a_2)$ is defined as the covariance of $\bar{Q}_P(a_1,W)$ and $\bar{Q}_P(a_2,W)$.
This inner product is indeed a bi-linear mapping and $\langle \phi,\phi\rangle_{\psi}=0$ implies $\phi=0$, making it a valid inner product.
\end{proof}

Even without using this orthonormal basis
we could define $\tilde{\phi}_\lambda\equiv \Sigma_{\phi_\lambda}^{-1}{\bf \phi}_\lambda$, giving us a new basis $\tilde{\phi}_{\lambda,j}$, $j\in {\cal R}_\lambda$. We can then still represent
\[
D^*_{\Psi_{\lambda}(),a,P}=\sum_{j\in {\cal R}_{\lambda}}\tilde{\phi}_j(a) D^*_{\tilde{\phi}_j,P}(O),\]
as a linear combination of $D^*_{\tilde{\phi}_j,P}$ with coefficients $\tilde{\phi}_j(a)$, $j\in {\cal R}_\lambda$. 
But in that case, it would not be a sum of uncorrelated random variables. 
$\Psi_{\lambda}(P)(a)$ is invariant w.r.t. any linear transformation of the basis functions so that this representation applies to any original basis choice and corresponding $\alpha_{\lambda}(P)$. The utility of this representation is that it shows that the variance of this canonical gradient grows with $n(\lambda)$ as $n(\lambda)$.

One can obtain a similar representation of the variance of the canonical gradient as a sum of squares of the orthonormal basis function for other inner products. 
\begin{lemma}\label{lem:canonicalvariance}
We have that $\Psi(P_{\lambda})(x)=\sum_{j\in {\cal R}_\lambda}\alpha_{\lambda,P}(j)\phi_j^*(x)$ so that, by Lemma \ref{lem:canonical gradient}, its canonical gradient is $D^*_{\Psi_{\lambda}(),x,P}=\sum_{j\in {\cal R}_\lambda}D_{\alpha_{\lambda,j}(),P}\phi_j^*(x)$.
This is a linear combination of the $n(\lambda)$ canonical gradients $D_{\alpha_{\lambda,j}(),P}=D^*_{\phi_j^*,P}$ of $\alpha_{\lambda,P}(j)$.
Let ${\bf \phi}_\lambda^*(x)=(\phi_j^*(x): j\in {\cal R}_\lambda)$ be this vector of coefficients.
Then, 
\[
\mbox{VAR}D^*_{\Psi_{\lambda}(),x,P}={\bf \phi}_\lambda^{*,\top}\Sigma_{\lambda}{\bf \phi}_\lambda^*,\]
where $\Sigma_{\lambda}$ is the covariance matrix of $n(\lambda)$-dimensional $D_{\alpha_{\lambda}(),P}$.
We have an eigenvalue decomposition of $\Sigma_{\lambda}=S_{\lambda}D_{\lambda}S^{\top}_{\lambda}$, where $D_{\lambda}$ is a diagonal matrix with its eigenvalues $\gamma(j)$, $j\in {\cal R}_\lambda$. We can then define $\Sigma_{\lambda}^{1/2}= S_{\lambda}D^{1/2}_{\lambda}$. Then it follows that this variance is given by 
\[
\sigma^2_{\lambda}(x)n(\lambda)= \sum_{j\in {\cal R}_\lambda} \{\phi_j^*(x)\}^2\gamma(j).\]
If we assume that the maximal eigenvalue $\max_j \gamma(j)=O(1)$, then it follows that $\sigma^2_{\lambda}(x)n(\lambda)\sim 
\sum_{j\in {\cal R}_\lambda}\{\phi_j^*(x)\}^2$.
\end{lemma}

Either way, we have
\[
\mbox{VAR}_P D^*_{\Psi_{\lambda}(),a,P}=\mbox{VAR}_P D^*_{\Psi_{\lambda}(),a,\bar{Q},P}+\mbox{VAR}_P D^*_{\Psi_{\lambda}(),a,Q_W,P},\]
and for variance estimation one would generally just use this direct representation. 


Informally, this shows that the variance of a TMLE of the projection curve $\Psi_{\lambda}(P_0)$ in $D^{(k)}({\cal R}_{\lambda})$ will behave as $O(n(\lambda)/n)$ corresponding with a rate of convergence $(n(\lambda)/n)^{1/2}$. Combined with the bias being $O((n(\lambda)^{-(k+1)})$, this yields the desired rate of convergence for the TMLE w.r.t. the true curve $\psi_0$.

\subsection{Defining the TMLE of pathwise differentiable approximation $\Psi_{\lambda}(P_0)$ of DRC}
Let $\bar{Q}_n$ be an initial estimator of $\bar{Q}_0$, and let $g_n$ be an estimator of $g_0$. 
If we assume $\bar{Q}_0\in D^{(\bar{k})}([0,1]^{\bar{d}})$ and $g_0\in D^{(m)}([0,1]^{\bar{d}})$, then we can use $\bar{k}$-th order and $m$-th order HAL-MLEs. 
For example, for $g_0(a\mid W)$ we might use the link function $\lambda(a\mid W)=\exp(\sum_{j\in {\cal R}_1(d)}\alpha(j)\bar{\phi}_j)$ with $\lambda(a\mid W)=g(a\mid W)/\bar{G}(a\mid W)$, and $\bar{G}(a\mid W)=P(A>a\mid W)$, and use glmnet Poisson to obtain an HAL-MLE $g_n$. Consider a least favorable path through $\bar{Q}_n$ with $n(\lambda)$-dimensional clever covariate \[
C_{g_n,\lambda}(A,W)=\omega(A)\Sigma_{\phi_\lambda}^{-1}{\bf \phi}_\lambda(A)\frac{1}{g_n(A\mid W)}.\]
We can define the fluctuation model $\bar{Q}_{n,\epsilon}=\bar{Q}_n+\epsilon C_{g_n,\lambda}$ and compute the least squares estimator $\epsilon_n$, resulting in the TMLE update $\bar{Q}_{n,\lambda}^*$ that solves $P_n D^*_{\alpha_{\lambda}(),\bar{Q}_{n,\lambda}^*,g_n}=0$.
Alternatively, we can use $C_{\lambda}(A,W)=\omega(A)\Sigma_{\phi_\lambda}^{-1}{\bf \phi}_\lambda(A)$ as the clever covariate and define the fluctuation model $\bar{Q}_{n,\epsilon}=\bar{Q}_n+\epsilon C_{\lambda}$.
We can then compute the weighted least squares estimator using weights $1/g_n(A\mid W)$, resulting in a TMLE $\bar{Q}_{n,\lambda}^*$.
Either one of these two TMLEs solves $P_n D^*_{\alpha_{\lambda}(),\bar{Q}_{n,\lambda}^*,g_n}=0$. 

We use short-hand notation $P_{n,\lambda}^*=(\bar{Q}_{n,\lambda}^*,g_{n,\lambda})$, and denote the efficient influence curve score equation as $P_n D^*_{\alpha_{\lambda}(),P_{n,\lambda}^*}=0$. 
The TMLE $\bar{Q}_{n,\lambda}^*$ maps into an estimator of the coefficients given by $\alpha_{\lambda}(P_{n,\lambda}^*)$:
\[
\alpha_{\lambda}(P_{n,\lambda}^*)=\arg\min_{\alpha} \int_a \omega(a)(Q_{W,n}\bar{Q}_{n,\lambda}^*(a,W)-\alpha^{\top}\phi^*_\lambda(a))^2da.
\]
 
The plug-in TMLE is then given by $\Psi_{\lambda}(P_{n,\lambda}^*)(a)=\alpha_{\lambda}(P_n^*)^{\top}{\bf \phi}_\lambda^*(a)$ and in our analysis of this TMLE we can use that for all $a$,
$P_n D^*_{\Psi_{\lambda}(),a,P_{n,\lambda}^*}=0$.

\subsection{Asymptotic normality theorem for sieve TMLE of DRC}
We now provide the following theorem for the asymptotic normality of the TMLE $\Psi_{\lambda}(P_{n,\lambda}^*)(a)$ of the DRC $\psi_0$, based on a fixed, non-data dependent, sieve of working models $D^{(k)}({\cal R}_{\lambda})\subset D^{(k)}([0,1])$.

\begin{theorem}\label{chtargetfunctiontmleI}
Let $O=(W,A,Y)\sim P_0\in {\cal M}$ with ${\cal M}$ nonparametric beyond a possible model for the conditional density of the continuous $A$ given $W$, and let the target function be the confounder-adjusted DRC $\Psi(P_0)(a)=E_{P_0}\bar{Q}_0(a,W)$, $\bar{Q}_P=E_P(Y\mid A,W)$. Assume $\psi_0\in D^{(k)}([0,1])$, $\bar{Q}_0\in D^{(\bar{k})}([0,1]^{\bar{d}})$, $g_0\in D^{(m)}([0,1]^{\bar{d}})$ with integers $k,\bar{k},m\geq 1$, set $k^*=k+1$, $\bar{k}^*=\bar{k}+1$, $m^*=m+1$. We also assume positivity $\min_{a\in [0,1]}g_0(a\mid W)>\delta>0$ $P_0$-a.e., and the smoothness condition relating the nuisance rates to the working-model dimension,
\[
 n^{-\bar{k}^*/(2\bar{k}^*+1)}n^{-m^*/(2m^*+1)} = o^-( n^{-(k^*+1)/(2k^*+1)}).\]
Let the $n(\lambda)$ basis functions be chosen so that our uniform approximation error applies:
 \[
 \sup_{\psi\in D^{(k)}_M([0,1])}\inf_{\alpha}\left\lvert{ \sum_{j\in {\cal R}_\lambda}\alpha(j)\phi_j^*-\psi}\right\rvert_{\infty}=O^+(n(\lambda)^{-k^*}).\]

Let $\Psi_{\lambda}(P_{n,\lambda}^*)(a)=\alpha_{\lambda}(P_{n,\lambda}^*)^{\top}{\bf \phi}_\lambda^*(a)$ be the plug-in TMLE defined above, constructed with an orthonormal basis $\phi_\lambda^*(a)$ so that $\Sigma_{\phi^*_\lambda}^{-1}$ is the identity matrix. From Lemma~\ref{lem:canonical gradient}, let $ D^*_{\Psi_{\lambda}(),a,P_{n,\lambda}^*}$ denote the gradient of $\Psi_{\lambda,a}$ evaluated at $P^*_{n,\lambda}$.

{\bf (i) Asymptotic normality with respect to the projection $\Psi_{\lambda}$.}
Uniformly in $a$,
\[
(n/n(\lambda))^{1/2}(\Psi_{\lambda}(P_{n,\lambda}^*)-\Psi_{\lambda}(P_0))(a)=
(n/n(\lambda))^{1/2}(P_n-P_0)D^*_{\Psi_{\lambda}(),a,P_0}+r_n(a),\]
with $\sup_a\mid r_n(a)\mid=o_P(1)$. Consequently, for each $a$,
\[
\sigma_{\lambda,P_0}^{-1}(a)(n/n(\lambda))^{1/2}(\Psi_{\lambda}(P_{n,\lambda}^*)-\Psi_{\lambda}(P_0))(a)\Rightarrow_d N(0,1),\]
with $\sigma^2_{\lambda,P_0}(a)=P_0\{D^*_{\Psi_{\lambda}(),a,P_0}\}^2$, and empirical process theory over the one-dimensional class $\{D^*_{\Psi_{\lambda}(),a,P_0}:a\}$ yields $0.95$-simultaneous confidence bands of width $O(\log n\,(n(\lambda)/n)^{1/2})$.

{\bf (ii) Asymptotic normality with respect to the target function $\psi_0$.}
By balancing the approximation bias $O^+(n(\lambda)^{-k^*})$ against the standard error by undersmoothing, i.e.\ taking $n(\lambda)\sim n^{1/(2k^*+1)}$ up to $\log n$-factors with $c$ large enough that $n(\lambda)^{-k^*}(\log n)^{c}=o((n(\lambda)/n)^{1/2})$, the same result holds with respect to the true target function $\psi_0=\Psi(P_0)$: uniformly in $a$,
\[
(n/n(\lambda))^{1/2}(\Psi_{\lambda}(P_{n,\lambda}^*)-\Psi(P_0))(a)=
(n/n(\lambda))^{1/2}(P_n-P_0)D^*_{\Psi_{\lambda}(),a,P_0}+r_n(a),\]
with $\sup_a\mid r_n(a)\mid=o_P(1)$. Consequently, for each $a$,
\[
\sigma_{\lambda,P_0}^{-1}(a)(n/n(\lambda))^{1/2}(\Psi_{\lambda}(P_{n,\lambda}^*)-\Psi(P_0))(a)\Rightarrow_d N(0,1),\]
and the same one-dimensional class $\{D^*_{\Psi_{\lambda}(),a,P_0}:a\}$ yields $0.95$-simultaneous confidence bands for $\psi_0$ of width $O(\log n\,(n(\lambda)/n)^{1/2})$.

\end{theorem}

\begin{proof}
By the uniform approximation error of the working model, we have for any $P\in {\cal M}$,
 \[
 \left\lvert \alpha_{\lambda}^{\top}(P){\bf \phi}_\lambda^*-\Psi(P)\right\rvert_{\infty}=O(n(\lambda)^{-k^*}).\]

 By Lemma \ref{lem:canonical gradient}, the parameter $\Psi_{\lambda}$ is pathwise differentiable at any $P$ with canonical gradient given by:
 \[
D^*_{\alpha_{\lambda}(),P}=D^*_{\alpha_{\lambda}(),\bar{Q},P}+D^*_{\alpha_{\lambda}(),Q_W,P},\]
where, in the orthonormal basis ${\bf \phi}_\lambda^*$,
\begin{eqnarray*}
D^*_{\alpha_{\lambda}(),Q_W,P}&=&\int_a {\bf \phi}_\lambda^*(a)\bar{Q}_P(a,W)\omega(a)da\\
D^*_{\alpha_{\lambda}(),\bar{Q},P}&=&\omega(A){\bf \phi}_\lambda^*(A)\frac{1}{g(A\mid W)}(Y-\bar{Q}_P(A,W)).
 \end{eqnarray*}
This implies the canonical gradient of $P\rightarrow \Psi_{\lambda}(P)(a)$:
 \[
 D^*_{\Psi_{\lambda}(),a,P}=D^{*,\top}_{\alpha_{\lambda}(),P}{\bf \phi}_\lambda^*(a).\]
 
By Lemma \ref{lemma1}, under a weak regularity condition, which we assume to hold, we have
\[
\sigma^2_{\lambda,P}(a)\equiv P\{ D^*_{\Psi_{\lambda}(),a,P}\}^2=O(n(\lambda)).\]
We define the exact remainder $R_{\Psi_{\lambda},a}(P,P_0)\equiv \Psi_{\lambda}(P)-\Psi_{\lambda}(P_0)+P_0 D^*_{\Psi_{\lambda}(),a,P}$, which satisfies
 $R_{\Psi_{\lambda},a}(P,P_0)=R_{\alpha_{\lambda}}(P,P_0)^{\top}{\bf \phi}_\lambda^*(a)$, where 
$R_{\alpha_{\lambda}}(P,P_0)=\alpha_{\lambda}(P)-\alpha_{\lambda}(P_0)+P_0 D^*_{\alpha_{\lambda}(),P}$ is the exact remainder for $\alpha_{\lambda}:{\cal M}\rightarrow\openr^{n(\lambda)}$. 

For all $a$ we have the exact expansion:
\[
\alpha_{\lambda,P_{n,\lambda}^*}^{\top}{\bf \phi}_\lambda^*(a)-\alpha_{\lambda,P_0}^{\top}{\bf \phi}_\lambda^*(a)=
(P_n-P_0)D^*_{\Psi_{\lambda},a,P_{n,\lambda}^*}+R_{\Psi_{\lambda},a}(P_{n,\lambda}^*,P_0).\]
Here the $P_n D^*_{\Psi_{\lambda}(),a,P_{n,\lambda}^*}$ term is absent since the TMLE step solves $P_n D^*_{\alpha_{\lambda}(),\bar{Q},P_{n,\lambda}^*,g_n}=0$ and the $L^2$-projection's normal equations give $P_n D^*_{\alpha_{\lambda}(),Q_W,P_{n,\lambda}^*}=0$, so $P_n D^*_{\Psi_{\lambda}(),a,P_{n,\lambda}^*}=0$ for all $a$.

This second-order remainder is double robust, with $R_{\Psi_{\lambda},a}(P,P_0)=O(\pl g-g_0\pl_{\infty}\pl \bar{Q}_P-\bar{Q}_0\pl_{\infty} n(\lambda))$. Therefore $R_{\Psi_{\lambda},a}(P_{n,\lambda}^*,P_0)=O_P^+( n^{-\bar{k}^*/(2\bar{k}^*+1)}n^{-m^*/(2m^*+1)} n(\lambda))$. Selecting $n(\lambda)\sim^+ n^{1/(2k^*+1)}$, this equals $O^+( n^{-\bar{k}^*/(2\bar{k}^*+1)}n^{-m^*/(2m^*+1)} n^{1/(2k^*+1)})$, which by the assumed smoothness condition $n^{-\bar{k}^*/(2\bar{k}^*+1)}n^{-m^*/(2m^*+1)} = o^-( n^{-(k^*+1)/(2k^*+1)})$ is $o_P((n(\lambda)/n)^{1/2})$, hence negligible relative to the rate of convergence $(n(\lambda)/n)^{1/2}$.

Consider next the empirical process term \[
R_{n,a}\equiv (P_n-P_0)\{D^*_{\Psi_{\lambda},a,P_{n,\lambda}^*}-D^*_{\Psi_{\lambda},a,P_0}\}.\]
If we divide the canonical gradient by $n(\lambda)$ it follows that it falls in a Donsker class implied by $\bar{Q}_{n,\lambda}^*\in D^{(\bar{k})}_M([0,1]^{\bar{d}})$ and $g_n\in D^{(m)}_M([0,1]^{\bar{d}})$, so that the entropy integral $J_{\infty}(\delta) $ of this class is implied by the worst of the two entropy integrals of $D^{(\bar{k})}([0,1]^{\bar{d}})$ and $D^{(m)}([0,1]^{\bar{d}})$, which we understand. In addition, once we divide by $n(\lambda)$ it also follows that the rate of convergence in $L^2(P_0)$-norm of the random function $f_n=(D^*_{\Psi_{\lambda},a,P_{n,\lambda}^*}-D^*_{\Psi_{\lambda},a,P_0})/n(\lambda)$
will be faster than the one implied by the worst of the two rates of convergence of $\bar{Q}_{n,\lambda}^*$ and $g_n$, by a factor $n(\lambda)^{-1/2}$ (by Lemma~\ref{lemma1}). Let $\delta_n$ be the resulting rate of $\pl f_n\pl_{P_0}$.
We can then bound this empirical process term by $n(\lambda) n^{-1/2}J_{\infty}(\delta_n)$, where $J_{\infty}(\delta)$ is the entropy integral of $D^{(\min(\bar{k},m))}([0,1]^{\bar{d}})$, assuming the strong positivity assumption $\min_{a\in [0,1]}g_0(a\mid W)>\delta>0$. We have bounded this by $O^+(\delta^{(2s+1)/(2s+2)})$ with $s=\min(\bar{k},m)$. We know that $\min(\bar{k},m)\geq 1$ so that we can bound it at minimal by $\delta^{3/4}$. For example, if the slowest rate of the nuisance estimators is $n^{-2/5}$ (i.e., $\min(\bar{k},m)=1$), then $\delta_n=n^{-2/5}n(\lambda)^{-1/2}$, giving the bound $n(\lambda)n^{-1/2}(n^{-2/5}n(\lambda)^{-1/2})^{3/4}=n(\lambda)^{5/8}n^{-4/5}$. If $n(\lambda)\sim n^{1/5}$ as it would be for $\bar{k}=1$, then this bound is $n^{-27/40}$ which is $o(n^{-2/5})=o((n(\lambda)/n)^{1/2})$ as desired. It follows that for our chosen $n(\lambda)$ we obtain a rate of convergence $R_{n,a}=o_P((n(\lambda)/n)^{1/2})$. 

Combining the two bounds with the exact expansion gives
\[
(n/n(\lambda))^{1/2}(\Psi_{\lambda}(P_{n,\lambda}^*)-\Psi_{\lambda}(P_0))(a)=
(n/n(\lambda))^{1/2}(P_n-P_0) D^*_{\Psi_{\lambda}(),a,P_0}+o_P(1).\]
The leading term is a normalized sum of independent mean-zero random variables with bounded variance, so the central limit theorem (CLT) applies and the right-hand side divided by $\sigma_{\lambda,P_0}(a)$ converges to $N(0,1)$. We note that this argument did not rely on cross-fitting, but it did rely on the sieve ${\cal R}_{\lambda}$ being non-random, so that $D^*_{\Psi_{\lambda}(),a,P_0}$ is a fixed function of $O$ and the CLT applies; were ${\cal R}_{\lambda}$ informed by data, this final step would break down.

Under the same conditions controlling the remainders uniformly, we have $\sup_a\mid R_{n,a}\mid=o_P(n(\lambda)/n)$ and  \mbox{$\sup_a \mid R_{\Psi_{\lambda},a}(P_{n,\lambda}^*,P_0)\mid =o_P((n(\lambda)/n)^{1/2})$}, so the asymptotic linearity holds uniformly in $a$ with $\sup_a \mid r_n(a)\mid =o_P(1)$. 

Dividing $D^*_{\Psi_\lambda(), a, P_0}$ by $n(\lambda)$ places it in the same bounded Donsker class identified above for $f_n$. Applying the entropy integral bound for the one-dimensional class $\{D^*_{\Psi_\lambda(), a, P_0}/n(\lambda) : a\}$, whose covering number is $O(1/\varepsilon)$, the supremum over $a$ of $(n/n(\lambda))^{1/2}(P_n - P_0)D^*_{\Psi_\lambda(), a, P_0}$ is bounded by $O_P(\log n \cdot n(\lambda)^{1/2})$, yielding $0.95$-simultaneous confidence bands of width $O(\log n\, (n(\lambda)/n)^{1/2})$.

Finally, selecting $n(\lambda)\sim n^{1/(2k^*+1)}$ up to $\log n$-factors, chosen so that $n(\lambda)^{-k^*}(\log n)^{c}=o((n(\lambda)/n)^{1/2})$ for some large enough power $c$, balances the bias $O^+(n(\lambda)^{-k^*})$ against the standard error $(n(\lambda)/n)^{1/2}$. The approximation bound $\lvert \alpha_{\lambda}^{\top}(P){\bf \phi}_{\lambda}^*-\Psi(P)\rvert_{\infty}=O(n(\lambda)^{-k^*})$ then makes the replacement of $\Psi_{\lambda}(P_0)$ by $\Psi(P_0)$ asymptotically negligible, yielding the stated result with respect to $\psi_0$. This establishes part (ii).
\end{proof}

\section{Sieve TMLE for data adaptive sieve}\label{Chtargetfunction4}

\subsection{Extension of Theorem \ref{chtargetfunctiontmleI} to data adaptive working models}
In the above analysis of the sieve TMLE the $n(\lambda)$-dimensional working models $D^{(k)}({\cal R}_{\lambda})$ were pre-specified, and, as a consequence, the linearization of the TMLE in terms of an empirical mean of an influence curve that is a sum over ${\cal R}_{\lambda}$ of functions of $O$ converges to a normal limit distribution, by the CLT. The only randomness was the possible data adaptive selection of $n(\lambda)$, which is easily handled. Suppose now that we introduce a sieve ${\cal R}_{\lambda_n}$ that is data-dependent, and specifically, was based on outcome data.
The sieve TMLE $P_n^*$ is constructed to solve the score equations of its own data-adaptive working model, not those of this fixed sieve.
Our strategy is nonetheless to analyze $P_n^*$ as though it were a sieve TMLE for the fixed sieve ${\cal R}_{\lambda}$. This requires that $P_n^*$ also solves the efficient influence curve equation of the fixed working model up to a negligible remainder, $P_n D^*_{\Psi_{\lambda}(), a, P_n^*}=o_P((n(\lambda)/n)^{1/2})$, which we establish in Lemma~\ref{lem:sieve rates} below. 
That is, that the data-adaptive sieve approximates the fixed sieve closely enough that the resulting residual is negligible relative to the rate of convergence of the TMLE with respect to the projection curve, which is what allows the asymptotic normality argument derived for the fixed sieve to carry over to $P_n^*$. 

We will build the TMLE from the data-adaptive working model $D^{(k)}({\cal R}_{\lambda,n})$.  Provided the data-adaptive working model $D^{(k)}({\cal R}_{\lambda,n})$ is rich enough to approximate the entire smoothness class $D^{(k)}_M([0,1])$ at the rate of the fixed sieve (the projection-approximation condition of the lemma below), we show that this data-adaptive sieve TMLE also approximately solves the efficient influence curve equation of the fixed working model $D^{(k)}({\cal R}_\lambda)$, up to a remainder that is negligible relative to its rate of convergence, so that the asymptotic normality argument carries over. This also requires that an empirical process term, arising from the decomposition of the score equation, is negligible. The two working models yield two definitions $\alpha_{\lambda,n}(P)$ and $\alpha_\lambda(P)$, respectively, with basis functions $\phi_{\lambda,n}$ and $\phi_\lambda$, respectively. The two corresponding canonical gradients are denoted with $D^*_{\alpha_{\lambda,n}(),P}$ and $D^*_{\alpha_\lambda(),P}$, respectively. They imply canonical gradients of $\Psi_{\lambda,n}(P)=\alpha_{\lambda,n}^{\top}\phi_{\lambda,n}$ and $\Psi_\lambda(P)=\alpha_\lambda^{\top}\phi_\lambda$ denoted with $D^*_{\Psi_{\lambda,n}(),a,P_0}$ and $D^*_{\Psi_\lambda(),a,P_0}$, obtained by taking the inner product of the gradient of $\alpha$ with the basis $\phi_{\lambda,n}(a)$ and $\phi_\lambda(a)$, respectively, as in $D^*_{\Psi(),a,P}=D^{*,\top}_{\alpha(),P}\phi(a)$ above.

For $\phi\in D^{(k)}([0,1])$, let $\Pi(\phi\mid D^{(k)}({\cal R}_{\lambda,n}))$ denote a projection of $A\rightarrow \phi(A)$ onto the working model $D^{(k)}({\cal R}_{\lambda,n})$, where one can use any norm to define the projection. We use this projection below to state how well the data-adaptive sieve approximates the fixed sieve.
Throughout the following lemma, $P_{n,\lambda}^*=(\bar{Q}_{n,\lambda}^*,g_{n,\lambda})$ denotes the TMLE targeting the data-adaptive coefficient $\alpha_{\lambda,n}(P_0)$, i.e. $\bar{Q}_{n,\lambda}^*$ solves $P_n D^*_{\alpha_{\lambda,n}(),P_{n,\lambda}^*}=0$.

\begin{lemma}
\label{lem:sieve rates}
Suppose that ${\cal R}_{\lambda,n}$ is a data dependent sieve, and let ${\cal R}_{\lambda}$ be an independent sieve (i.e. fixed) of size $n(\lambda)\leq n(\lambda,n)$.

We assume that $D^{(k)}({\cal R}_{\lambda,n})$ satisfies
\[
\sup_{\phi\in D^{(k)}_M([0,1])} \left\lvert \phi-\Pi(\phi\mid D^{(k)}({\cal R}_{\lambda,n}))\right\rvert_{\infty}=O_P^+(n(\lambda)^{-k^*}).
\]
We assume the propensity is bounded away from zero, $\inf_{a,w} g_n(a\mid w)>\delta>0$ for some $\delta>0$. We then assume that the empirical-process term, arising from the decomposition of the score equation, satisfies:
\[
n^{-1/2} O_P\bigl(J_{\infty}(\delta_n,D_{M_n}^{\min(\bar{k},m)}([0,1]^{\bar{d}}))\bigr)=o_P((n(\lambda)/n)^{1/2}),\]
where $J_{\infty}(\delta,{\cal F})$ denotes the entropy (bracketing) integral of the function class ${\cal F}$ up to radius $\delta$ (see \citep{van2023higher}); the radius is $\delta_n=O^+(n(\lambda)^{-k^*})$; and
$M_n\equiv I(\bar{k}\leq m)\pl \bar{Q}_{n,\lambda}^*\pl^*_{v,\bar{k}}+I(\bar{k}>m)\pl g_n\pl^*_{v,m}$, where $\pl f\pl^*_{v,k}$ denotes the $k$-th order sectional variation norm of $f$ (i.e. the smallest $M$ with $f\in D^{(k)}_M$, so that $M_n$ is the $\min(\bar{k},m)$-th order sectional variation norm of whichever of $\bar{Q}_{n,\lambda}^*$, $g_n$ is less smooth). Finally, we assume the bias term is negligible:
\[
n(\lambda)^{-k^*} \pl \bar{Q}_{n,\lambda}^*-\bar{Q}_0\pl_{P_0}=o_P((n(\lambda)/n)^{1/2}).\]
In particular, this last condition holds if we use the $\bar{k}$-th order cross-validated HAL-MLE as the initial estimator of $\bar{Q}_{n,\lambda}^*$ and $\bar{Q}_0\in D^{(\bar{k})}_M([0,1]^{\bar{d}})$, since then $\pl \bar{Q}_{n,\lambda}^*-\bar{Q}_0\pl_{P_0} =O^+( n^{-\bar{k}^*/(2\bar{k}^*+1)})$.

 Then, we have
 \[
 P_n D^*_{\Psi_\lambda(),a,P_{n,\lambda}^*}=
o_P((n(\lambda)/n)^{1/2}).\]
\end{lemma}

\begin{proof}
Let $P_{n,\lambda}^*$ be the TMLE of $P_0$ targeting $\alpha_{\lambda,n}(P_0)$, which solves $P_n D^*_{\alpha_{\lambda,n}(),P_{n,\lambda}^*}=0$ and therefore $P_n D^*_{\Psi_{\lambda,n}(),a,P_{n,\lambda}^*}=0$ for all $a$. This is the score equation of the working model it targets, but the asymptotic normality argument of the previous subsection requires the score equation of the fixed working model, $P_n D^*_{\Psi_\lambda(),a,P_{n,\lambda}^*}=0$, which the TMLE does not target directly. We decompose this fixed working model score into its $Q_W$- and $\bar{Q}$-components. We note that we have $P_n D^*_{\alpha_\lambda,Q_W,P_{n,\lambda}^*}=0$ due to using the empirical distribution of $W$. So it remains to establish that
\[
P_n D^*_{\Psi_\lambda(),a,P_{n,\lambda}^*}= o_P((n(\lambda)/n)^{1/2}),
\]
that is, that the TMLE built from the data-adaptive sieve solves the score equations of the fixed sieve up to a remainder of order $o_P((n(\lambda)/n)^{1/2})$, which is negligible relative to the rate of convergence of the TMLE.

Recall
\[
D^*_{\Psi_\lambda(),\bar{Q},P} =\phi_\lambda^{\top}(a)\omega(A)\Sigma_{\phi_\lambda}^{-1}{\phi}_\lambda(A)\frac{1}{g(A\mid W)}(Y-\bar{Q}_P(A,W)).
\]
Let $\tilde{\phi}_\lambda=\Sigma_{\phi_\lambda}^{-1}{\phi}_\lambda$, and, for $\phi\in D^{(k)}([0,1])$, let $D^*_{\phi,\bar{Q},P}\equiv \phi(A)\omega(A)/g(A\mid W)(Y-\bar{Q}_P(A,W))$ denote the $\bar Q$-tangent-space score in direction $\phi$ (i.e., the first term of $D^*_{\phi,P}$ from Lemma~\ref{lemma1}, dropping the $Q_W$-correction term, which we have already shown vanishes after applying $P_n$ above). Then, $D^*_{\alpha_\lambda(),\bar{Q},P}=D^*_{\tilde{\phi}_\lambda,\bar{Q},P}$, and $D^*_{\phi,\bar Q,P}$ is linear in $\phi$.
Thus, $D^*_{\Psi_\lambda(),\bar{Q},P}=\sum_{j\in \cal R_\lambda}\phi_j(a)D^*_{\tilde{\phi}_j,\bar{Q},P}$.
We know that $P_n D^*_{\tilde{\phi}_{\lambda,n},\bar{Q},P_{n,\lambda}^*}=0$. Recalling the projection $\Pi(\phi\mid D^{(k)}({\cal R}_{\lambda,n}))$ defined above, note that $P_n D^*_{\Pi(\phi\mid D^{(k)}({\cal R}_{\lambda,n})),\bar{Q},P_{n,\lambda}^*}=0$ for all $\phi\in D^{(k)}([0,1])$.
By the projection-approximation rate assumed above,
\[
\sup_{\phi\in D^{(k)}_M([0,1])} \left\lvert \phi-\Pi(\phi\mid D^{(k)}({\cal R}_{\lambda,n}))\right\rvert_{\infty}=O^+(n(\lambda)^{-k^*}),\]
where we know that if ${\cal R}_{\lambda,n}$ was a fully effective working model then this could hold with $n(\lambda,n)=n(\lambda)$, but data adaptively selected working models might have ineffective basis functions while the overall method still controls overfitting.
We can bound the desired score equation by subtracting the projection on the solved scores; decomposing it in an empirical process term and a second order bias term, as follows:
\[
\begin{array}{l}
\sum_{j\in {\cal R}_\lambda}\phi_j(a) P_n D^*_{\tilde{\phi}_j,\bar{Q},P_{n,\lambda}^*}=
\sum_{j\in {\cal R}_\lambda}\phi_j(a)P_n \omega(A)\tilde{\phi}_j(A)/g_n(A\mid W)(Y-\bar{Q}_{n,\lambda}^*(A,W))\\
=
\sum_{j\in {\cal R}_\lambda}\phi_j(a) P_n \left(\tilde{\phi}_j-\Pi(\tilde{\phi}_j\mid D^{(k)}({\cal R}_{\lambda,n}))\right) \omega(A)/g_n(A\mid W)(Y-\bar{Q}_{n,\lambda}^*(A,W))\\
\equiv
\sum_{j\in {\cal R}_\lambda}\phi_j(a)
P_n f_{j,n}(A) \omega(A)/g_n(A\mid W)(Y-\bar{Q}_{n,\lambda}^*(A,W))\\
=\sum_{j\in {\cal R}_\lambda}\phi_j(a)
(P_n-P_0)f_{j,n}(A)\omega(A)/g_n(A\mid W)(Y-\bar{Q}_{n,\lambda}^*)\\
+\sum_{j\in {\cal R}_\lambda}\phi_j(a)
P_0 f_{j,n}(A)\omega(A)/g_n(A\mid W)(\bar{Q}_0-\bar{Q}_{n,\lambda}^*)\\
=o_P((n(\lambda)/n)^{1/2})+O^+_P(n(\lambda)^{-k^*})\pl \bar{Q}_{n,\lambda}^*-\bar{Q}_0\pl_{P_0}.
\end{array}
\]
The first term is $o_P((n(\lambda)/n)^{1/2})$ since the empirical-process term is negligible, and the second is $o_P((n(\lambda)/n)^{1/2})$ by the negligible-bias condition together with the Cauchy-Schwarz inequality.
\end{proof}
{\bf Remark (verifying the empirical-process condition):}
The empirical process term can be bounded by the entropy integral $J_{\infty}(\delta_n)$ implied by $D^{(\bar{k})}_M([0,1]^{\bar{d}})$ containing $\bar{Q}_{n,\lambda}^*$ and
$D^{(m)}_M([0,1]^{\bar{d}})$ containing $g_n$, using the positivity condition.
Thus, this yields a bound $J_{\infty}(\delta_n,D_{M_n}^{\min(\bar{k},m)}([0,1]^{\bar{d}}))$, where
we can set $\delta_n=O^+(n(\lambda)^{-k^*})$, by assumption on the rate of convergence of $f_{j,n}$. $M_n$ represents the sectional variation norm (i.e. $L_1$-norm if we use $\bar{k}$-th order HAL as initial esitmator $Q_n$) associated with the fit $g_n$ and $\bar{Q}_{n,\lambda}^*$. So the empirical process term is then bounded by
$n^{-1/2} O_P(J_{\infty}(\delta_n,D_{M_n}^{\min(\bar{k},m)}([0,1]^{\bar{d}}))$.
Given $\min(\bar{k},m)\geq 1$, and the rate we will choose for $n(\lambda)$, it follows that the latter is $o_P((n(\lambda)/n)^{1/2})$, as required by the empirical-process condition, if $M_n$ converges slowly enough to infinity. However, note that the targeting only adds a function of $A$ to the initial $Q_n$ so that as long as we control the $L_1$-norm of $\epsilon$ in the targeting step, the sectional variation norm of $Q_{n,\lambda}^*$ is controlled. Fortunately this can be explicitly verified: one could enforce a bound on $n(\lambda,n)$ that guarantees this $L_1$-norm increase is limited. Besides, in our DRC example, we can just use an equally spaced grid so that $n(\lambda,n)/n(\lambda)=O_P(1)$.
But these arguments show that, in general, too large $n(\lambda,n)/n(\lambda)$ could hurt the sieve TMLE of the target function.

{\bf Remark (on the smoothness condition $k\geq\bar{k}$):}
Since the target function is smoother than $\bar{Q}_0$: i.e., $k\geq \bar{k}$, the fluctuation model of the initial estimator $\bar{Q}_n$ includes basis functions $\phi_j$, $j\in {\cal R}_{\lambda,n}$ that are at least as smooth as the $\bar{k}$-th order splines in the initial HAL $\bar{Q}_n$. Therefore, the targeting step will not deteriorate the rate of convergence of $\bar{Q}_n$ based on smoothness assumption $\bar{Q}_0\in D^{(\bar{k})}([0,1]^{\bar{d}})$.

The size $n(\lambda)$ of the oracle fixed sieve can be of smaller size than the actual working model $D^{(k)}({\cal R}_{\lambda,n})$ used in the sieve TMLE, as long as the conditions of the theorem are still met (i.e. not too much overfit so that sectional variation norms are controlled). However, if $n(\lambda,n)$ grows too fast then the targeting step in the sieve TMLE will make the fits $Q_{n,\lambda}^*$ too erratic so that the rate of convergence of the initial $Q_{n,\lambda}$ might deteriorate affecting the second order remainder conditions in the theorem, and the empirical process conditions may be violated as well. On the other hand, if we replace the targeting step by a LASSO step with bounded $L_1$-norm of $\epsilon$, as in next section, then this issue can be avoided.

By application of this lemma, we can now apply the above Theorem \ref{chtargetfunctiontmleI}
for the sieve TMLE for the non-random sieve \ $D^{(k)}({\cal R}_\lambda)$.

\begin{theorem}[Asymptotic normality of the data-adaptive sieve TMLE of the DRC]\label{chtargetfunctiontmleIadaptive}
Adopt the setting and assumptions of Theorem~\ref{chtargetfunctiontmleI}. Let ${\cal R}_\lambda$ be a fixed companion sieve of size $n(\lambda)$, where the second-order remainder and empirical-process terms are $o_P((n(\lambda)/n)^{1/2})$.

Let ${\cal R}_{\lambda,n}$ be the data-adaptive sieve of size $n(\lambda,n)\geq n(\lambda)$, and let $P_{n,\lambda}^*=(\bar{Q}_{n,\lambda}^*,g_{n,\lambda})$ be the corresponding TMLE targeting the data-adaptive coefficient $\alpha_{\lambda,n}(P_0)$, i.e.\ solving $P_n D^*_{\alpha_{\lambda,n}(),P_{n,\lambda}^*}=0$, with plug-in estimator $\Psi_{\lambda,n}(P_{n,\lambda}^*)(a)=\alpha_{\lambda,n}(P_{n,\lambda}^*)^{\top}{\bf \phi}_{\lambda,n}^*(a)$. Suppose the conditions of Lemma~\ref{lem:sieve rates} hold, together with the nuisance rate $\pl \bar{Q}_{n,\lambda}^*-\bar{Q}_0\pl_{P_0}=O^+_P(n^{-\bar{k}^*/(2\bar{k}^*+1)})$.

Then, the conclusions of Theorem~\ref{chtargetfunctiontmleI} carry over to the data-adaptive sieve TMLE, with respect to the projection and with respect to the target function:

{\bf (i) Asymptotic normality with respect to the projection $\Psi_{\lambda}$.}
Uniformly in $a$,
\[
(n/n(\lambda))^{1/2}(\Psi_{\lambda}(P_{n,\lambda}^*)-\Psi_{\lambda}(P_0))(a)=
(n/n(\lambda))^{1/2}(P_n-P_0)D^*_{\Psi_{\lambda}(),a,P_0}+o_P(1),
\]
where $D^*_{\Psi_{\lambda}(),a,P_0}$ is the canonical gradient of the fixed-sieve projection $\Psi_{\lambda}$ from Lemma~\ref{lem:canonical gradient}. Consequently, for each $a$,
\[
\sigma_{\lambda,P_0}^{-1}(a)(n/n(\lambda))^{1/2}(\Psi_{\lambda}(P_{n,\lambda}^*)-\Psi_{\lambda}(P_0))(a)\Rightarrow_d N(0,1),
\]
with $\sigma^2_{\lambda,P_0}(a)=P_0\{D^*_{\Psi_{\lambda}(),a,P_0}\}^2$, and the one-dimensional class $\{D^*_{\Psi_{\lambda}(),a,P_0}:a\}$ yields $0.95$-simultaneous confidence bands of width $O(\log n\,(n(\lambda)/n)^{1/2})$.

{\bf (ii) Asymptotic normality with respect to the target function $\psi_0$.}
Balancing the approximation bias $O^+(n(\lambda)^{-k^*})$ against the standard error by undersmoothing, i.e.\ taking $n(\lambda)\sim n^{1/(2k^*+1)}$ up to $\log n$-factors with $c$ large enough that $n(\lambda)^{-k^*}(\log n)^{c}=o((n(\lambda)/n)^{1/2})$, the same result holds for the data-adaptive plug-in estimator $\Psi_{\lambda,n}(P_{n,\lambda}^*)$ relative to the true target function: uniformly in $a$,
\[
(n/n(\lambda))^{1/2}(\Psi_{\lambda,n}(P_{n,\lambda}^*)-\Psi(P_0))(a)=
(n/n(\lambda))^{1/2}(P_n-P_0)D^*_{\Psi_{\lambda}(),a,P_0}+r_n(a),
\]
with $\sup_a\mid r_n(a)\mid=o_P(1)$. Consequently, for each $a$,
\[
\sigma_{\lambda,P_0}^{-1}(a)(n/n(\lambda))^{1/2}(\Psi_{\lambda,n}(P_{n,\lambda}^*)-\Psi(P_0))(a)\Rightarrow_d N(0,1),
\]
and the same one-dimensional class $\{D^*_{\Psi_{\lambda}(),a,P_0}:a\}$ yields $0.95$-simultaneous confidence bands of width $O(\log n\,(n(\lambda)/n)^{1/2})$.
\end{theorem}

\begin{proof}
    Apply the exact first-order expansion of Theorem~\ref{chtargetfunctiontmleI} to the fixed-sieve projection $\Psi_{\lambda}$ evaluated at $P_{n,\lambda}^*$. With the exact remainder $R_{\Psi_{\lambda},a}(P,P_0)=\Psi_{\lambda}(P)-\Psi_{\lambda}(P_0)+P_0 D^*_{\Psi_{\lambda}(),a,P}$,
    \[
    \Psi_{\lambda}(P_{n,\lambda}^*)(a)-\Psi_{\lambda}(P_0)(a)=
    (P_n-P_0)D^*_{\Psi_{\lambda}(),a,P_{n,\lambda}^*}
    - P_n D^*_{\Psi_{\lambda}(),a,P_{n,\lambda}^*}
    +R_{\Psi_{\lambda},a}(P_{n,\lambda}^*,P_0).
    \]

    By Lemma~\ref{lem:sieve rates}, the data-adaptive TMLE $P_{n,\lambda}^*$ solves the scores of the fixed sieve up to a negligible order, uniformly in $a$:
    \[
    P_n D^*_{\Psi_{\lambda}(),a,P_{n,\lambda}^*}=o_P((n(\lambda)/n)^{1/2}).
    \]

    By assumptions on the nuisance rates, as in Lemma~\ref{lem:sieve rates}, the remainder is the product
    \[
    R_{\Psi_{\lambda},a}(P_{n,\lambda}^*,P_0)=O_P^+\!\bigl(\pl \bar{Q}_{n,\lambda}^*-\bar{Q}_0\pl_{P_0}\,\pl g_{n,\lambda}-g_0\pl_{P_0}\,n(\lambda)\bigr)=O_P^+\!\bigl(n^{-\bar{k}^*/(2\bar{k}^*+1)}\,n^{-m^*/(2m^*+1)}\,n(\lambda)\bigr),  
    \]
    so at $n(\lambda)\sim n^{1/(2k^*+1)}$ it is $o_P((n(\lambda)/n)^{1/2})$ by the smoothness condition $n^{-\bar{k}^*/(2\bar{k}^*+1)}n^{-m^*/(2m^*+1)}=o^-(n^{-(k^*+1)/(2k^*+1)})$. The empirical-process term is bounded by $n(\lambda)\,n^{-1/2}J_{\infty}(\delta_n)$, where the radius $\delta_n$ is the $L^2(P_0)$-rate of the normalized gradient difference $\{D^*_{\Psi_{\lambda}(),a,P_{n,\lambda}^*}-D^*_{\Psi_{\lambda}(),a,P_0}\}/n(\lambda)$ and is therefore set by the nuisance rates above. This too is $o_P((n(\lambda)/n)^{1/2})$. 
    
    Both the remainder and the empirical process term are thus $o_P((n(\lambda)/n)^{1/2})$ uniformly in $a$, and so
    \[
    \Psi_{\lambda}(P_{n,\lambda}^*)(a)-\Psi_{\lambda}(P_0)(a)=
    (P_n-P_0)D^*_{\Psi_{\lambda}(),a,P_0}+o_P((n(\lambda)/n)^{1/2}),
    \]
    uniformly in $a$.   
    Since ${\cal R}_\lambda$ is non-random, $D^*_{\Psi_{\lambda}(),a,P_0}$ is a fixed function of $O$, so the leading term is a normalized sum of i.i.d.\ mean-zero variables and the CLT applies.
    We thus obtain asymptotic normality with respect to the projection $\Psi_{\lambda}$: for each $a$,
    \[
    \sigma_{\lambda,P_0}^{-1}(a)(n/n(\lambda))^{1/2}(\Psi_{\lambda}(P_{n,\lambda}^*)-\Psi_{\lambda}(P_0))(a)\Rightarrow_d N(0,1),\qquad \sigma^2_{\lambda,P_0}(a)=P_0\{D^*_{\Psi_{\lambda}(),a,P_0}\}^2.
    \]
    and, uniformly in $a$, empirical process theory with the one-dimensional class $\{D^*_{\Psi_{\lambda}(),a,P_0}:a\}$ implies $0.95$-simultaneous confidence bands of width $O(\log n\,(n(\lambda)/n)^{1/2})$. This establishes part (i).
    
    At $n(\lambda)\sim n^{1/(2k^*+1)}$ the approximation order $O^+(n(\lambda)^{-k^*})$ equals the standard-error rate $(n(\lambda)/n)^{1/2}$, so, exactly as in Theorem~\ref{chtargetfunctiontmleI}, we undersmooth by $\log n$-factors, taking $c$ large enough that $n(\lambda)^{-k^*}(\log n)^{c}=o((n(\lambda)/n)^{1/2})$. Under this choice $\sup_a\mid \Psi_{\lambda}(P_0)-\Psi(P_0)\mid(a)$ and $\sup_a\mid \Psi_{\lambda,n}(P_{n,\lambda}^*)-\Psi_{\lambda}(P_{n,\lambda}^*)\mid(a)$ are $o_P((n(\lambda)/n)^{1/2})$, giving
    \[
    (n/n(\lambda))^{1/2}(\Psi_{\lambda,n}(P_{n,\lambda}^*)-\Psi(P_0))(a)=
    (n/n(\lambda))^{1/2}(P_n-P_0)D^*_{\Psi_{\lambda}(),a,P_0}+r_n(a),
    \]
    with $\sup_a\mid r_n(a)\mid=o_P(1)$, and pointwise normality and simultaneous bands follow exactly as in Theorem~\ref{chtargetfunctiontmleI}. This establishes part (ii).
\end{proof}

\section{Targeted HAL-MLE, avoiding construction of sieve}\label{section4b}

The sieve TMLE of Section~\ref{Chtargetfunction3} required the user to construct an appropriate fixed sieve $D^{(k)}({\cal R}_\lambda)$ with the uniform approximation error $O^+(n(\lambda)^{-k^*})$ w.r.t.\ the target function space. The data-adaptive sieve TMLE of Section~\ref{Chtargetfunction4} relaxed this requirement, allowing the sieve $D^{(k)}({\cal R}_{\lambda,n})$ to be somewhat ineffective in approximating the whole target function space relative to the optimal $O^+(n(\lambda)^{-k^*})$ error, as long as the resulting $Q_{n,\lambda}^*$
is not too much of an overfit so that the empirical process and remainder conditions are still met. In either case, we need to construct a careful sieve to carry out these TMLEs. This is trivial in the univariate DRC example, but becomes challenging for higher dimensional target functions.

As in the HAL literature, in our proposed Targeted HAL-MLE, we first choose a large initial basis ${\cal R}_N$ of $N$ basis functions so that $D^{(k)}({\cal R}_N)\approx D^{(k)}({\cal R}_1)=D^{(k)}([0,1]^{d})$.
Consider the same TMLE as above but with this high dimensional working model $D^{(k)}({\cal R}_N)$ instead of $D^{(k)}({\cal R}_{\lambda,n})$. This will now involve a targeting step with a high dimensional (i.e., $N$) parametric least favorable path corresponding with the $N$-dimensional canonical gradient $D_{\alpha_{{\cal R}_N}(),P}$.
In particular, in the DRC example, the targeting step involved least squares regression of the residual $Y-\bar{Q}_n$ onto $\sum_{j\in {\cal R}_N}\epsilon(j)\phi_j$ with weights $1/g_n(A\mid W)$. Fitting $\epsilon$ with the MLE as in standard TMLE would be a bad overfit of $Q_0$ since it would be like an MLE over a high dimensional linear model. In addition, it would be ill defined if the number of basis functions is larger than sample size $n$. However, instead of doing the unpenalized MLE step, we use LASSO to fit $\epsilon$: that is, we maximize the log-likelihood or minimize empirical risk over $\epsilon$ under a constraint $L_1$-norm $\sum_{j\in {\cal R}_N}\mid \epsilon(j)\mid \leq \lambda$. This $L_1$-norm corresponds with the increase in $k$-th order sectional variation norm of the plug-in TMLE of the target function relative to plug-in using $Q_n$ as well as the increase in sectional variation norm of $Q_{n,\lambda}^*$ relative to $Q_n$. So by controlling the $L_1$-norm we completely control the conditions of the theorem: we preserve rates of convergence of the nuisance estimators and we control empirical process conditions.
We now define $D^{(k)}({\cal R}_{\lambda,n})$ as the working model for each $\lambda$, where ${\cal R}_{\lambda,n}$ is the set of non-zero coefficients $\epsilon_n(j)$ in the LASSO fit and $n(\lambda,n)$ represents the size of this working model.
One could select $\lambda$ with a cross-validation selector and we discuss such a selector in further detail in the next section.

If the resulting working model is sparse, as one expects for LASSO, one could consider a relaxed LASSO step instead in which we approximate an actual MLE over the $\epsilon$ restricted to its non-zero coefficients. At this step one could relax the $L_1$-norm till $P_n D^*_{\Psi_{{\cal R}_{\lambda,n}}(),a,P_{n,\lambda}^*}=o_P((n(\lambda,n)/n)^{1/2})$ as we typically recommend for standard TMLE, thereby not relaxing more than needed. In other words, this would be like the sieve TMLE with this sieve of working models, which thus simply involves refitting the non-zero coefficients $\epsilon(j)$ in the LASSO fit with an unpenalized or weakly penalized MLE step.
The combined procedure for this targeting step would then be nothing else than relaxed LASSO. The final TMLE $P_{n,\lambda}^*$ can then be viewed as a data adaptive sieve TMLE of the projection $\Psi_{{\cal R}_{\lambda,n}}(P_0)$ of $\psi_0$ onto $D^{(k)}({\cal R}_{\lambda,n})$ and influence curve based inference could be carried out. In other words, we can apply the above theorem for the data adaptive sieve TMLE including its suggested inference method.

However, it is not necessary to do the relaxed LASSO, one could simply accept the LASSO as the targeting step defining the Targeted HAL-MLE. 
The reason why we can do this is that the proof fully relies on the linear span of the scores solved by the targeted $Q_{n,\lambda}^*$ approximating a sparse effective (unknown) working model $D^{(k)}({\cal R}_\lambda)$ of size $n(\lambda)$. The relaxed LASSO would solve the non-zero coefficient scores equations $P_n D^*_{\alpha_{\lambda,n}(),P_{n,\lambda}^*,j=0}$, $j\in {\cal R}_{n,\lambda}$, exactly while the LASSO would lose only one degree of freedom and solve $n(\lambda,n)-1$ linear combinations of these $n(\lambda,n)$ scores, already well approximating the full $n(\lambda,n)$-dimensionsal score equation space. Moreover, by not using relaxed LASSO, one can select a larger $\lambda$ and thereby end up selecting larger working model, and thus solve extra scores. 
In addition, by using LASSO we control the sectional variation norms and preserve the rate of convergence of the initial $Q_n$.
So from a pure pointwise asymptotic normality point of view, the relaxation is not important and might not even be beneficial.

For inference based on the Targeted HAL-MLE one could still use influence curve based inference treating it as a TMLE of the projection of $\psi_0$ onto $D^{(k)}({\cal R}_{\lambda,n})$, or one could use a targeted bootstrap for inference.

\subsection{Selecting the $L_1$-norm in the LASSO of the targeting step in T-HAL-MLE}
\label{Sec:L1norm}
The above T-HAL-MLE is indexed by a choice of $L_1$-norm in the LASSO targeting step. 
The larger $\lambda$ the larger the resulting $D^{(k)}({\cal R}_{\lambda,n})$, thereby ranging from a TMLE targeting a low dimensional projection of $\psi_0$ towards a close approximmation of the whole $\psi_0$. For each of these $\lambda$-specific TMLE we have a variance estimator based on the efficient influence curve of $\Psi_{{\cal R}_{\lambda,n}}(P)$ which will be increasing in $\lambda$, while the asymptotic bias will be reduced as $\lambda$ increases. A cross-validation selector of $\lambda$ might represent the minimal choice since it is only concerned about the fitting of $Q_0$, so that, just as with HAL, some undersmoothing of $\lambda$ might be needed. For that purpose, Lepski's method is very sensible in which we look for the plateau in the $\lambda$-specific confidence intervals, going from small $\lambda$ to large $\lambda$. This essentially corresponds with starting at the cross-validation selector and keep increasing $\lambda$ stepwise till the change in estimate is dominated by change in standard error estimator. We refer to \citep{zhang2025constructingconfidenceintervalsinfinitedimensional,JonathanLevykernelsmoothingcdfCATE} and Chapter 25 of \citep{vanderLaan&Rose18} by van der Laan, Bibaut, Luedtke.

Let $P_{n,\lambda}^*$ represent the TMLE, which in the DRC example corresponds with $Q_{W,n},\bar{Q}_{n,\lambda}^*$, where $\bar{Q}_{n,\lambda}^*$ is the result of the relaxed LASSO targeting step applied to initial estimator $\bar{Q}_n$.
Note that honest cross-validation would require refitting $P_n^0$ or $\bar{Q}_n$ on each training sample; computing the TMLE update using $\lambda$ as $L_1$-norm in the LASSO based on this same training sample; and evaluating the cross-validated risk on the corresponding validation sample; average across the $V$ sample splits in a $V$-fold cross-validation scheme; and selecting $\lambda_{n,cv}$ as the choice of $\lambda$ that optimizes the $V$-fold cross-validated risk over all these $\lambda$-specific T-HAL-MLEs. If $\bar{Q}_n$ was already a discrete Super-Learner, then one could set up the $V$-fold sample splitting so that we already have the fit $\bar{Q}_{n,v}$ on the training sample, thereby avoiding having to redo the discrete Super-Learner on the training sample. In that case, it would not be a completely honest cross-validation selector since the selection of the algorithm in the discrete Super-Learner was already based on the validation samples. Nonetheless, this would be an appropriate and theoretically grounded cross-validation selector. 

{\bf Rate-optimal selector for the $L_1$-norm:}\label{par:rate-optimal-selector}
A computationally efficient alternative to Lepski's approach is to select $\lambda$ by focusing on $n(\lambda,n)$ directly. From Theorem~\ref{chtargetfunctiontmleI}, we know that setting $n(\lambda,n)\sim n^{1/(2k^*+1)}$ achieves the optimal rate $n^{-k^*/(2k^*+1)}$ up to $\log n$-factors. This motivates constraining $\lambda$ from below
\[
\lambda_n^* = \min\left\{ \lambda \geq 0 : n(\lambda,n) \geq \left\lceil c \cdot n^{1/(2k^*+1)} \right\rceil \right\}
\]
for some constant $c$. This can be implemented by fixing the initial estimator $\bar{Q}_n$, rerunning only the LASSO targeting step with $\lambda$ increasing stepwise, and stopping at the first $\lambda$ where the working model size $n(\lambda,n)$ meets the threshold. 
Unlike Lepski's method, this requires no cross-fitting of $\bar{Q}_n$, making it substantially more computationally efficient. 

We also propose a flexible data-adaptive version by truncating the cross-validated sieve size from below and from above, constraining it to a theoretically motivated interval. We can choose to select the cross-validation-chosen penalty if the resulting sieve size falls within the interval $\mathcal{I}_n = [\lceil c_1 \cdot n^{1/(2k^*+1)} \rceil,\, \lfloor c_2 \cdot n^{1/(2k^*+1)} \rfloor]$, otherwise, we truncate to the nearest boundary, choosing a sieve of size $\lceil c_1 \cdot n^{1/(2k^*+1)} \rceil$ if the cross-validated size is too small, and $\lfloor c_2 \cdot n^{1/(2k^*+1)} \rfloor$ if it is too large.
\begin{equation*}
\lambda^*_{n,cv} = \begin{cases}
\inf\left\{\lambda : n(\lambda,n) \geq \lceil c_1 \cdot n^{1/(2k^*+1)} \rceil\right\}
  & \text{if } n(\lambda_{n,cv}) < \lceil c_1 \cdot n^{1/(2k^*+1)} \rceil \\[6pt]
\sup\left\{\lambda : n(\lambda,n) \leq \lfloor c_2 \cdot n^{1/(2k^*+1)} \rfloor\right\}
  & \text{if } n(\lambda_{n,cv}) > \lfloor c_2 \cdot n^{1/(2k^*+1)} \rfloor \\[6pt]
\lambda_{n,cv}
  & \text{otherwise}
\end{cases}
\end{equation*}
Sensible choices of $c$, $c_1$, and $c_2$ can be found via simulation and are further discussed in Section~\ref{sec:simulation}.


\section{General asymptotic normality of TMLE of projection of target function on spline working model}\label{Chtargetfunction5}

In this section, we give a general formulation and proof for asymptotic normality results beyond the DRC example. We set out the framework for a fixed sieve ${\cal R}_\lambda$ and determine the canonical gradient of the projection and its variance. We then state the asymptotic normality of the fixed-sieve TMLE with respect to both the projection $\Psi_\lambda(P_0)$ and the true target $\Psi(P_0)$ (Theorem~\ref{chtargetfunctiontmleIgen}). Finally, following Section~\ref{Chtargetfunction4}, we introduce the data-adaptive sieve ${\cal R}_{\lambda,n}$ and show that it approximately solves the score equations of the fixed companion sieve ${\cal R}_\lambda$, and conclude that the asymptotic normality statement carries over to the data-adaptive sieve TMLE (Theorem~\ref{chtargetfunctiontmleIgenadaptive}).

\subsection{General framework of TMLE for target functionals}\label{Chtargetfunction5framework}

We observe $n$ i.i.d. copies of $O\sim P_0\in {\cal M}$ and our target parameter is $\Psi:{\cal M}\rightarrow D^{(k)}([0,1]^{d})$ so that
$\Psi(P_0)$ represents our target function.
As in Sections~\ref{Chtargetfunction3} and~\ref{Chtargetfunction4}, we first develop the framework for a fixed sieve ${\cal R}_{\lambda}$.

Let ${\bf \Psi}_{\lambda}=D^{(k)}({\cal R}_{\lambda})$ be an $n(\lambda)$-dimensional sieve of $D^{(k)}([0,1]^{d})$ satisfying:
\[
\sup_{\psi\in D^{(k)}([0,1]^{d})}\inf_{\psi_{\lambda}\in D^{(k)}({\cal R}_{\lambda})}\pl \psi-\psi_{\lambda}\pl_{\infty}=O^+(n(\lambda)^{-k^*}).\]
Given an inner product $\langle \cdot,\cdot\rangle_{\psi}$ on $D^{(k)}([0,1]^{d})$,
define
\[
\alpha_{\lambda}(P)\equiv \arg\min_{\alpha} \left\lvert \Psi(P)-\sum_{j\in {\cal R}_{\lambda}}\alpha(j)\phi_j^*\right\rvert_{\psi}.\]
We then define $\Psi_{\lambda}:{\cal M}\rightarrow D^{(k)}({\cal R}_{\lambda})$ as
$\Psi_{\lambda}(P)=\alpha_{\lambda}(P)^{\top}{\bf \phi}^*_{\lambda}$ representing the projection of $\Psi(P)$ onto $D^{(k)}({\cal R}_{\lambda})$.

We assume $\alpha_{\lambda}$ is pathwise differentiable at any $P\in {\cal M}$ and let $D^*_{\alpha_{\lambda}(),P}$ be its canonical gradient.
We note that $\Psi_{\lambda}$ is pathwise differentiable at $P$ with canonical gradient $D^*_{\Psi_{\lambda}(),P}={\bf \phi}^{*,\top}_{\lambda}D^{*}_{\alpha_{\lambda}(),P}$ and, in particular, $\Psi_{\lambda,x}(P)=\Psi_{\lambda}(P)(x)$ is pathwise differentiable at $P$ with canonical gradient $D^*_{\Psi_{\lambda}(),x,P}={\bf \phi}^{*,\top}_{\lambda}(x)D^{*}_{\alpha_{\lambda}(),P}$. See Appendix~\ref{app:canonicalgradient} for guidance on how to construct canonical gradients for pathwise-differentiable parameters $\alpha_{\lambda}$ .

Let $R_{\alpha_{\lambda}}(P,P_0)\equiv \alpha_{\lambda}(P)-\alpha_{\lambda}(P_0)+P_0 D^*_{\alpha_{\lambda}(),P}$ be the exact second order remainder for $\alpha_{\lambda}()$.
Then, $R_{\Psi_{\lambda},x}(P,P_0)\equiv \Psi_{\lambda}(P)(x)-\Psi_{\lambda}(P_0)(x)+P_0 D^*_{\Psi_{\lambda}(),x,P}$ is given by
$R_{\Psi_{\lambda},x}(P,P_0)=R_{\alpha_{\lambda}}(P,P_0)^{\top}{\bf \phi}^*_{\lambda}(x)$, or, we could denote that globally,
$R_{\Psi_{\lambda}}(P,P_0)=R_{\alpha_{\lambda}}(P,P_0)^{\top}{\bf \phi}^*_{\lambda}$ as a function in $x$.

Let $P_n^0$ be an initial estimator of $P_0$, and let $P_{n,\lambda}^*$ be a TMLE targeting $\alpha_{\lambda}(P_0)$ satisfying $P_n D^*_{\alpha_{\lambda}(),P_{n,\lambda}^*}=0$. Since $D^*_{\Psi_{\lambda}(),x,P}$ is linear in $D^*_{\alpha_{\lambda}(),P}$, this implies $P_n D^*_{\Psi_{\lambda}(),x,P_{n,\lambda}^*}=0$ uniformly in $x$.

\subsection{Canonical gradient of the projection and its variance}

The framework above shows that $\Psi_{\lambda}$ is pathwise differentiable with canonical gradient $D^*_{\Psi_{\lambda}(),x,P}={\bf \phi}^{*,\top}_{\lambda}(x)D^{*}_{\alpha_{\lambda}(),P}$, a linear combination of the $n(\lambda)$ canonical gradients of the coefficients $\alpha_{\lambda}(P)$. We now show that the variance of this canonical gradient grows with $n(\lambda)$ as $n(\lambda)$.

\begin{lemma}\label{lem:variancegen}
We have $\Psi_{\lambda}(P)(x)=\sum_{j\in {\cal R}_\lambda}\alpha_{\lambda,P}(j)\phi_j^*(x)$, so that its canonical gradient $D^*_{\Psi_{\lambda}(),x,P}=\sum_{j\in {\cal R}_\lambda}D_{\alpha_{\lambda,j}(),P}\phi_j^*(x)$ is a linear combination of the $n(\lambda)$ canonical gradients $D_{\alpha_{\lambda,j}(),P}=D^*_{\phi_j^*,P}$ of the coefficients $\alpha_{\lambda,P}(j)$, where $\phi\rightarrow D^*_{\phi,P}$ extends linearly to any direction $\phi\in D^{(k)}({\cal R}_\lambda)$.

Define the inner product on $D^{(k)}({\cal R}_\lambda)$
\[
\langle \phi_1,\phi_2\rangle_{\psi,P}\equiv E_P D^*_{\phi_1,P}D^*_{\phi_2,P}.\]
If $\phi_j^*$, $j\in {\cal R}_{\lambda}$, is an orthonormal basis of $D^{(k)}({\cal R}_\lambda)$ w.r.t. this inner product, then
\[
D^*_{\Psi_{\lambda}(),x,P}=\sum_{j\in {\cal R}_{\lambda}}\phi_j^*(x)D^*_{\phi_j^*,P}\]
is a sum of uncorrelated random variables, so that
\[
\mbox{VAR}_P D^*_{\Psi_{\lambda}(),x,P}=\sum_{j\in {\cal R}_{\lambda}}\{\phi_j^*(x)\}^2.\]

For a general basis ${\bf \phi}_\lambda^*$, let ${\bf \phi}_\lambda^*(x)=(\phi_j^*(x): j\in {\cal R}_\lambda)$ and let $\Sigma_{\lambda}$ be the covariance matrix of the $n(\lambda)$-dimensional $D^*_{\alpha_{\lambda}(),P}$. Then
\[
\mbox{VAR}_P D^*_{\Psi_{\lambda}(),x,P}={\bf \phi}_\lambda^{*,\top}(x)\Sigma_{\lambda}{\bf \phi}_\lambda^*(x).\]
Writing the eigenvalue decomposition $\Sigma_{\lambda}=S_{\lambda}D_{\lambda}S^{\top}_{\lambda}$, where $D_{\lambda}$ is diagonal with eigenvalues $\gamma(j)$, $j\in {\cal R}_\lambda$, this variance is given by
\[
\sigma^2_{\lambda}(x)n(\lambda)= \sum_{j\in {\cal R}_\lambda} \{\phi_j^*(x)\}^2\gamma(j).\]
If we assume the maximal eigenvalue $\max_j \gamma(j)=O(1)$, then it follows that $\sigma^2_{\lambda}(x)n(\lambda)\sim \sum_{j\in {\cal R}_\lambda}\{\phi_j^*(x)\}^2=O(n(\lambda))$, so that the per-dimension variance
\[
\sigma^2_{\lambda}(x)\equiv \mbox{VAR}_P D^*_{\Psi_{\lambda}(),x,P}/n(\lambda)=O(1).\]
\end{lemma}

Informally, this shows that the variance of a TMLE of the projection curve $\Psi_{\lambda}(P_0)$ in $D^{(k)}({\cal R}_{\lambda})$ behaves as $O(n(\lambda)/n)$, corresponding with a rate of convergence $(n(\lambda)/n)^{1/2}$. Combined with the bias being $O(n(\lambda)^{-k^*})$, this yields the desired rate of convergence for the TMLE w.r.t. the true target function.

\subsection{Asymptotic normality of the fixed sieve}

\begin{assumption}\label{ass:targetfunctiongen}
We assume the following:
\begin{enumerate}
\item[{\rm (A1)}] {\bf Sieve approximation error:}
\[
\sup_{\psi\in D^{(k)}([0,1]^{d})}\inf_{\psi_{\lambda}\in D^{(k)}({\cal R}_{\lambda})}\pl \psi-\psi_{\lambda}\pl_{\infty}=O^+(n(\lambda)^{-k^*}).\]
\item[{\rm (A2)}] {\bf Pathwise differentiability:} $\alpha_{\lambda}$ is pathwise differentiable at every $P\in {\cal M}$ with canonical gradient $D^*_{\alpha_{\lambda}(),P}$.
\item[{\rm (A3)}] {\bf Sup-norm bias of the projection:} the approximation error is achieved by the projection (see \citet{van2023higher} for the result that the optimal approximation rate is attained by $L^2$-projections, among others):
\[
\sup_{P\in {\cal M}}\pl \Psi_{\lambda}(P)-\Psi(P)\pl_{\infty}=O^+(n(\lambda)^{-k^*}).\]
\item[{\rm (A4)}] {\bf Negligible second-order remainder:} $\sup_x \left| R_{\Psi_{\lambda},x}(P_{n,\lambda}^*,P_0)\right| =o_P((n(\lambda)/n)^{1/2})$.
\item[{\rm (A5)}] {\bf Negligible empirical-process term:} with $R_{n,x}\equiv (P_n-P_0)\{D^*_{\Psi_{\lambda}(),x,P_{n,\lambda}^*}-D^*_{\Psi_{\lambda}(),x,P_0}\}$, $\sup_x \mid R_{n,x}\mid =o_P((n(\lambda)/n)^{1/2})$.
\end{enumerate}
\end{assumption}

For (A4), we typically have that $R_{\Psi_{\lambda},x}(P,P_0)=O(\pl p_{n,\lambda}^*-p_0\pl_{\infty}^2 n(\lambda))$ by using Cauchy-Schwarz inequality and if we use $\bar{k}$-th order HAL-MLEs for $p_n^0$ with $k\geq 1$, then we obtain the square of the sup-norm rate of HAL-MLE $O^+(n^{-2\bar{k}^*/(2\bar{k}^*+1)})$. So then the assumption (A4) holds if
\[
n^{-2\bar{k}^*/(2\bar{k}^*+1)}n(\lambda)=o^-((n(\lambda)/n)^{1/2}),
\]
while we set $n(\lambda)\sim^+ n^{1/(2k^*+1)}$.
As shown above in the DRC example, (A5) is easily shown to hold using the rates of convergence of $p_{n,\lambda}^*-p_0$; writing $R_{n,x}=n(\lambda) R_{n,x}/n(\lambda)$; applying the entropy integral to this class of functions
$\{D^*_{\Psi_{\lambda}(),x,P}/n(\lambda): P\in {\cal M}\}$.

\begin{theorem}\label{chtargetfunctiontmleIgen}
Consider the framework for the target function, sieve approximation and corresponding TMLE set out above in Section~\ref{Chtargetfunction5framework} for the fixed sieve ${\cal R}_{\lambda}$, and suppose that Assumptions~\ref{ass:targetfunctiongen} hold. Then the conclusions hold in two stages, with respect to the projection and with respect to the target function.

{\bf (i) Asymptotic normality with respect to the projection $\Psi_{\lambda}$.}
The sieve TMLE $\Psi_{\lambda}(P_{n,\lambda}^*)$ is asymptotically linear with respect to the projection $\Psi_{\lambda}(P_0)$, uniformly in $x$,
\[
(n/n(\lambda))^{1/2}(\Psi_{\lambda}(P_{n,\lambda}^*)-\Psi_{\lambda}(P_0))(x)=
(n/n(\lambda))^{1/2}(P_n-P_0)D^*_{\Psi_{\lambda}(),x,P_0}+o_P(1),\]
so that, by Lemma~\ref{lem:variancegen} and the CLT, the right-hand side divided by $\sigma_{\lambda}(x)$ converges to $N(0,1)$ for each $x$.

{\bf (ii) Asymptotic normality with respect to the target function $\Psi(P_0)$.}
Selecting $n(\lambda)\sim^+ n^{1/(2k^*+1)}$ so that the bias $O^+(n(\lambda)^{-k^*})$ is $o((n(\lambda)/n)^{1/2})$, the same statement holds with respect to the true target function $\Psi(P_0)$:
\[
(n/n(\lambda))^{1/2}(\Psi_{\lambda}(P_{n,\lambda}^*)-\Psi(P_0))(x)=
(n/n(\lambda))^{1/2}(P_n-P_0)D^*_{\Psi_{\lambda}(),x,P_0}+r_n(x),\]
with $\sup_x\mid r_n(x)\mid=o_P(1)$, so that $(n/n(\lambda))^{1/2}\sigma_{\lambda}(x)^{-1}(\Psi_{\lambda}(P_{n,\lambda}^*)-\Psi(P_0))(x)\Rightarrow_d N(0,1)$ for each $x$. Moreover, since the bounds (A4) and (A5) are uniform in $x$, the process converges uniformly at rate $(n(\lambda)/n)^{1/2}\log n$, yielding simultaneous confidence bands of width $O(\log n\,(n(\lambda)/n)^{1/2})$.

\end{theorem}
\begin{proof}
We have:
\begin{eqnarray*}
\Psi_{\lambda}(P_{n,\lambda}^*)(x)-\Psi(P_0)(x)&=&\Psi_{\lambda}(P_{n,\lambda}^*)(x)-\Psi_{\lambda}(P_0)(x)+(\Psi_{\lambda}(P_0)-\Psi(P_0))(x)\\
&=&\Psi_{\lambda}(P_{n,\lambda}^*)(x)-\Psi_{\lambda}(P_0)(x)+O^+(n(\lambda)^{-k^*}),
\end{eqnarray*}
where, by (A3), the bias term has this rate uniformly in $x$: $\sup_{x}\mid \Psi_{\lambda}(P_0)-\Psi(P_0)\mid (x)=O^+(n(\lambda)^{-k^*})$.
We select the rate $n(\lambda)\sim^+ n^{1/(2k^*+1)}$ optimally so that $n(\lambda)^{-k^*}=o((n(\lambda)/n)^{1/2})$, making the bias term asymptotically negligible.

We have the exact expansion of the estimation term:
\begin{eqnarray*}
\Psi_{\lambda}(P_{n,\lambda}^*)(x)-\Psi_{\lambda}(P_0)(x)&=&
\alpha_{\lambda,P_{n,\lambda}^*}^{\top}{\bf \phi}_{\lambda}^*(x)-\alpha_{\lambda,P_0}^{\top}{\bf \phi}_{\lambda}^*(x)\\
&=&
(P_n-P_0)D^*_{\Psi_{\lambda}(),x,P_{n,\lambda}^*}+R_{\Psi_{\lambda},x}(P_{n,\lambda}^*,P_0),
\end{eqnarray*}
where the term $P_n D^*_{\Psi_{\lambda}(),x,P_{n,\lambda}^*}$ is absent since the TMLE step solves $P_n D^*_{\alpha_{\lambda}(),P_{n,\lambda}^*}=0$, so $P_n D^*_{\Psi_{\lambda}(),x,P_{n,\lambda}^*}=0$ for all $x$.
By (A4), $\sup_x\mid R_{\Psi_{\lambda},x}(P_{n,\lambda}^*,P_0)\mid=o_P((n(\lambda)/n)^{1/2})$.

Decomposing the empirical-process term,
\[
(P_n-P_0)D^*_{\Psi_{\lambda}(),x,P_{n,\lambda}^*}=(P_n-P_0)D^*_{\Psi_{\lambda}(),x,P_0}+R_{n,x},\]
with $R_{n,x}$ as in (A5). By (A5), $\sup_x\mid R_{n,x}\mid=o_P((n(\lambda)/n)^{1/2})$.

We therefore conclude, uniformly in $x$,
\[
\Psi_{\lambda}(P_{n,\lambda}^*)(x)-\Psi_{\lambda}(P_0)(x)=
(P_n-P_0)D^*_{\Psi_{\lambda}(),x,P_0}+o_P((n(\lambda)/n)^{1/2}),\]
which establishes part (i); adding the negligible bias term gives the same expansion with $\Psi_{\lambda}(P_0)$ replaced by $\Psi(P_0)$.

The leading term is a normalized sum of independent mean-zero random variables, so by Lemma~\ref{lem:variancegen} and the CLT, $(n/n(\lambda))^{1/2}\sigma_{\lambda}(x)^{-1}(\Psi_{\lambda}(P_{n,\lambda}^*)-\Psi(P_0))(x)\Rightarrow_d N(0,1)$.
Finally, since the bounds (A4) and (A5) are uniform in $x$, the empirical-mean approximation of $\Psi_{\lambda}(P_{n,\lambda}^*)-\Psi(P_0)$ by independent mean-zero random variables holds uniformly at rate $(n(\lambda)/n)^{1/2}\log n$, yielding simultaneous confidence bands of width $O(\log n\,(n(\lambda)/n)^{1/2})$. This establishes part (ii).

\end{proof}

\subsection{Sieve TMLE with a data-adaptive sieve}

Suppose now that the sieve ${\cal R}_{\lambda,n}$ of size $n(\lambda,n)$ is data dependent, and was based on outcome data, and let ${\cal R}_\lambda$ be the fixed companion sieve of size $n(\lambda)\leq n(\lambda,n)$ satisfying the $O^+(n(\lambda)^{-k^*})$ uniform approximation error (A1). The data-adaptive sieve TMLE $P_{n,\lambda}^*$ is constructed to solve the score equations of its own working model $D^{(k)}({\cal R}_{\lambda,n})$, $P_n D^*_{\alpha_{\lambda,n}(),P_{n,\lambda}^*}=0$, and hence $P_n D^*_{\Psi_{\lambda,n}(),x,P_{n,\lambda}^*}=0$ uniformly in $x$, where $\Psi_{\lambda,n}(P)=\alpha_{\lambda,n}(P)^{\top}{\bf \phi}_{\lambda,n}^*$ is the projection of $\Psi(P)$ onto $D^{(k)}({\cal R}_{\lambda,n})$. As in Section~\ref{Chtargetfunction4}, our strategy is nonetheless to analyze $P_{n,\lambda}^*$ as though it were a sieve TMLE for the fixed sieve ${\cal R}_\lambda$. This requires that $P_{n,\lambda}^*$ also solves the efficient influence curve equation of the fixed working model up to a negligible remainder, $P_n D^*_{\Psi_{\lambda}(),x,P_{n,\lambda}^*}=o_P((n(\lambda)/n)^{1/2})$, which we establish in the following lemma. Then, the asymptotic normality argument for a fixed sieve TMLE in Theorem~\ref{chtargetfunctiontmleIgen} carries over to the data-adaptive TMLE $P_{n,\lambda}^*$.

\begin{lemma}\label{lem:sieverategen}
Suppose ${\cal R}_{\lambda,n}$ is a data dependent sieve, and let ${\cal R}_{\lambda}$ be the fixed companion sieve of size $n(\lambda)\leq n(\lambda,n)$. Assume the data dependent sieve $D^{(k)}({\cal R}_{\lambda,n})$ satisfies the projection-approximation bound
\[
\sup_{\phi\in D^{(k)}_M([0,1]^{d})} \left\lvert \phi-\Pi(\phi\mid D^{(k)}({\cal R}_{\lambda,n}))\right\rvert_{\infty}=O_P^+(n(\lambda)^{-k^*}),\]
where $\Pi(\cdot\mid D^{(k)}({\cal R}_{\lambda,n}))$ denotes a projection onto the working model $D^{(k)}({\cal R}_{\lambda,n})$, together with the positivity, entropy, and bias-rate conditions of Lemma~\ref{lem:sieve rates}. Then, uniformly in $x$,
\[
P_n D^*_{\Psi_\lambda(),x,P_{n,\lambda}^*}= o_P((n(\lambda)/n)^{1/2}).\]
\end{lemma}
\begin{proof}
The TMLE $P_{n,\lambda}^*$ solves $P_n D^*_{\alpha_{\lambda,n}(),P_{n,\lambda}^*}=0$, hence $P_n D^*_{\Psi_{\lambda,n}(),x,P_{n,\lambda}^*}=0$ for all $x$, but the asymptotic normality argument requires the score equation of the fixed working model, which the TMLE does not target directly.
Since the canonical gradient $D^*_{\alpha_{\lambda,n}(),P}$ is linear in ${\bf \phi}_{\lambda,n}$, the score equation gives $P_n D^*_{\phi,P_{n,\lambda}^*}=0$ for every $\phi\in D^{(k)}({\cal R}_{\lambda,n})$, where $\phi\rightarrow D^*_{\phi,P}$ is the gradient in direction $\phi$ as in Lemma~\ref{lem:variancegen}. Recalling the projection $\Pi(\phi\mid D^{(k)}({\cal R}_{\lambda,n}))$, it follows that $P_n D^*_{\Pi(\phi\mid D^{(k)}({\cal R}_{\lambda,n})),P_{n,\lambda}^*}=0$ for all $\phi\in D^{(k)}([0,1]^{d})$.
For $j\in {\cal R}_\lambda$, let
\[
e_{n,j}\equiv \phi_j-\Pi(\phi_j\mid D^{(k)}({\cal R}_{\lambda,n})),\]
so that, writing $D^*_{\Psi_\lambda(),x,P}=\sum_{j\in {\cal R}_\lambda}\phi_j(x)D^*_{\phi_j,P}$,
\[
P_n D^*_{\Psi_\lambda(),x,P_{n,\lambda}^*}=\sum_{j\in {\cal R}_\lambda}\phi_j(x)P_n D^*_{\phi_j,P_{n,\lambda}^*}=\sum_{j\in {\cal R}_\lambda}\phi_j(x)P_n D^*_{e_{n,j},P_{n,\lambda}^*}.\]
We decompose $P_n D^*_{e_{n,j},P_{n,\lambda}^*}=(P_n-P_0)D^*_{e_{n,j},P_{n,\lambda}^*}+P_0 D^*_{e_{n,j},P_{n,\lambda}^*}$. The empirical process term divided by $n(\lambda)$ can be bounded in terms of the entropy integral of the class implied by the HAL-MLE $P_{n,\lambda}$, together with the rate $e_{n,j}=O_P^+(n(\lambda)^{-k^*})$; the term $P_0 D^*_{e_{n,j},P_{n,\lambda}^*}$ can be further written as $P_0 \{D^*_{e_{n,j},P_{n,\lambda}^*}-D^*_{e_{n,j},P_0}\}$, showing it can be bounded in terms of a product of $\pl e_{n,j}\pl$ and $\pl p_{n,\lambda}^*-p_0\pl$. By the positivity, entropy, and bias-rate conditions both terms are $o_P((n(\lambda)/n)^{1/2})$ uniformly in $x$.
See also our proof for the DRC example (Lemma~\ref{lem:sieve rates}) and Appendix~\ref{Chtargetfunction6} for an example of the causal effect of a binary treatment on survival outcomes conditional on covariates.
\end{proof}

\begin{theorem}\label{chtargetfunctiontmleIgenadaptive}
Suppose the conditions of Theorem~\ref{chtargetfunctiontmleIgen} hold for the fixed companion sieve $D^{(k)}({\cal R}_\lambda)$ of size $n(\lambda)\sim^+ n^{1/(2k^*+1)}$, and the conditions of Lemma~\ref{lem:sieverategen} hold for the data-adaptive sieve $D^{(k)}({\cal R}_{\lambda,n})$ of size $n(\lambda,n)\geq n(\lambda)$. Then $P_{n,\lambda}^*$ solves the fixed-sieve efficient influence-curve equation up to negligible order, $P_n D^*_{\Psi_{\lambda}(),x,P_{n,\lambda}^*}=o_P((n(\lambda)/n)^{1/2})$ uniformly in $x$, and the conclusion of Theorem~\ref{chtargetfunctiontmleIgen} carries over to the data-adaptive sieve TMLE in the same two stages, with respect to the projection and with respect to the target function.

{\bf (i) Asymptotic normality with respect to the projection $\Psi_{\lambda}$.}
Uniformly in $x$,
\[
(n/n(\lambda))^{1/2}(\Psi_{\lambda}(P_{n,\lambda}^*)-\Psi_{\lambda}(P_0))(x)=
(n/n(\lambda))^{1/2}(P_n-P_0)D^*_{\Psi_{\lambda}(),x,P_0}+o_P(1),\]
where $D^*_{\Psi_{\lambda}(),x,P_0}$ is the canonical gradient of the fixed-sieve projection $\Psi_{\lambda}$. Consequently, for each $x$,
\[
(n/n(\lambda))^{1/2}\sigma_{\lambda}(x)^{-1}(\Psi_{\lambda}(P_{n,\lambda}^*)-\Psi_{\lambda}(P_0))(x)\Rightarrow_d N(0,1),\]
with $\sigma^2_{\lambda}(x)=\mbox{VAR}_{P_0} D^*_{\Psi_{\lambda}(),x,P_0}/n(\lambda)$ as in Lemma~\ref{lem:variancegen}, and the one-dimensional class $\{D^*_{\Psi_{\lambda}(),x,P_0}:x\}$ yields simultaneous confidence bands of width $O(\log n\,(n(\lambda)/n)^{1/2})$.

{\bf (ii) Asymptotic normality with respect to the target function $\Psi(P_0)$.}
Undersmoothing $n(\lambda)\sim^+ n^{1/(2k^*+1)}$ by $\log n$-factors to render the approximation bias negligible, the same result holds for the data-adaptive plug-in estimator $\Psi_{\lambda,n}(P_{n,\lambda}^*)$ relative to the true target function: uniformly in $x$,
\[
(n/n(\lambda))^{1/2}(\Psi_{\lambda,n}(P_{n,\lambda}^*)-\Psi(P_0))(x)=
(n/n(\lambda))^{1/2}(P_n-P_0)D^*_{\Psi_{\lambda}(),x,P_0}+r_n(x),\]
with $\sup_x\mid r_n(x)\mid=o_P(1)$. Consequently, for each $x$,
\[
(n/n(\lambda))^{1/2}\sigma_{\lambda}(x)^{-1}(\Psi_{\lambda,n}(P_{n,\lambda}^*)-\Psi(P_0))(x)\Rightarrow_d N(0,1),\]
and the same one-dimensional class $\{D^*_{\Psi_{\lambda}(),x,P_0}:x\}$ yields simultaneous confidence bands of width $O(\log n\,(n(\lambda)/n)^{1/2})$.
\end{theorem}
\begin{proof}
Apply the exact first-order expansion of Theorem~\ref{chtargetfunctiontmleIgen} to the fixed-sieve projection $\Psi_{\lambda}$ evaluated at $P_{n,\lambda}^*$:
\[
\Psi_{\lambda}(P_{n,\lambda}^*)(x)-\Psi_{\lambda}(P_0)(x)=
(P_n-P_0)D^*_{\Psi_{\lambda}(),x,P_{n,\lambda}^*}
- P_n D^*_{\Psi_{\lambda}(),x,P_{n,\lambda}^*}
+R_{\Psi_{\lambda},x}(P_{n,\lambda}^*,P_0).\]
By Lemma~\ref{lem:sieverategen}, $P_n D^*_{\Psi_{\lambda}(),x,P_{n,\lambda}^*}=o_P((n(\lambda)/n)^{1/2})$ uniformly in $x$. By (A4) the remainder $R_{\Psi_{\lambda},x}(P_{n,\lambda}^*,P_0)$, and by (A5) the empirical-process term $(P_n-P_0)\{D^*_{\Psi_{\lambda}(),x,P_{n,\lambda}^*}-D^*_{\Psi_{\lambda}(),x,P_0}\}$, are $o_P((n(\lambda)/n)^{1/2})$ uniformly in $x$. Hence, uniformly in $x$,
\[
\Psi_{\lambda}(P_{n,\lambda}^*)(x)-\Psi_{\lambda}(P_0)(x)=
(P_n-P_0)D^*_{\Psi_{\lambda}(),x,P_0}+o_P((n(\lambda)/n)^{1/2}).\]
Since ${\cal R}_\lambda$ is non-random, $D^*_{\Psi_{\lambda}(),x,P_0}$ is a fixed function of $O$, so the leading term is a normalized sum of i.i.d.\ mean-zero variables and the CLT applies, giving asymptotic normality with respect to the projection $\Psi_{\lambda}$. This establishes part (i).
Finally, at $n(\lambda)\sim^+ n^{1/(2k^*+1)}$ we undersmooth by $\log n$-factors exactly as in Theorem~\ref{chtargetfunctiontmleIgen}, so that $\sup_x\mid \Psi_{\lambda}(P_0)-\Psi(P_0)\mid(x)$ and $\sup_x\mid \Psi_{\lambda,n}(P_{n,\lambda}^*)-\Psi_{\lambda}(P_{n,\lambda}^*)\mid(x)$ are $o_P((n(\lambda)/n)^{1/2})$ by (A3), giving the stated result with respect to $\Psi(P_0)$, with pointwise normality and simultaneous bands following exactly as in Theorem~\ref{chtargetfunctiontmleIgen}. This establishes part (ii).
\end{proof}

The application of Theorem~\ref{chtargetfunctiontmleIgenadaptive} to T-HAL-MLE, with data dependent sieve given by ${\cal R}_{\lambda,n}$ indexed by the $L_1$-norm used in the targeted LASSO step, yields the theorem for T-HAL-MLE.

\section{Simulation Study}\label{sec:simulation}

We evaluate the finite-sample performance of T-HAL-MLE through a Monte Carlo study that varies the data-generating process, treatment distribution, and sample size.
We compare three types of estimators of $\Psi(P)$:
\begin{enumerate}[leftmargin=2em]
 \item \textbf{HAL-MLE}: The plug-in estimator $\hat{\psi}^{\text{HAL}}(a)$ obtained by marginalizing the HAL outcome fit over the empirical distribution of $W$. No targeting step is applied. This estimator serves as our baseline for comparison.
 \item \textbf{T-HAL-MLE}: The proposed targeted estimator $\hat{\psi}^{\text{T-HAL}}(a)$, where the targeted outcomes are marginalized over the empirical distribution of $W$, and projected onto the basis functions chosen via the LASSO at the targeting stage. An algorithm for this estimator is given in Algorithm \ref{alg:thal-mle}.
 \item \textbf{T-HAL-MLE (Plugin)}: An alternative $\hat{\psi}^{\text{T-HAL,pl}}(a)$ that evaluates the targeted outcome regression directly at each test point, i.e. the plug-in marginal before projecting onto the chosen basis (see also Algorithm \ref{alg:thal-mle}). Since the working model selected by HAL can be rich, the plug-in marginal may already approximate the projection closely, we include this estimator to quantify how much the final projection step contributes in practice.
\end{enumerate}
As a benchmark we include the \textbf{Oracle Projection} $\hat{\psi}^{\text{oracle}}(a)$, the least-squares projection of the true DRC $\Psi(P)$ onto the HAL active basis selected in each replication.
By comparing the Oracle Projection to the truth we can investigate the basis approximation error, a reflection of how well the working model can approximate the true curve, reflecting the best the estimator could possibly do given the chosen basis.

\begin{algorithm}
  \caption{Targeted HAL-MLE (T-HAL-MLE) of the Dose Response Curve}
  \label{alg:thal-mle}
  \begin{algorithmic}[1]
    \REQUIRE $(O_i)_{i=1}^n = (W_i,A_i,Y_i)_{i=1}^n$, smoothness order $k_1$, basis dimension $N$, rate constants $c_1, c_2 > 0$
    \STATE \textbf{Initial estimate.} Fit $\bar{Q}_n$ via $k$-th order HAL-MLE (or Super Learner) of $\bar{Q}_0 = \EE(Y\mid A,W)$.
    \STATE \textbf{Propensity.} Fit $g_n(A\mid W)$ via HAL density estimation; estimate the marginal density $\omega(a)$ via kernel density estimation.
    \STATE \textbf{Large basis.} Form $N$-dimensional basis $\boldsymbol{\phi}_{1,N}(a)$ spanning $D^{(k_1)}(\mathcal{R}_{1,N})$. Compute $\Sigma = P_n\bigl[\boldsymbol{\phi}_{1,N}(A)\boldsymbol{\phi}_{1,N}^{\top}(A)\bigr]$, then construct clever covariate
      \[
        C_{g_n}(A,W) = \frac{\omega(A)}{g_n(A\mid W)}\,\Sigma^{-1}\boldsymbol{\phi}_{1,N}(A).
      \]
    \STATE \textbf{LASSO targeting.} With $\bar{Q}_n$ as offset, for each candidate $\lambda \geq 0$ solve
      \[
        \hat{\varepsilon}_n(\lambda) = \argmin_{\pl\varepsilon\pl_1\leq\lambda} P_n\,L\!\bigl(\bar{Q}_n + C_{g_n}\,\varepsilon\bigr).
      \]
      Set $\bar{Q}^*_{n,\lambda} = \bar{Q}_n + C_{g_n}\,\hat{\varepsilon}_n(\lambda)$ and record $n(\lambda) = \bigl|\{j:\hat{\varepsilon}_n(\lambda)(j)\neq 0\}\bigr|$.
        \STATE \textbf{Rate-optimal $\lambda$ selection.} Let $\lambda_{n,cv}$ be the cross-validated selector (minimiser of $V$-fold cross-validated risk), and find $\lambda^*_{n,cv}$:
    \[
    \lambda^*_{n,cv} = \begin{cases}
    \inf\left\{\lambda : n(\lambda) \leq \lceil c_1 \cdot n^{1/(2k_1^*+1)} \rceil\right\}
      & \text{if } n(\lambda_{n,cv}) < \lceil c_1 \cdot n^{1/(2k_1^*+1)} \rceil \\[6pt]
    \sup\left\{\lambda : n(\lambda) \geq \lfloor c_2 \cdot n^{1/(2k_1^*+1)} \rfloor\right\}
      & \text{if } n(\lambda_{n,cv}) > \lfloor c_2 \cdot n^{1/(2k_1^*+1)} \rfloor \\[6pt]
    \lambda_{n,cv}
      & \text{otherwise.}
    \end{cases}
    \]
    \STATE \textbf{Plug-in T-HAL-MLE.} At $\lambda^* = \lambda^*_{n,cv}$ let $\boldsymbol{\phi}_{1,\lambda^*}$ be the basis functions with $\hat{\varepsilon}_n(\lambda^*)(j)\neq 0$. Define the plug-in marginal
      \[
        \hat{\psi}^{\text{T-HAL,pl}}(a) = \frac{1}{n}\sum_{i=1}^n \bar{Q}^*_{n,\lambda^*}(a,W_i).
      \]
    \STATE \textbf{Projection onto selected working model.} Project $\hat{\psi}^{\text{T-HAL,pl}}$ onto the span of $\boldsymbol{\phi}_{1,\lambda^*}$ by solving
      \[
        \hat{\alpha}_n = \argmin_{\alpha}\; P_n\!\left[\bigl(\hat{\psi}^{\text{T-HAL,pl}}(A) - \alpha^{\top}\boldsymbol{\phi}_{1,\lambda^*}(A)\bigr)^2\right].
      \]
      This has the closed-form OLS solution $\hat{\alpha}_n = \bigl(P_n[\boldsymbol{\phi}_{1,\lambda^*}\boldsymbol{\phi}_{1,\lambda^*}^{\top}]\bigr)^{-1} P_n[\boldsymbol{\phi}_{1,\lambda^*}(A)\,\hat{\psi}^{\text{T-HAL,pl}}(A)]$.
    \STATE \textbf{T-HAL-MLE.} The final estimator is
      \[
        \hat{\psi}^{\mathrm{T\text{-}HAL}}(a) = \hat{\alpha}_n^{\top}\boldsymbol{\phi}_{1,\lambda^*}(a).
      \]
  \end{algorithmic}
\end{algorithm}
 
\subsection{Data-generating process}\label{sec:sim-design}

We consider three data-generating processes (DGPs), each specifying a joint distribution $P_0$ over $(W, A, Y) \in \mathbb{R}^3$ with $W \sim \mathcal{N}(0,1)$ and $Y = \bar{Q}_0(A, W) + \varepsilon$, $\varepsilon \sim \mathcal{N}(0,1)$. We restrict to a one-dimensional covariate $W$ that favours the HAL-MLE and is thus particularly suitable for benchmarking the T-HAL-MLE against the HAL-MLE under conditions that most favour the latter. 

In all DGPs the true conditional mean takes the form $\bar{Q}_0(a,w) = 3\sin(1.5a) + 0.5a + g(w) + (a-5)g(w)$, where the component $g(w)$ varies across DGPs and governs the complexity of the relationship between $W$ and $Y$:
\begin{enumerate}
\item \textbf{Single discontinuity:} $g(w) = 5\cdot\mathds{1}[w > 0]$, introducing a single step change in the dose-response curve at $w = 0$ and an interaction between $a$ and $w$.
\item \textbf{Multiple discontinuities:} $g(w) = 4\bigl(\mathds{1}[w > -1] - \mathds{1}[w > 0] + \mathds{1}[w > 1]\bigr)$, introducing three sign-alternating jumps.
\item \textbf{High frequency:} $g(w) = 3\sin(7\pi w)$, a smooth but rapidly oscillating function that makes estimation particularly challenging.
\end{enumerate}
Each DGP is paired with two treatment distributions: a normal distribution $A \mid W \sim \mathcal{N}\bigl(\mu = 5 + 1.5W,\, \sigma^2 = 2^2\bigr)$ truncated to $[0,10]$, and a uniform distribution $A \sim \mathrm{Uniform}(0,10)$, which provides equal coverage across the dose range, equivalent to randomization of treatment. For each replication, all estimators are evaluated at the same $m = 25$ equally-spaced dose test points $\{a_1, \ldots, a_{25}\}$ in $[1,9]$. We run $B = 1{,}000$ replications at $n \in \{500, 1000, 2000\}$, for a total of 18 simulation scenarios (3 DGPs $\times$ 2 treatment distributions $\times$ 3 sample sizes).

\subsection{Implementation}\label{sec:sim-implementation}

In the first two DGPs, we estimate the outcome model using HAL \citep{hal9001_manual} with zero-order spline basis functions due to the discontinuity of the outcome functions, while for the high-frequency DGP, the outcome is better approximated using first-order splines. The conditional treatment density $g_0(a \mid w)$ is estimated via the R package \texttt{haldensify} \citep{hejazi_haldensify_pkg_2022}. We stabilize the weights by multiplying by the marginal density of $A$, $\omega(a) = E_0[g_0(a \mid W)]$, which we find using kernal density estimation \citep{wand_kernel_1994, sheather_reliable_1991}. While our DGPs do not suffer from large inverse propensity weights (IPW) weights, we recommend following \citep{gruber_data-adaptive_2022} who suggest truncating IPW weights from above at $\sqrt{n}\log(n/5)$ in order to prevent extreme weights from dominating the targeting step and leading to bias in the estimator.

Since $\psi_0$ is smooth in $a$, we assume $\psi_0 \in D^{(1)}_M([0,10])$, i.e., the true DRC belongs to the first-order HAL function class with $k_1 = 1$. Following the rate-optimal selector described in Section~\ref{par:rate-optimal-selector}, we choose $\lambda^*_{n,cv}$ such that we have a minimum working model size of at least $\lceil 6 \cdot n^{1/5}\rceil$ and a maximum of $\lfloor 9 \cdot n^{1/5} \rfloor$. The constants $c_1 = 6$ and $c_2 = 9$ were selected by evaluating a grid of candidate values in a separate set of simulations and choosing those that minimized RMSE across settings.

We also considered a fixed selector that chooses the smallest $\lambda$ such that $n(\lambda) \geq \lceil c \cdot n^{1/5} \rceil$ with $c = 6$. To assess the two approaches at the oracle level, Table~\ref{tab:oracleproj} reports bias and RMSE of the oracle projection onto the working model selected by each method, averaged over dose points. The fixed selector yields consistently higher RMSE across all DGPs and sample sizes, with differences ranging from roughly 0.02-0.04, while bias differences are negligible, indicating the benefit of allowing the cross-validated $L_1$-norm $\lambda_{n,cv}$ to choose an appropriate number of basis functions most relevant for the targeting, as long as it selects enough basis functions. Given that the adaptive selector outperformed consistently at the oracle level, we adopt it as the default and do not include the fixed selector among the main comparisons.

One could consider replacing the HAL-based adaptive selector with a prespecified fixed knot basis spanning the support of $A$, which would be a natural choice in the univariate setting considered here and thus avoiding the use of LASSO for selecting basis functions. However, such an approach does not scale to high dimensional multivariate functions, where specifying a reasonable fixed basis becomes statistically and computationally prohibitive. Since a key motivation for the HAL-based approach is its applicability beyond univariate target functionals, we do not include a fixed-knot alternative in our comparisons.

We estimate the variance of $\hat{\psi}^{\text{HAL}}(a)$ and $\hat{\psi}^{\text{T-HAL}}(a)$ using their respective canonical gradients (\citep{shi2025halbasedpluginestimationpointwise} and Equation~\eqref{canonicalgradient}) and the Delta Method. For the $\hat{\psi}^{\text{T-HAL}}(a)$ in particular, we cross-fit the outcome model using $K=2$ folds to ensure that the outcomes $Y$ are independent of the data used to fit $\bar{Q}_n$, which creates true residuals when constructing the canonical gradient. Wald-type confidence intervals are then constructed as $\hat \psi (a_j)\pm 1.96\times \cdot \hat{\sigma}(a_j)$, where $\hat \sigma(a_j) = n^{-1/2}\bigl(P_n\{D^*_{\Psi_\lambda(),a_j,P_{n,\lambda}^*}\}^2\bigr)^{1/2}$ is the estimated standard error at $a_j$. In practice, we invert $\Sigma_{\phi_{1,\lambda}}$ using the Moore-Penrose pseudo-inverse \citep{Penrose_1955}, as the Gram matrix can be rank-deficient when the HAL active set contains many nearly collinear basis functions. This is a recognized computational limitation of HAL \citep{li2025regularizedtmle}.

\subsection{Evaluation Metrics}
We evaluate the estimators on the accuracy of the point estimates, the size of any bias relative to the estimator's sampling variability, and the validity of the resulting confidence intervals. All quantities are computed pointwise at each of the $J=25$ dose evaluation points, for each estimator, DGP, sample size, and treatment distribution, and summarised across the $B$ replications.

We report the pointwise RMSE,
\[
\mathrm{RMSE}(a_j) = \left\{\frac{1}{B}\sum_{b=1}^{B} \bigl[\hat{\psi}^{(b)}(a_j) - \psi_0(a_j)\bigr]^2\right\}^{1/2},
\]
and the pointwise bias,
\[ B(a_j) = \frac{1}{B}\sum_{b=1}^{B} \bigl[\hat{\psi}^{(b)}(a_j) - \psi_0(a_j)\bigr]. \] 

To assess the potential for valid statistical inference, we examine the ratio of pointwise bias to Monte Carlo standard error (MC-SE), $B(a_j)/\hat{\sigma}_{\text{MC}}(a_j)$ where $\hat{\sigma}_{\text{MC}}(a_j)$ is the Monte Carlo standard deviation across replications. Valid Wald inference requires the bias to be negligible relative to the sampling variability of the estimator. Test points with a ratio within $\pm 1/\log(n)$ are considered to have negligible bias relative to the variability of the estimator.

To assess the validity of the interval estimates we report, at each dose point, the empirical coverage and mean width of two types of $95\%$ confidence interval. The first uses the Monte Carlo standard deviation in place of an estimated standard error, with empirical coverage
\[ \frac{1}{B}\sum_{b=1}^{B} \mathds{1}\!\left[\psi_0(a_j) \in \hat{\psi}^{(b)}(a_j) \pm 1.96\,\hat{\sigma}_{\text{MC}}(a_j)\right], \]
and mean width $\frac{1}{B}\sum_{b=1}^{B} 3.92\,\hat{\sigma}_{\text{MC}}(a_j)$. Because they are constructed from the Monte Carlo standard deviation, these MC-SE intervals isolate the sampling behaviour of the estimator itself from the quality of any variance estimator. The second is the Wald-type interval, with empirical coverage
\[ \frac{1}{B}\sum_{b=1}^{B} \mathds{1} \left[\psi_0(a_j) \in \hat{\psi}^{(b)}(a_j) \pm 1.96\,\hat{\sigma}^{(b)}(a_j)\right], \]
and mean width $\frac{1}{B}\sum_{b=1}^{B} 3.92\,\hat{\sigma}^{(b)}(a_j)$, constructed using the standard error $\hat{\sigma}^{(b)}(a_j)$ estimated from the canonical gradient within each replication. The T-HAL-MLE (Plugin) is excluded from the Wald comparison as it does not constitute a valid TMLE in the sense of projecting onto a finite-dimensional working model, and no valid variance estimator is proposed for it.

\subsection{Results}

Figure~\ref{fig:bc_rmse} presents boxplots of the pointwise RMSE.
In the single and multiple discontinuity settings, the T-HAL-MLE had smaller RMSE than the T-HAL-MLE (Plugin), which in turn had smaller RMSE than the HAL-MLE.
All estimators perform worse in the high-frequency DGP, which represents the most challenging setting due to the rapid oscillation of the outcome model in $W$, with HAL-MLE achieving the largest RMSE. In this particular setting, the two T-HAL-MLE methods performed comparably. In all settings, RMSE decreases consistently with increasing sample size, as expected. The oracle projection has the smallest RMSE, indicating that the chosen working model is able to approximate the truth well.

Figure~\ref{fig:bc_bias} displays the pointwise bias, 
across the same scenarios. Under the uniform treatment distribution, the HAL-MLE exhibits bias in the single and multiple discontinuity settings that persists across sample sizes, while the T-HAL-MLE (Plugin) shows a smaller but non-negligible bias. The T-HAL-MLE remains closest to zero throughout, demonstrating that the TMLE targeting step and projection successfully reduces bias. In the single discontinuity setting with normal treatment distribution, the HAL-MLE exhibits considerably wider spread and larger outliers than the two TMLE variants, which remain well-centred near zero. All three estimators exhibit substantial bias in the multiple discontinuities setting under normal treatment distribution, with the HAL-MLE and T-HAL-MLE (Plugin) most severely affected. Even the T-HAL-MLE shows non-negligible bias at $n=500$ and $n=1000$ in this setting, though it remains the best-performing estimator. Notably, the HAL-MLE bias in this setting increases at $n=2000$. Across all other settings, bias decreases with increasing sample size for all estimators, as expected. The bias of the oracle projection is the smallest of all the estimators and is centered around zero, indicating that the chosen working model is able to approximate the true curve without bias. 

\begin{figure}
 \centering
 \includegraphics[height=0.4\textheight]{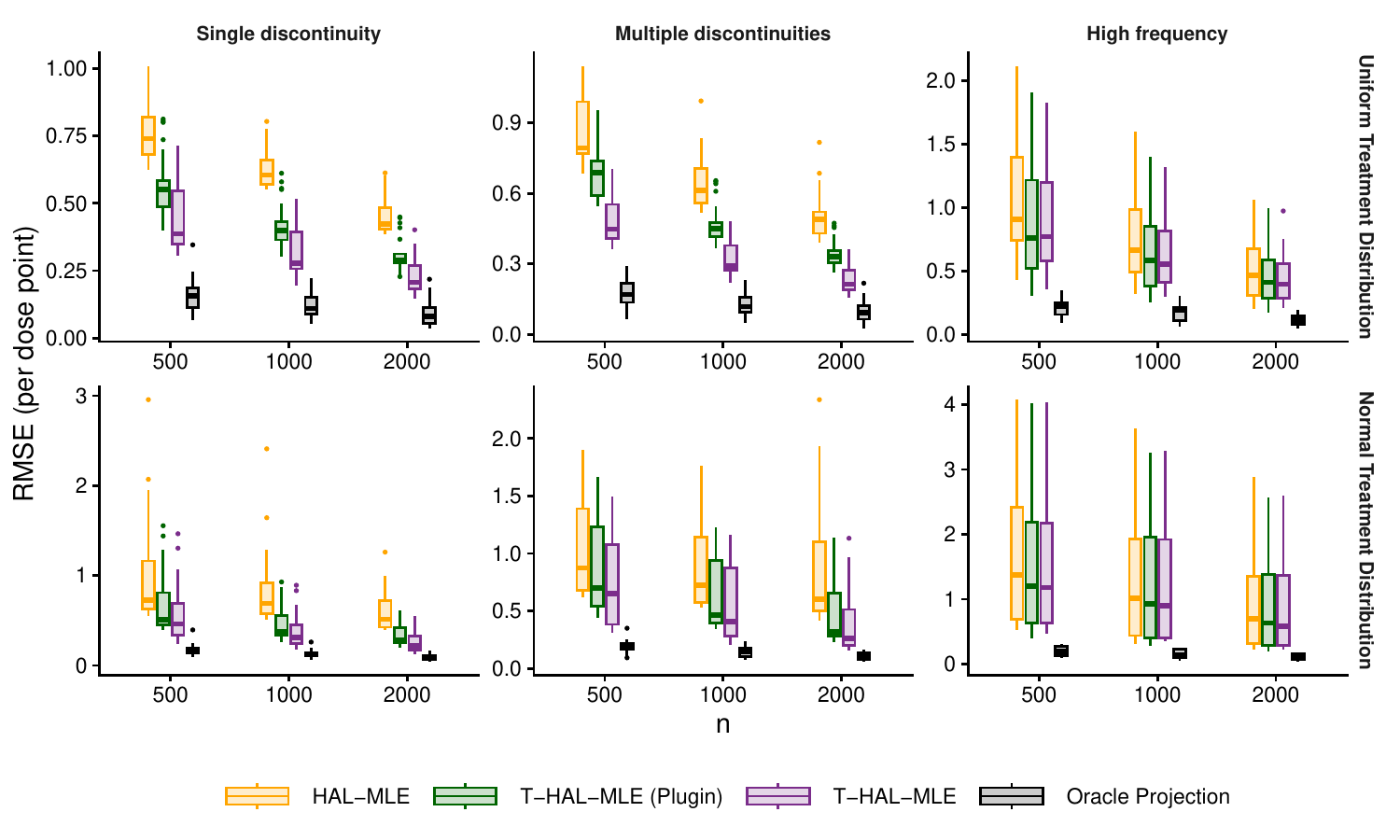}
 \caption{Boxplots of the RMSE across the $m=25$ dose evaluation points $a_j\in[1,9]$, for each estimator, DGP, sample size $n\in\{500,1000,2000\}$, and treatment distribution.
 Each observation in a boxplot is the average RMSE at one test point.}
 \label{fig:bc_rmse}
\end{figure}

\begin{figure}
 \centering
 \includegraphics[height=0.4\textheight]{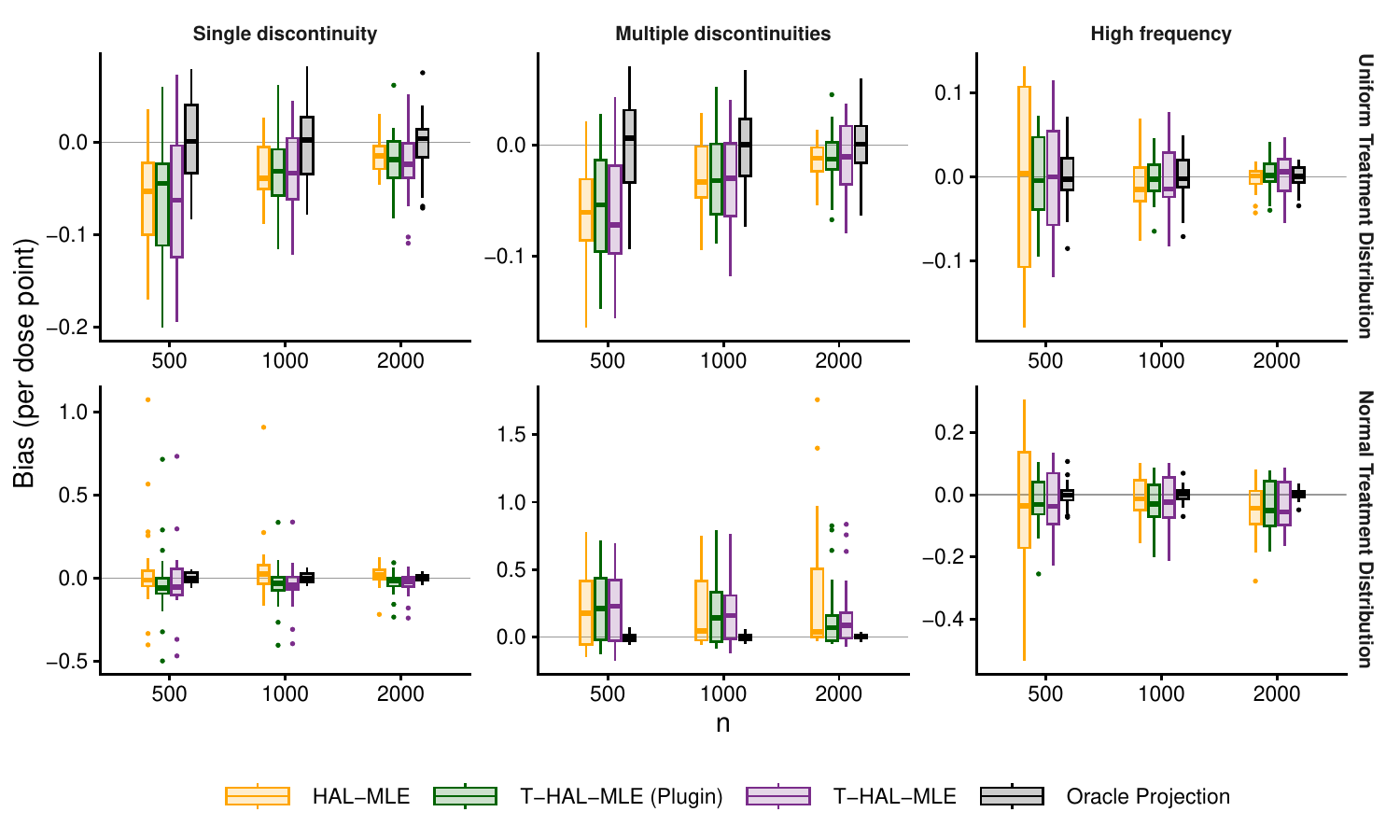}
 \caption{Boxplots of the bias across the $m=25$ dose evaluation points $a_j\in[1,9]$, for each estimator, DGP, sample size $n\in\{500,1000,2000\}$, and treatment distribution.
 Each observation in a boxplot is the average bias at one test point.}
 \label{fig:bc_bias}
\end{figure}

We present the ratio of pointwise bias to Monte Carlo standard error (MC-SE) in Figure~\ref{fig:bc_ratio}. Under the uniform treatment distribution, the majority of test points fall within this band for all estimators and DGPs, with ratios shrinking toward zero as $n$ increases. The normal treatment distribution is considerably more challenging, particularly in the multiple discontinuities DGP where the T-HAL-MLE (Plugin) exhibits large positive bias-to-MC-SE ratios that persist even at $n=2000$. The T-HAL-MLE also exhibits wider spread in the bias-to-MC-SE ratio in some settings, most notably the multiple discontinuities DGP with normal treatment distribution. The greater efficiency of the T-HAL-MLE results in a smaller MC-SE, so that even modest residual bias at individual dose points produces a larger ratio. Nonetheless, the spread of the T-HAL-MLE ratios decreases with $n$ across all settings, consistent with the bias becoming negligible relative to the estimator's variability at larger sample sizes.

We assess coverage of MC-SE confidence intervals in Figure~\ref{fig:bc_oracle_cov}. Under the uniform treatment distribution, all three estimators achieve close to nominal 95\% coverage across all DGPs and sample sizes, with the exception of slight overcoverage for the HAL-MLE in the high-frequency setting. Under the normal treatment distribution, the single discontinuity and high-frequency settings yield broadly nominal coverage for all estimators. The multiple discontinuities DGP is again the most challenging: both the T-HAL-MLE (Plugin) and T-HAL-MLE exhibit substantial undercoverage, with some dose points falling as low as 75--80\%, and this does not fully resolve at $n=2000$. We investigate this specific DGP further in Figure \ref{fig:oracle_n2000} and find that the outliers with undercoverage occur at the tails of the treatment where data is more sparse due to the normal treatment distribution.
Figure~\ref{fig:bc_oracle_width} shows the width of the MC-SE intervals. The HAL-MLE produces the widest intervals across all DGPs and treatment distributions, while the T-HAL-MLE produces the narrowest, with the T-HAL-MLE (Plugin) lying between the two. This ordering is consistent across all settings and sample sizes, and all three estimators exhibit decreasing interval width with increasing $n$.

\begin{figure}
 \centering
 \includegraphics[height=0.4\textheight]{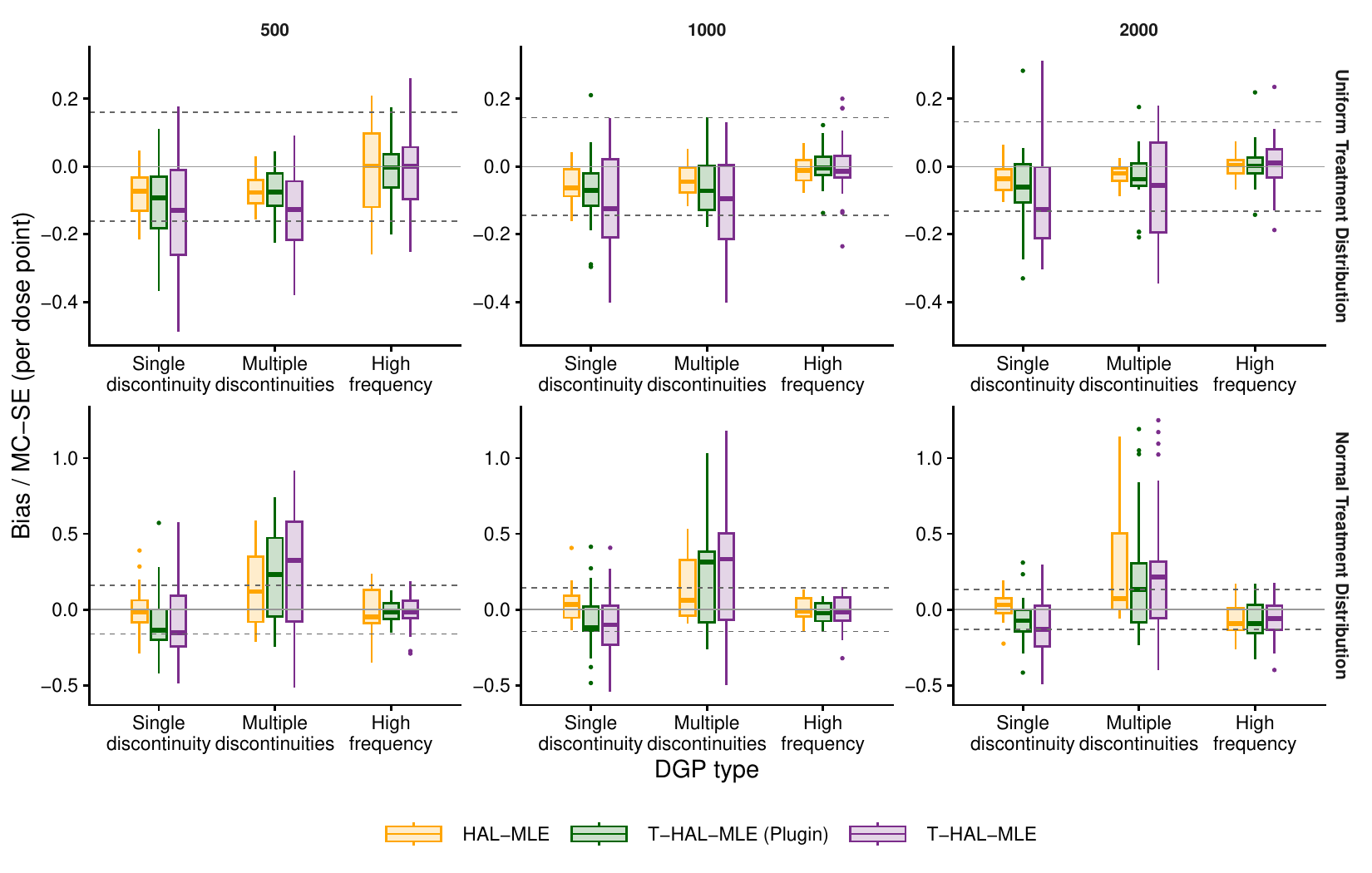}
 \caption{Boxplots of the Bias/MC-SE ratio across the $m=25$ dose evaluation points, for each estimator, DGP, sample size, and treatment distribution.
 Each observation is the ratio at one test point.
 Dashed lines mark $\pm 1/\log(n)$.}
 \label{fig:bc_ratio}
\end{figure}

\begin{figure}
 \centering
 \includegraphics[height=0.4\textheight]{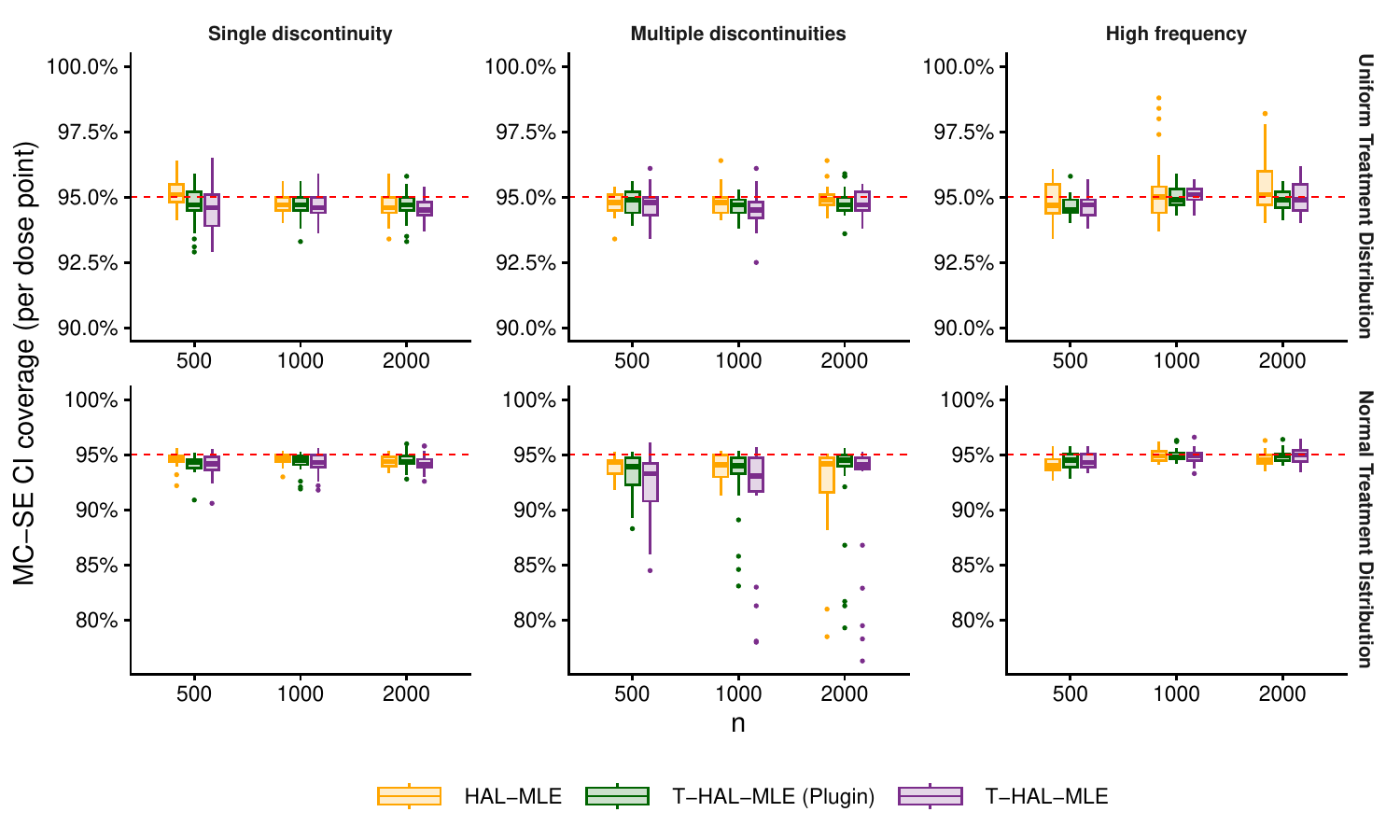}
 \caption{Boxplots of the MC-SE coverage of intervals, using the Monte Carlo standard deviation, across the $m=25$ dose evaluation points, for each estimator, DGP, sample size, and treatment distribution.
 Each observation is the empirical coverage at one test point.
 The dashed horizontal line marks the 95\% nominal level.}
 \label{fig:bc_oracle_cov}
\end{figure}

\begin{figure}
    \centering
    \includegraphics[height=0.4\textheight]{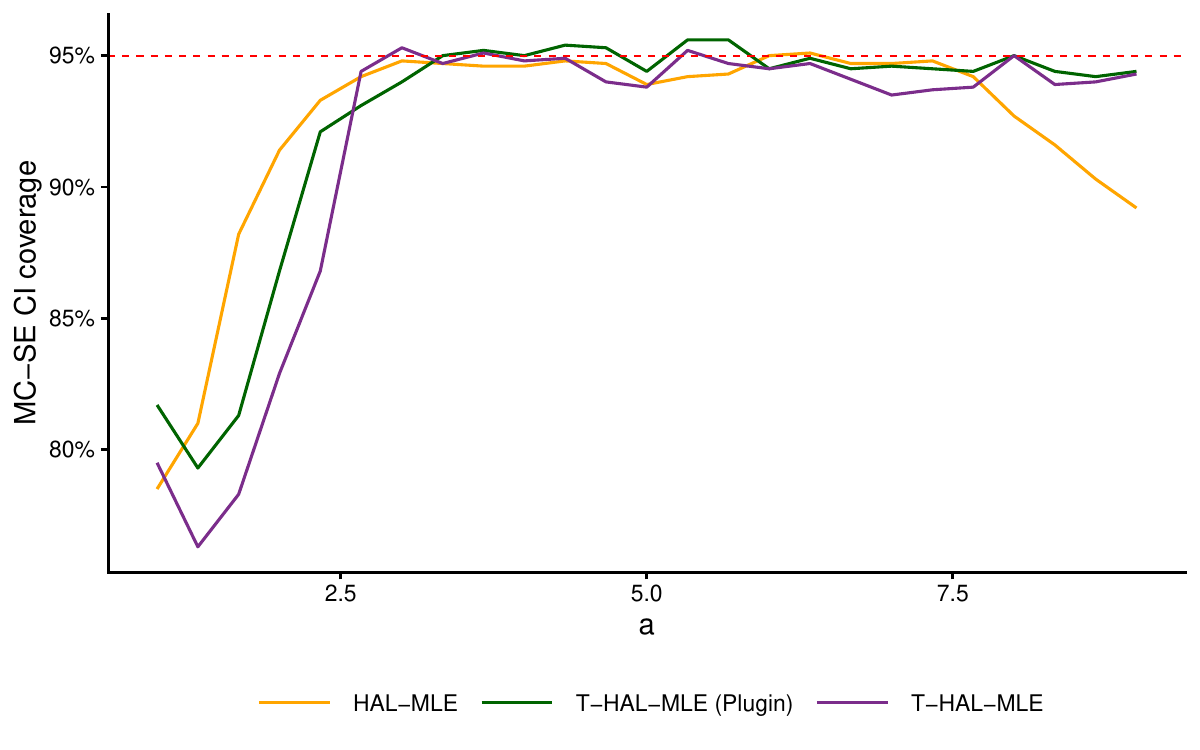}
    \caption{Oracle coverage across test points $a$ for the Multiple Discontinuity DGP and normal treatment distribution with $n=2000$.}
    \label{fig:oracle_n2000}
\end{figure}

\begin{figure}
 \centering
 \includegraphics[height=0.4\textheight]{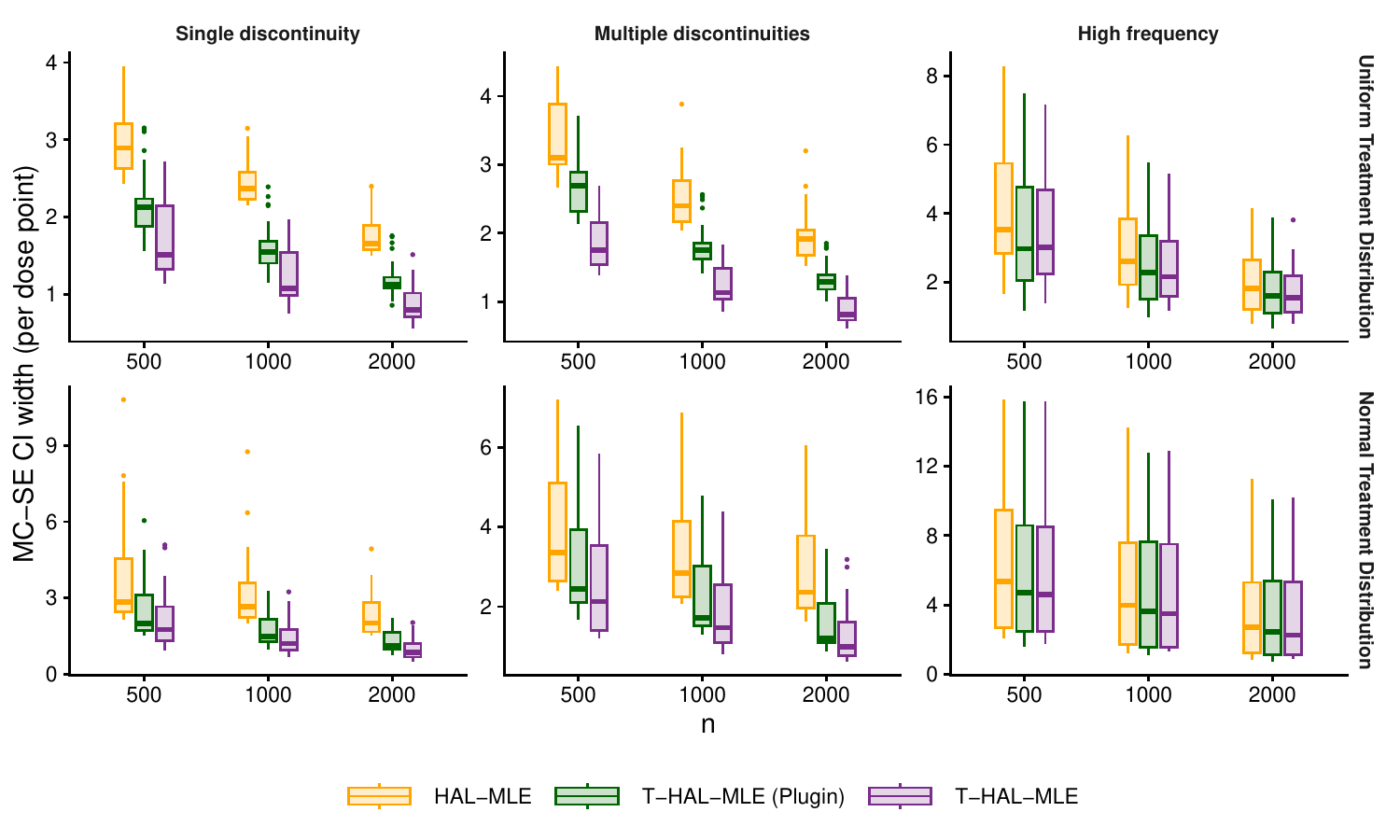}
 \caption{Boxplots of the width of MC-SE intervals, using the Monte Carlo standard deviation, across the $m=25$ dose evaluation points, for each estimator, DGP, sample size, and treatment distribution.
 Each observation is the mean width at one test point.
 }
 \label{fig:bc_oracle_width}
\end{figure}

Figure~\ref{fig:bc_ci_cov} presents coverage of the Wald-type confidence intervals. Under the uniform treatment distribution, the T-HAL-MLE approaches nominal 95\% coverage across all DGPs, with coverage improving toward the nominal level as $n$ increases. The HAL-MLE substantially undercovers in all settings, with coverage starting around 75--80\% at $n=500$ and improving only modestly with $n$. Under the normal treatment distribution, both estimators exhibit considerable undercoverage across all DGPs, with the HAL-MLE most severely affected.

Figure~\ref{fig:bc_ci_width} shows the width of the Wald-type intervals. Although the HAL-MLE undercovers, it still produces the widest confidence intervals, while the T-HAL-MLE produces narrower intervals than the HAL-MLE, even though both methods rely on the same underlying outcome model.

\subsection{Summary of findings}

Across the accuracy metrics, the T-HAL-MLE is consistently the best-performing estimator, achieving the smallest RMSE and the bias closest to zero. Its advantage over the T-HAL-MLE (Plugin) and the HAL-MLE demonstrates that the TMLE targeting step and projection successfully reduces the bias of the plug-in estimator. The oracle projection attains the smallest RMSE overall and a bias centred around zero, indicating that the chosen working model is able to approximate the true curve well and without meaningful bias.

The multiple discontinuities DGP under the normal treatment distribution is the most challenging setting throughout. The residual bias of the T-HAL-MLE there is concentrated at the tails of the treatment distribution, where the data are sparse (Figure~\ref{fig:oracle_n2000}), and the resulting undercoverage of the MC-SE intervals directly reflects that the bias remains non-negligible relative to the estimator's variability at those dose points, consistent with the bias-to-MC-SE ratios in Figure~\ref{fig:bc_ratio}. The greater efficiency of the T-HAL-MLE also contributes: its smaller MC-SE means that even modest residual bias at individual dose points produces a larger ratio, and hence the wider spread seen in Figure~\ref{fig:bc_ratio}. The HAL-MLE is comparatively less affected in this setting, though this is largely attributable to its larger variance producing wider intervals that are more likely to capture the true value despite the bias.

When the true Monte Carlo variance is used, the T-HAL-MLE achieves close to nominal coverage across almost all settings, showing that the estimator itself is well calibrated; the undercoverage of its Wald-type intervals therefore reflects the performance of the variance estimator rather than the estimator itself. The HAL-MLE produces the widest intervals because its variance estimator reflects the full complexity of the fitted outcome model, which includes a large number of basis functions, whereas the T-HAL-MLE targets a finite-dimensional working model that is considerably lower dimensional than the fitted outcome model, and as such its variance estimator reflects only the complexity of this lower dimensional working model. The resulting ordering, with the T-HAL-MLE narrowest and the HAL-MLE widest and the T-HAL-MLE (Plugin) in between, is consistent across all settings and sample sizes. This reduction in variance comes at the cost of some approximation error introduced by the working model projection, but this approximation error does not translate into meaningful undercoverage, as the MC-SE coverage results confirm. Alternative variance estimation approaches, such as the targeted bootstrap, may yield improved coverage.

\begin{figure}
 \centering
 \includegraphics[height=0.4\textheight]{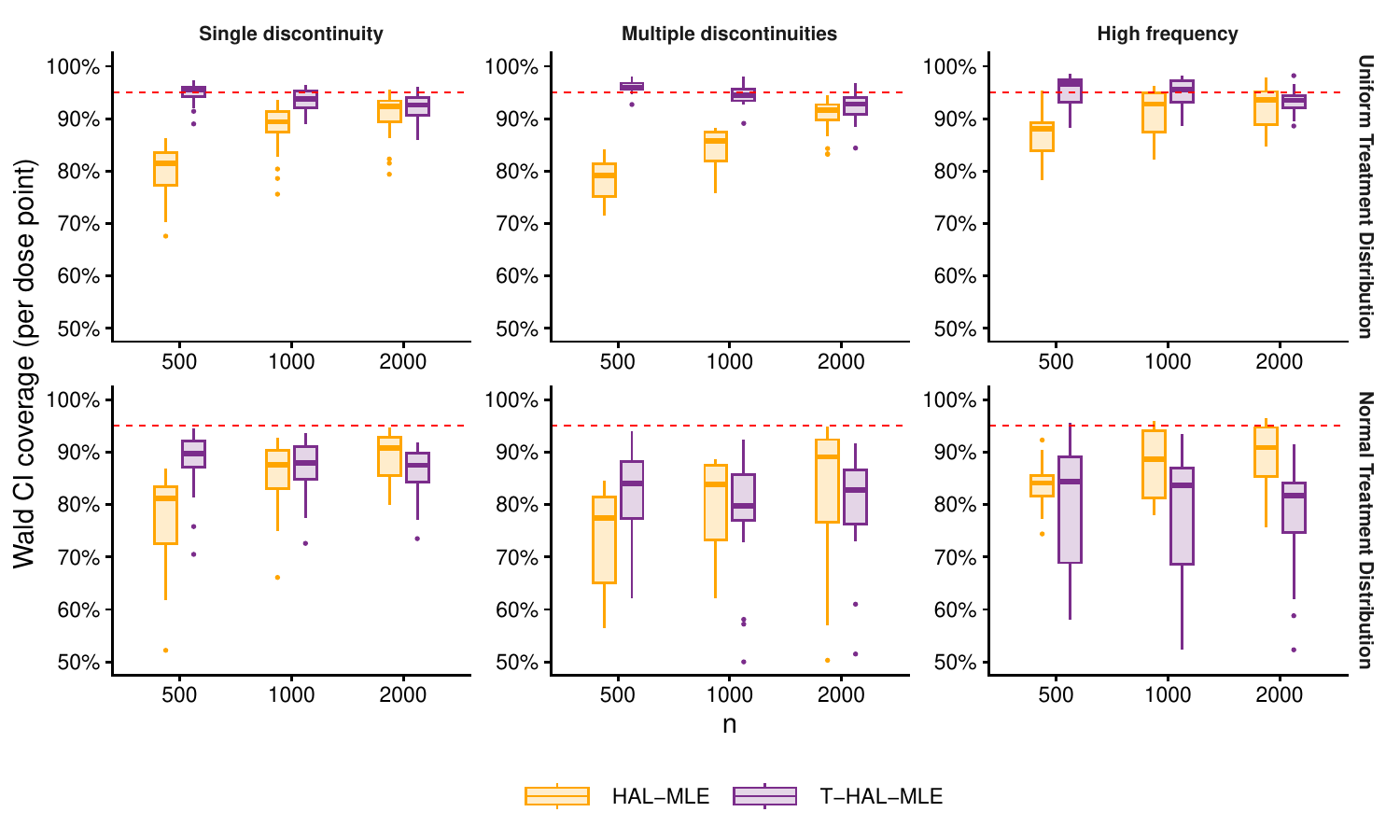}
 \caption{Boxplots of the Wald coverage across the $m=25$ dose evaluation points, for each estimator, DGP, sample size, and treatment distribution.
 Each observation is the empirical coverage at one test point.
 The dashed horizontal line marks the 95\% nominal level.}
 \label{fig:bc_ci_cov}
\end{figure}

\begin{figure}
 \centering
 \includegraphics[height=0.4\textheight]{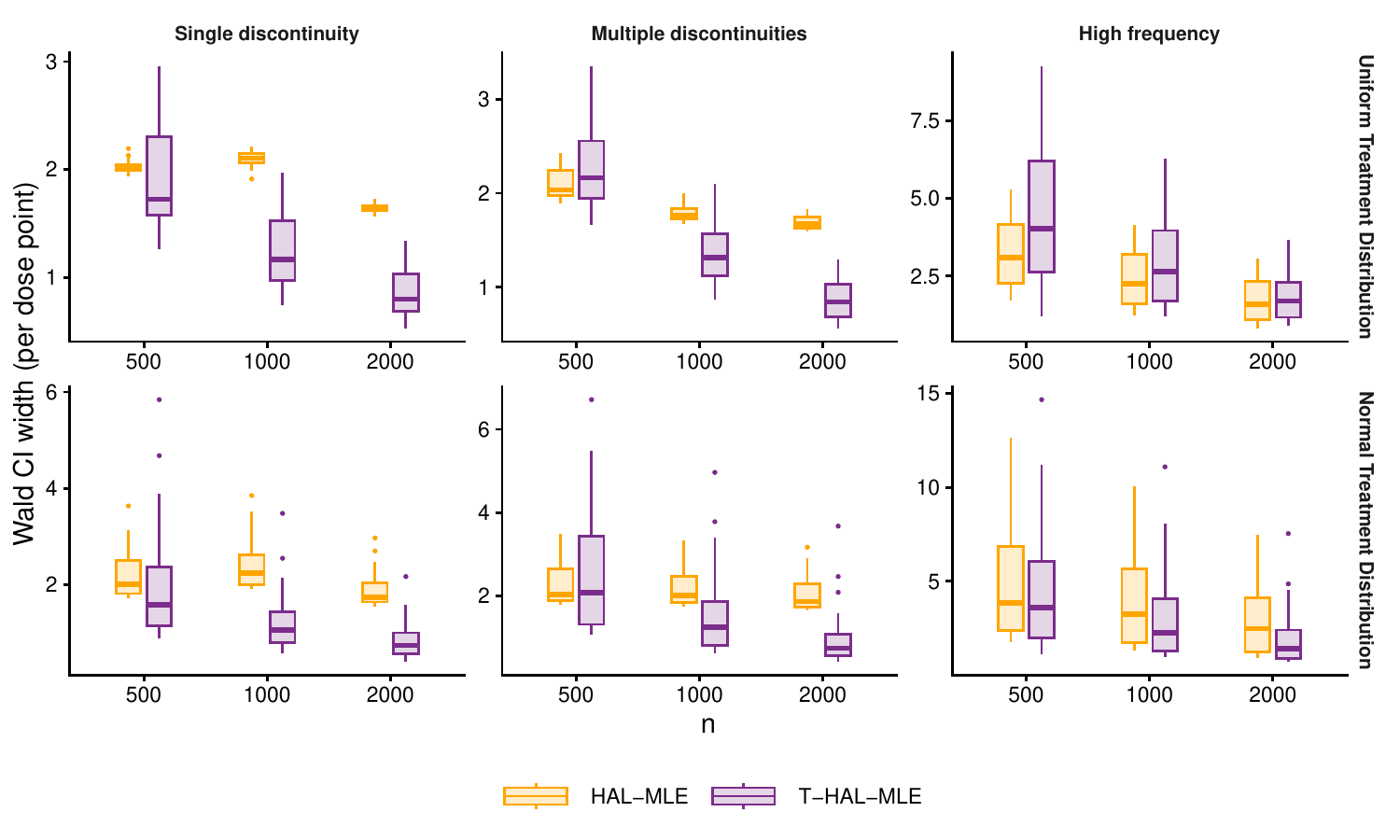}
 \caption{Boxplots of the average Wald-type CI width across the $m=25$ dose evaluation points, for each estimator, DGP, sample size, and treatment distribution. Each observation is the average interval width at one test point.}
 \label{fig:bc_ci_width}
\end{figure}

\section{Discussion}\label{Chtargetfunction7}
In this article we first demonstrated that plug-in HAL-MLE of lower dimensional target functions is sub-optimal, in spite of its inheritance of all the properties of the HAL plugged in. The issue is that its bias is driven by the bias of the HAL-MLE itself, and it fundamentally relies on undersmoothing to remove bias for the plug-in HAL-MLE. However, undersmoothing is ineffective in achieving its goal and comes with deteriorating performance of HAL itself. Theoretically, the $\log n$-factors in the rate of this plug-in will be driven by the dimension ${\bar{d}}$ of the $Q_0$-function, while we can obtain HAL-rates for the lower dimensional $D^{(k)}([0,1]^{d})$ of the target function. 
In fact, beyond losing $\log n$-factors in the rate, the plug-in HAL can fail to achieve the desired optimal rate for the smoothness and dimension of the target function space when the HAL-MLE estimates a less smooth and higher dimensional function: this can easily happen since the target function is obtained from the $Q_0$ function by integrating most of the variables. This motivated applying TMLE to estimate projections of the target function onto a HAL-sieve in target function space. In essence we are replacing the global undersmoothing in HAL-MLE by a targeted undersmoothing using the targeting step of the TMLE, where the targeting step preserves the smoothness of the initial $Q_n$. 

There are various important properties of this TMLE:
\begin{itemize}
\item The bias is driven by the approximation error of the working model in the lower dimensional target function space.
\item The TMLE step corresponds with undersmoothing the initial density estimator but now it selects the right basis functions that matter for bias reduction of the plug-in estimator.
\item The LASSO version of the TMLE step applied to a large working model can be used to data adaptively learn a working model.
\item The LASSO version of the TMLE preserves the sectional variation norm of the initial estimator $Q_n$, thereby fully controlling the conditions of the asymptotic normality theorem, while it is still effective in solving the key score equations.
\item We obtain influence curve based inference based on pointwise asymptotic normality and can also obtain simultaneous confidence bands based on the influence curve across a collection of points.
\end{itemize}

Our framework is general in the sense that the HAL-based sieve used to approximate the target function could in principle be replaced by any other sieve. However, the dimension-free convergence rates established for HAL are a key theoretical advantage, and most alternative approximation schemes do not attain the rates we present in this paper, making such substitutions theoretically unattractive without further justification.

A limitation of this TMLE-based approach is that, relative to a plug-in HAL-MLE, it will be more sensitive to inverse weighting by $1/g(A\mid W)$, which might hurt its finite sample robustness. If need be, one can use A-TMLE \citep{laan_adaptive-tmle_2025} to further stabilize the TMLE by developing the TMLE presented here for data adaptively selected models ${\cal M}_n$ for $P_0$ or, equivalently, $D^{(k)}({\cal R}_n({\bar{d}}))\subset D^{(k)}({\cal R}({\bar{d}}))$ for $\bar{Q}_0$, limiting the influence of extreme weights due to positivity violations.

A question remains on the impact of the choice of projection or equivalently the choice of inner product $\langle\cdot,\cdot\rangle_{\psi}$.
Is there a projection on the working model that would result in the most stable, least variable, TMLE among all projection specific TMLEs?
This will be addressed in a next article in which we define the projection in model space by projecting the data distribution onto a submodel ${\cal M}(\lambda)=\{P\in {\cal M}:\Psi(P)\in D^{(k)}({\cal R}_{\lambda,n})\}$ for which the target function values are all in our working model $D^{(k)}({\cal R}_{\lambda,n})$.
We will show that this represents a norm free projection, and might therefore be called a canonical projection implied by the statistical model and the working model. In addition, we will show that it will not affect the bias in a meaningful way, while it substantially decreases the variance of the TMLE, especially in situations in which there are practical positivity violations for $\Psi(P)$.

\textbf{Research funding:}
V.R. acknowledges the receipt of studentship awards from the Health Data Research UK-The Alan Turing Institute Wellcome PhD Programme in Health Data Science (Grant Ref: 218529/Z/19/Z). 
M.J.L. was supported by Novo Nordisk under the project gift Joint Initiative Causal Inference (JICI-48853).

\bibliography{references_used}


\processdelayedfloats

\appendix

\label{Chtargetfunctionappendix}

\section*{Appendix}

\section{Robust version of $\Psi(P)$ to deal with positivity issues}

We can often associate with $\Psi(P)(x)$ a positivity assumption. One could restrict the support of $P$ to a set $A(\delta,x)$ so that $\Psi(P_{\delta})(x)$ is well supported and the positivity assumption is controlled by a $\delta>0$. Let $P_{\delta,x}$ be the conditional distribution of $O$, given $O\in A(\delta,x)$. We could then create a robust version $\Psi_{\delta}(P)$ defined by $\Psi_{\delta}(P)(x)=\Psi(P_{\delta,x})(x)$. This could correspond with a $\Psi^F(Q_{\delta,x})(x)$. 
For example, consider $\Psi(P)(a)=E_W Q_P(a,W)$. This relies on $g_P(a\mid w)>0$ for $P_W$ almost every $w$. Therefore, we could define a set $A(\delta,a)=\{w: g(a\mid w)>\delta\}$ and define the $\delta$-robust version $\Psi_{\delta}(P)(a)=\int_{w\in A(\delta,a)} Q_P(a,w)dP(w)/P_W(A(\delta,a))$.
We could also define a robust version that uses extrapolation as follows. 
Let $P_{\delta}$ be the conditional distribution of $(W,A,Y)$, given $g(A\mid W)>\delta$. 
Define $Q_{\delta,P}=\arg\min_{Q\in D^{(k)}({\cal R}({\bar{d}})} P_{\delta}L(Q)$ and $\Psi_{\delta}(P)(a)=\int_w Q_{\delta,P}(a,w)dP(w)$. Note that we now still average over all $w$, but $Q_{\delta,P}(a,w)=\sum_{j\in {\cal R}({\bar{d}})}\beta_{\delta,P}(j)\phi_j(a,w)$ is always defined, even for $(a,w)$ with $g(a\mid w)<\delta$ or $g(a\mid w)=0$. 
In this manner, we can define parameters $\Psi_{\delta}(P)$ that are not suffering from positivity, but nonetheless, are not pathwise diffferentiable so that we can apply our general TMLE approach to this target function.

\section{Reducing analysis of plug-in estimators to linear functionals in a single $\bar{Q}_P$ optimizing a risk function.}\label{Apptargetfunction2}


In this section we argue that the analysis of plug-in estimators $\Phi(\bar{Q}_j(P): j=1,\ldots,J)$ involving multiple functional parameters $\bar{Q}_j(P)$ defined as minimizer of a risk function can be decomposed in the analysis of linear parameters $\Phi_{j,n}(\bar{Q}_{j,P})$ , $j=1,\ldots,J$, up till second order remainders that can be separately shown to be negligible. In addition, this still applies when some of the functional parameters are not variation independent. 
This appears to show that one could focus on constructing TMLE of linear functionals of a single $Q$.

\paragraph{\bf How to handle multiple functional parameters: }
Suppose $\Psi(P)=\Phi(\bar{Q}_1,\bar{Q}_2)$ for two functional parameters $\bar{Q}_j(P)=\arg\min_{\bar{Q}_j}PL_j(\bar{Q}_j)$, $j=1,2$.
For a given plug-in estimator $(\bar{Q}_{1n},\bar{Q}_{2n})$ we have
\[
\begin{array}{l}
\Phi(\bar{Q}_{1n},\bar{Q}_{2n})-\Phi(\bar{Q}_{10},\bar{Q}_{20})\approx
\Phi(\bar{Q}_{1n},\bar{Q}_{20})-\Phi(\bar{Q}_{10},\bar{Q}_{20})+\Phi(\bar{Q}_{10},\bar{Q}_{2n})-\Phi(\bar{Q}_{10},\bar{Q}_{20})
\end{array}
\]
up till a second order term. 
Therefore we can focus on analyzing $\bar{Q}_1\rightarrow \Phi(\bar{Q}_1,\bar{Q}_{20})$ and similarly for $\bar{Q}_2$. 
However, at that point, we have that our functional depends on unknown nuisance parameter which could complicate a TMLE since its targeting would be indexed by $\bar{Q}_{20}$ as well.
However, we also have
\[
\begin{array}{l}
\Phi(\bar{Q}_{1n},\bar{Q}_{2n})-\Phi(\bar{Q}_{10},\bar{Q}_{20})\approx
\Phi(\bar{Q}_{1n},\bar{Q}_{2n})-\Phi(\bar{Q}_{10},\bar{Q}_{2n})+\Phi(\bar{Q}_{1n},\bar{Q}_{2n})-\Phi(\bar{Q}_{1n},\bar{Q}_{20})
\end{array}
\]
 up till a second order term. So we could focus on analyzing $\Phi_{1n}(\bar{Q}_{1n})\equiv \Phi(\bar{Q}_{1n},\bar{Q}_{2n}) $ as estimator of $\Phi_{1n}(\bar{Q}_{10})$ and similarly we can define $\Phi_{2n}(\bar{Q}_{2n})\equiv \Phi(\bar{Q}_{1n},\bar{Q}_{2n})$ as estimator of $
\Phi_{2n}(\bar{Q}_{20})$. 
 If there is dependence on a marginal distribution of $W$ that is naturally and relatively precisely estimated with the empirical distribution, then one might focus on 
 $\Phi(Q_{W,n},\bar{Q}_n)-\Phi(Q_{W,n},\bar{Q}_0)$. However, this could result in lack of pathwise differentiability (even) of projection parameters, due to the empirical mean over $W$, even though at the true $Q_{W,0}$ these would be pathwise differentiable. In that case, we might still focus on the parameter $\bar{Q}\rightarrow \Phi(Q_{W,0},\bar{Q})$ and one might note that the canonical gradient of a pathwise derivative of this functional of $\bar{Q}$ would typically show no dependence on $Q_{W,0}$ so that it does not cause any complications in the development of a TMLE. 

So far we conclude it suffices to focus on analyzing $\Phi(\bar{Q})$ for a known $\Phi$ and a single functional parameter $\bar{Q}$. 

{\bf How to reduce analysis of non-linear functional to a linear functional of $\bar{Q}_P$:}
We now want to argue that we can focus on analyzing functionals $\Phi$ that are {\em linear} in $\bar{Q}$. 
We have
\[
\begin{array}{l}
\Phi(\bar{Q}_n)-\Phi(\bar{Q}_0)\approx \frac{\partial}{\partial\bar{Q}_0}\Phi(\bar{Q}_0)(\bar{Q}_n-\bar{Q}_0)\\
= \dot{\Phi}_0(\bar{Q}_n-\bar{Q}_0)\\
\approx \dot{\Phi}_{\bar{Q}_n}(\bar{Q}_n-\bar{Q}_0),
\end{array}
\]
again up till a second order difference in $(\bar{Q}_n-\bar{Q}_0)$.
Thus, we can define the linear functional $\dot{\Phi}_{\bar{Q}_n}(\bar{Q}_P)$ as our parameter of interest. 
Therefore, we conclude that we could focus on analyzing a known linear functional $\Psi(P)=\Phi(\bar{Q})$ with $\Phi$ being linear mapping in $\bar{Q}$.

{\bf Handling non-variation independent functional parameters:}
Above we considered the case that $\Psi(P)$ is a function of two or more variation independent functional parameters $\bar{Q}_j(P)$. Can we handle the case that $\Phi(\bar{Q}_1,\bar{Q}_2(\bar{Q}_1))$ where the second parameter $\bar{Q}_2$ is implied by $\bar{Q}_1$ and for a given $\bar{Q}_1 $ it is a minimizer of a risk function. 
For example, this situation occurs when the target functional is defined in terms of a sequential regression such as
$E(Y_1\mid V)= E( E(Y\mid A=1,W)\mid V)$, where $\bar{Q}_1(1,W)=E_P(Y\mid A=1,W)$ and $\bar{Q}_2(\bar{Q}_1)=E_P(\bar{Q}_1(1,W)\mid V)$. 
We then want to consider plug-in estimators $\Phi(\bar{Q}_{1n},\bar{Q}_{2n}(\bar{Q}_{1n}))$, where $\bar{Q}_{1n}$ is, for example, an HAL-MLE or HAL-TMLE and, for the given estimator $\bar{Q}_{1n}$, $\bar{Q}_{2n}(\bar{Q}_{1n})$ is an HAL-MLE or HAL-TMLE of $\bar{Q}_{20}(\bar{Q}_{1n})$. 
We have \[
\begin{array}{l}
\Phi(\bar{Q}_{1n},\bar{Q}_{2n}\bar{Q}_{1n})-\Phi(\bar{Q}_{10},\bar{Q}_{20}\bar{Q}_{10})=
\Phi(\bar{Q}_{1n},\bar{Q}_{20}\bar{Q}_{10})-\Phi(\bar{Q}_{10},\bar{Q}_{20}\bar{Q}_{10})\\
+
\Phi(\bar{Q}_{1n},\bar{Q}_{2n}\bar{Q}_{1n})-\Phi(\bar{Q}_{1n},\bar{Q}_{20}\bar{Q}_{10}).
\end{array}
\]
The first difference is of the form $\Phi_0(\bar{Q}_{1n})-\Phi_0(\bar{Q}_{10})$ and can thus be analyzed accordingly as covered by above remarks. We proceed with the second term as follows: 
\[
\begin{array}{l}
=\Phi(\bar{Q}_{1n},\bar{Q}_{2n}\bar{Q}_{1n})-\Phi(\bar{Q}_{1n},\bar{Q}_{20}\bar{Q}_{1n})+
\Phi(\bar{Q}_{1n},\bar{Q}_{20}\bar{Q}_{1n})-\Phi(\bar{Q}_{1n},\bar{Q}_{20}\bar{Q}_{10}).
\end{array}
\]
The first term is again of the type $\Phi_{2n}(\bar{Q}_{2n})-\Phi_{2n}(\bar{Q}_{20})$ and is thus covered by above. 
 So it remains to understand last term.
This is just another particular functional $\Phi_{n,0}(\bar{Q}_{1n})-\Phi_{n,0}(\bar{Q}_{10})$ and thus represents an additional plug-in estimator to analyze for $\bar{Q}_{1n}$. At this point this functional would still depend on unknown $\bar{Q}_{20}$ but, as explained above, we can replace it by $\bar{Q}_{2n}$ introducing another second order term. So we conclude that the analysis of such a plug-in estimator involves the study of three separate terms each one involving a function of a functional parameter that is a minimizer of a risk function:
\[
\begin{array}{l}
\Phi(\bar{Q}_{1n},\bar{Q}_{2n}\bar{Q}_{1n})-\Phi(\bar{Q}_{10},\bar{Q}_{20}\bar{Q}_{10})\approx
\left\{ \Phi(\bar{Q}_{1n},\bar{Q}_{2n}\bar{Q}_{1n})-\Phi(\bar{Q}_{10},\bar{Q}_{2n}\bar{Q}_{1n})\right\} \\
+
\left\{ \Phi(\bar{Q}_{1n},\bar{Q}_{2n}\bar{Q}_{1n})-\Phi(\bar{Q}_{1n},\bar{Q}_{20}\bar{Q}_{1n})\right \} +
\left\{ \Phi(\bar{Q}_{1n},\bar{Q}_{2n}\bar{Q}_{1n})-\Phi(\bar{Q}_{1n},\bar{Q}_{2n}\bar{Q}_{10})\right \} \\
\equiv \Phi_{1n}(\bar{Q}_{1n})-\Phi_{1n}(\bar{Q}_{10})+\Phi_{2n}(\bar{Q}_{2n})-\Phi_{2n}(\bar{Q}_{20})\\
+\Phi_{3n}(\bar{Q}_{1n})-\Phi_{2n}(\bar{Q}_{10}).
\end{array}
\]
Again, as remarked above, we can replace the possible non-linear functionals in linear functionals creating another second order difference.

 Overall, based on the above remarks, we conclude that an analysis of a plug-in estimator reduces to the analysis of linear functionals $\Psi(P)=\Phi(\bar{Q}_P)$ of a functional parameter $\bar{Q}_P=\arg\min_{\bar{Q}}PL(\bar{Q})$, introducing second order remainders that are easily converging to zero at a faster rate that $\bar{Q}_n$ converges to its target $\bar{Q}_0$ or $\bar{Q}_{{\cal R}(d),0}$. 


\newpage

\section{Computing canonical gradient of pathwise differentiable approximation of target function}
\label{app:canonicalgradient}

 One key ingredient of this TMLE for target functions is that we need to find the canonical gradient of $\alpha_{\lambda}:{\cal M}\rightarrow\openr^{n(\lambda)}$ defined by $\alpha_{\lambda,j}(P)=\langle \Phi(\bar{Q}_{\lambda,P}),\phi_{1,j}^*\rangle_{\psi}$, $j\in {\cal R}_1(\lambda)$. We can focus on finding this canonical gradient at a $P\in {\cal M}(\lambda)$. For log-likelihood behaving loss, we have that
 this canonical gradient equals the canonical gradient of $\alpha_{\lambda}:{\cal M}(\lambda)\rightarrow\openr^{n(\lambda)}$ that assumes the smaller model ${\cal M}(\lambda)$, or equivalently assumed that $\bar{Q}\in {\cal Q}(\lambda)$. Let's first consider the case that $D^{(k)}({\cal R}(d))=D^{(k)}([0,1]^d)$ so that ${\cal M}$ is a nonparametric model. 
 Let's denote $\alpha_{\phi}:{\cal M}\rightarrow\openr$ by $\alpha_{\phi}(P)=\langle \Phi(\bar{Q}_P),\phi\rangle_{\psi}$.
 Let $D^*_{\alpha_{\phi}(),P}$ be the canonical gradient at $P$ of $\alpha_{\phi}$ and not that this is linear in $\phi$.
 Note that $C_{\lambda}(P)(l)=\alpha_{\phi_{1,l}^*}(P)$, so that $D^*_{C_{\lambda}(),l,P}=(D^*_{\alpha_{\phi_{1,l}^*}(),P}$, $l\not\in {\cal R}_{1,\lambda})$. Let $\alpha_j(P)\equiv \langle \Phi(\bar{Q}_P),\phi_{1,j}^*\rangle_{\psi}$. Thus, $\alpha_j(P)=\alpha_{\phi_{1,j}^*}(P)$. 
 Therefore the canonical gradient of $\alpha_j$ at $P$ is given by $D^*_{\alpha_{\phi_{1,j}^*}(),P}$, $j\in {\cal R}_{1,\lambda}$.
 The canonical gradient of $\alpha_{\lambda,j}$ at $P$ is given by 
 \[
 D^*_{\alpha_{\lambda}(),j,P}=D^*_{\alpha_{\phi_{1,j}^*}(),P}-\Pi(D^*_{\alpha_{\phi_{1,j}^*},P}\mid H_{1,\lambda}),\]
 where $H_{1,\lambda}$ is the linear span of $D^*_{\phi_{1,l}^*,P}$, $l\not\in {\cal R}_{1,\lambda}$.
 We note that this can be represented as 
 \[
 D^*_{\alpha_{\lambda}(),j,P}=D^*_{\alpha_{\phi_{1,j}^*-\phi_{1,j}^{\perp}}(),P},\]
 where $\phi_{1,j}^{\perp}\in D^{(k)}({\cal R}_{1,\lambda}^c)$ is chosen so that $D^*_{\alpha_{\phi_{1,j}^{\perp}}(),P}$ equals the projection of $D^*_{\alpha_{\phi_{1,j}^*}(),P}$ onto $H_{1,\lambda}$. Note that $\phi_{1,j}^{\perp}$ is a linear combination $\phi_{1,l}^*$, $l\in {\cal R}_1(\lambda)^c$.
 
 So we can conclude that the only challenge in calculating the canonical gradient of $\alpha_{\lambda,j}:{\cal M}\rightarrow\openr$ is to determine
 the canonical gradient of $\alpha_{\phi}()$ at $P$ and one needs to carry out a projection onto $H_{1,\lambda}$.
 We can easily determine the canonical gradient $D^*_{\alpha_{\phi}(),P}$ of this in the nonparametric model (i.e., when $D^{(k)}({\cal R}(d))=D^{(k)}([0,1]^d)$: for example, for the DRC example and $\langle h_1,h_2\rangle_{\psi}=\int h_1(a)h_2(a) da$, we would have
$D^*_{\alpha_{\phi}(),P }=\phi(A)/g(A\mid W)(Y-\bar{Q}_P(A,W))$. 
Suppose now that $D^{(k)}({\cal R}(d))$ is strictly smaller than $D^{(k)}([0,1]^d)$. In that case we can obtain the canonical gradient of $D^*_{\alpha_{\phi}(),P}$ on this smaller model ${\cal M}(D^{(k)}({\cal R}(d))$ by projecting the one for the nonparametric model, say $D^*_{\alpha_{\phi}(),P,np}$, onto the tangent space of this actual model.
This tangent space is of the form $S_{h}(\bar{Q}_{\beta})=\frac{d}{d\bar{Q}}L(\bar{Q}_{\beta})(\sum_{m\in {\cal R}(d)}h(m)\phi_m)$ using paths $\beta+\delta h$. We also know that $D_{\alpha_{\phi}(),P}$ has similar form $D_{\alpha_{\phi}(),P}=\frac{d}{d\bar{Q}}L(\bar{Q})(\sum_{m}h_{np}^*(m)\phi_m)$ for a specified $h_{np}^*$ and where the sum is now over all spline basis functions spanning $D^{(k)}([0,1]^d)$. So we could write
$D_{\alpha_{\phi},P}=S_{h^*_{np}}(\bar{Q}_{\beta})$. 
Therefore, the projection of $D_{\alpha_{\phi}(),P,np}$ onto this tangent space typically corresponds with projecting $\sum_{m}h_{np}^*(m)\phi_m$ onto
$D^{(k)}({\cal R}(d))=\{\sum_{m\in {\cal R}(d)}h(m)\phi_m: h\}$ in some $L^2(\mu)$ Hilbert space. 
This projection is then represented by $D^*_{\alpha_{\phi}(),P}=S_{h^*}(\bar{Q}_{\beta})$ with $h^*\in D^{(k)}({\cal R}(d))$. 
A possible inverse weighting in $D_{\alpha_{\phi}(),P,np}$ would be stabilized in $D^*_{\alpha_{\phi}(),P}$ due to using a non-saturated submodel $D^{(k)}({\cal R}(d))$ of $D^{(k)}([0,1]^d)$. Determining the projection onto $H_{1,\lambda}$ is not any harder once we have determined these canonical gradients, so that, again, it is all about determining $D^*_{\alpha_{\phi}(),P}$.

In the next subsection we discuss a general approach for determining the canonical gradient of $D^*_{\alpha_{\phi}(),P}$ for general HAL-models $D^{(k)}({\cal R}(d))$.

We summarize this formula for the canonical gradient of $\alpha_{\lambda}:{\cal M}\rightarrow\openr^{n(\lambda)}$ (and thereby for our approximate target function $\Psi_{\lambda}(P)$) in the following lemma. 
\begin{lemma}
Define $\alpha_{\phi}:{\cal M}\rightarrow\openr$ by $\alpha_{\phi}(P)=\langle \Phi(\bar{Q}_P),\phi\rangle_{\psi}$.
 Let $D^*_{\alpha_{\phi}(),P}$ be the canonical gradient at $P$ of $\alpha_{\phi}$ and note that this is linear in $\phi$.
 Note that $C_{\lambda}(P)(l)=\alpha_{\phi_{1,l}^*}(P)$, so that $D^*_{C_{\lambda}(),l,P}=(D^*_{\alpha_{\phi_{1,l}^*}(),P}$, $l\not\in {\cal R}_{1,\lambda})$. Let $\alpha_j(P)\equiv \langle \Phi(\bar{Q}_P),\phi_{1,j}^*\rangle_{\psi}$. Thus, $\alpha_j(P)=\alpha_{\phi_{1,j}^*}(P)$. 
 Therefore the canonical gradient of $\alpha_j$ at $P$ is given by $D^*_{\alpha_{\phi_{1,j}^*}(),P}$, $j\in {\cal R}_{1,\lambda}$.
 The canonical gradient of $\alpha_{\lambda,j}$ at $P$ is given by 
 \[
 D^*_{\alpha_{\lambda}(),j,P}=D^*_{\alpha_{\phi_{1,j}^*}(),P}-\Pi(D^*_{\alpha_{\phi_{1,j}^*},P}\mid H_{1,\lambda}),\]
 where $H_{1,\lambda}$ is the linear span of $D^*_{\phi_{1,l}^*,P}$, $l\not\in {\cal R}_{1,\lambda}$ in $L^2_0(P)$.
 We note that this can be represented as 
 \[
 D^*_{\alpha_{\lambda}(),j,P}=D^*_{\alpha_{\phi_{1,j}^*-\phi_{1,j}^{\perp}}(),P},\]
 where $\phi_{1,j}^{\perp}\in D^{(k)}({\cal R}_{1,\lambda}^c)$ is chosen so that $D^*_{\alpha_{\phi_{1,j}^{\perp}} (),P}$ equals the projection of $D^*_{\alpha_{\phi_{1,j}^*}(),P}$ onto $H_{1,\lambda}$. Note that $\phi_{1,j}^{\perp}$ is a linear combination $\phi_{1,l}^*$, $l\in {\cal R}_1(\lambda)^c$.
 If $D^*_{\alpha_{\phi},P,np}$ is the canonical gradient for a nonparametric model ${\cal M}(D^{(k)}([0,1]^d)$ while ${\cal M}={\cal M}(D^{(k)}({\cal R}(d))$, then $D^*_{\alpha_{\phi}(),P}$ can be determined as the projection of latter on tangent space of model ${\cal M}(D^{(k)}({\cal R}(d))$.
 
 The canonical gradient of $P\rightarrow \Psi_{\lambda}(P)(x)$ is given by $D^*_{\Psi_{\lambda}(),x,P}=D^{*,\top}_{\alpha_{\lambda}(),P}{\bf \phi}_{1,\lambda}^*(x)$. 
\end{lemma}

\subsection{General approach for determining canonical gradient of pathwise differentiable parameters for HAL-models $D^{(k)}({\cal R}(d))$.}
We refer to Chapter 8 of \citep{vanderLaan&Rose18} about computing projections of initial gradients onto tangent space of HAL-models $D^{(k)}({\cal R}(d))$.

If one knows the canonical gradient $D^*_{\phi,P,np}$ in a nonparametric model ${\cal M}(D^{(k)}([0,1]^d))$, then one can focus on projecting it on the $L$-tangent space to obtain the canonical gradient $D^*_{\phi,P}$ for the actual model ${\cal M}(D^{(k)}({\cal R}(d))$. Let's first consider this case that $D^*_{\phi,P,np}$ is known. Let $D^0_{\beta}$ be such an initial gradient at $\bar{Q}_{\beta}$ only suppressing its dependence on $\bar{Q}$, since the other nuisance parameters are typically not updated so that we can view them as fixed. 
The tangent space for model $D^{(k)}({\cal R}(d))$ consists of $S_h(\beta)\equiv \frac{d}{d\bar{Q}_{\beta}}L(\bar{Q}_{\beta})(\sum_{j\in {\cal R}(d)}h(j)\phi_j)$ with $h$ varying over all vectors $h\in {\cal R}^N$ with $N=\mid {\cal R}(d)\mid$. 
One can then define 
\[
h^*_{\beta}\equiv\arg\min_h P\left( D^0_{\beta}(O)-S_h(\beta)(O)\right)^2.\]
Note that $S_h(\beta)=\sum_{j\in {\cal R}(d)} h(j) S_{e_j}(\beta)$, $e_{j_0}$ is the unit vector that is one can the $j_0$-th component so that 
$\sum_{j\in {\cal R}(d)}e_{j_0}(j)\phi_j=\phi_{j_0}$.
Thus, determining $h^*_{\beta}$ is reduced to a least squares regression problem:
\[
h^*_{\beta}\equiv \arg\min_h E_P\left( D^0_{\beta}(O)-\sum_{j\in {\cal R}(d)} h(j) S_{e_j}(\beta)(O)\right)^2.\]
$\bar{Q}_{\beta}$ does not fully determine the data distribution so that it remains to specify $P$. Therefore, one might replace $P$ by an estimator of $P_{\beta,n}$ of $P_0$ compatible with $\bar{Q}_P=\bar{Q}_{\beta}$.
Presumably this might be a continuous measure, so that the expectation $E_P$ is an integral. One could now decide to draw a large sample of $B$ observations from $P_{\beta,n}$ and replace the expectation by its empirical mean:
\[
h^*_{\beta,n}=\arg\min_h \frac{1}{B}\sum_{b=1}^B\left( D^0_{\beta}(O_b)-\sum_{j\in {\cal R}(d)}h(j) S_{e_j}(\beta)(O_b)\right)^2.\]
To approximate this solution one might use Lasso regression:
\[
h^*_{\beta,n}=\arg\min_{h,\pl h\pl_1<C_n} \frac{1}{B}\sum_{b=1}^B\left( D^0_{\beta}(O_b)-\sum_{j\in {\cal R}(d)}h(j) S_{e_j}(\beta)(O_b)\right)^2,\]
where one can select $C_n$ with cross-validation applied to this sample of $B$ observations. 
One might also estimate the canonical gradient with $h^*_{\beta,n}$ defined as follows:
\[
h^*_{\beta,n}\equiv \arg\min_ {h,\pl h\pl_1<C_n} P_n \left( D^0_{\beta}(O)-\sum_{j\in {\cal R}(d)} h(j) S_{e_j}(\beta)(O)\right)^2.\]
The latter is a particularly appealing method since it avoids having to simulate from the distribution $B$, while if $\hat{P}$ is a consistent estimator of $P_0$, this will still be a consistent estimator of the canonical gradient at $P$, and notice that it is still an element of the tangent space at $\bar{Q}_{\beta}$.

However, there are various problems where an initial gradient is not easily available. 
So let's also consider a method for computing the canonical gradient $D^*_{\phi,P}$ for model ${\cal M}(D^{(k)}({\cal R}(d))$ that does not rely on first computing an initial gradient for $\alpha_{\phi}:{\cal M}\rightarrow\openr$ at $P$. 
For the sake of demonstration, let's consider the DRC example for the model $D^{(k)}({\cal R}_{\lambda}(d))$ for $\bar{Q}_0$.
We could define a score operator $A_{\beta}: {\cal H}_{\beta}^F\rightarrow L^2_0(P)$ defined by 
$A_{\beta}(h)=\sum_{j\in {\cal R}_{\lambda}(d)}h(j)\phi_j(A,W) (Y-\bar{Q}_{\beta}(A,W))$. We can then put an inner product on ${\cal H}_{\beta}^F=\openr^{N}$ with $N=\mid {\cal R}_{\lambda}(d)\mid$, such as $\langle h_1,h_2\rangle^F=\sum_{j\in {\cal R}_{\lambda}(d)}h_1(j)h_2(j)$
Suppose that we aim to establish the canonical gradient of $P\rightarrow \Psi^F(\bar{Q}_{\beta_{\lambda}(P)})$ on ${\cal M}$, where this is equivalent with $P\rightarrow \Psi^F(\bar{Q}_{\lambda,P})$. Then, we can first determine the pathwise derivative
$\frac{d}{d\delta_0}\Psi^F(\bar{Q}_{\beta_{\lambda}+\delta_0 h})$ as a linear operator in $h\in H^F_{\lambda}=\{h(j): j\in {\cal R}_{\lambda}(d)\}$. By the Riesz-representation theorem we can write it as
$\langle D^F_{\beta_{\lambda}},h\rangle^F$ for some $D^F_{\beta_{\lambda}}\in {\cal H}^F_{\lambda}$.
For example, in the case that $\Psi^F(\bar{Q})(a)=E_W \bar{Q}(a,W)$, we have
\[
\begin{array}{l}
\frac{d}{d\delta_0}\Psi^F(\bar{Q}_{\beta_{\lambda}+\delta_0 h})=E_W \sum_{j\in {\cal R}_{\lambda}(d)}h(j)\phi_j(a,W)\\
=\sum_{j\in {\cal R}_{\lambda}(d)} h(j)E_W\phi_j(a,W)\\
=\langle D^F_a,h\rangle^F,
\end{array}
\]
where $D^F_a=(E_W\phi_j(a,W): j\in {\cal R}_{\lambda}(d))$.
Let's now determine the adjoint $A^{\top}_{\beta_{\lambda}}: L^2_0(P)\rightarrow {\cal H}_{\lambda}^F$. We have
\begin{eqnarray*}
\langle A_{\beta_{\lambda}}(h),S\rangle_P&=&E_P \sum_{j\in {\cal R}_{\lambda}(d)}h(j)\phi_j(Y-\bar{Q}_{\beta_{\lambda}}) S(W,A,Y)\\
&=&\sum_{j\in {\cal R}_{\lambda}(d)}h(j) \left\{ E_P \phi_j(A,W)(Y-\bar{Q}_{\beta_{\lambda}}(A,W)) S(W,A,Y)\right\}\\
&\equiv& \langle h, A^{\top}_{\beta_{\lambda}}(S)\rangle^F.
\end{eqnarray*}
Thus, the adjoint is given by:
\[
A^{\top}_{\beta_{\lambda}}(S)= (E_P\{\phi_j(A,W)(Y-\bar{Q}_{\beta_{\lambda}}(A,W)) S(W,A,Y)\}: j\in {\cal R}_{\lambda}(d)).\]
The information operator is then defined by $I_{\beta_{\lambda}}=A_{\beta_{\lambda}}^{\top}A_{\beta_{\lambda} }: {\cal H}^F_{\lambda}\rightarrow{\cal H}^F_{\lambda}$:
in the DRC example, it is given by
\[
\begin{array}{l}
I_{\beta_{\lambda}}(h)=\left( E_P\{\phi_j(A,W)(Y-\bar{Q}_{\beta_{\lambda}}(A,W)) \sum_{m\in {\cal R}_{\lambda}(d)}h(m)\phi_m(A,W)(Y-\bar{Q}_{\beta_{\lambda}})\}: j\in {\cal R}_{\lambda}(d)\right) \\
=\left( \sum_{m\in {\cal R}_{\lambda}(d)} h(m) E_P\left\{\phi_j\phi_m(A,W)(Y-\bar{Q}_{\beta_{\lambda}}(A,W))^2\right\}:j\in {\cal R}_{\lambda}(d)\right)\\
=\left( \sum_{m\in {\cal R}_{\lambda}(d)} h(m) E_P\phi_j\phi_m\sigma^2_{\beta_{\lambda}}: j\in {\cal R}_{\lambda}(d)\right),
\end{array}
\]
where $\sigma^2_{\beta}(A,W)= E((Y-\bar{Q}_{\beta_{\lambda}}(A,W))^2\mid A,W)$.
Thus, $I_{\beta_{\lambda}}$ can be represented by a matrix $I_{\beta_{\lambda}}(j,m)=E_P\phi_j\phi_m\sigma^2_{\beta_{\lambda}}$. 
The canonical gradient is given by $D^*_{\beta_{\lambda}}=A_{\beta_{\lambda}}I_{\beta_{\lambda}}^{-1} D^F_{\beta_{\lambda}}$, where $D^F_{\beta_{\lambda}}$ represents the canonical gradient of the pathwise derivative of $\beta\rightarrow \Phi(\bar{Q}_{\beta})$ as a linear operator on $H_{\lambda}^F$.
So this representation requires solving $I_{\beta_{\lambda}}(h)=D^F_{\beta_{\lambda},a}$. This corresponds with, 
\[
\sum_{m\in {\cal R}(d)} E_P\{\phi_j\phi_m\sigma^2_{\beta_{\lambda}}\} h(m)=E_W\phi_j(a,W)\mbox{ $j\in {\cal R}_{\lambda}(d)$}.\]
If ${\cal R}_{\lambda}(d)$ is a large set, then the inversion of the $N\times N$-matrix $I_{\beta_{\lambda}}$ might be computationally challenging and unstable. 
Therefore, one might approximate this vector $h=I_{\beta_{\lambda}}^{-1}D^F_{\beta_{\lambda},a}$ by defining $Y(j)\equiv E_W\phi_j(a,W)$, 
$X_j(m)=E_P\phi_j\phi_m(Y-\bar{Q}_{\beta_{\lambda}})^2$, and fit a least squares LASSO regression estimator
$\tilde{h}=\arg\min_{h,\pl h\pl_1<C} \sum_{j\in {\cal R}_{\lambda}(d)}( Y(j)-\sum_{m\in {\cal R}_{\lambda}(d)}h(m)X_j(m))^2$ using $L_1$-regularization. In this way we can use Lasso regression based on $N=\mid {\cal R}_{\lambda}(d)\mid $ observations to approximate the solution $h^*_{\beta_{\lambda}}=I_{\beta_{\lambda}}^{-1}D^F_{\beta}$. Note that $E_W\phi_j(a,W)$ would be estimated with empirical mean, and we could also estimate $E_P\phi_j\phi_m\sigma^2_{\beta_{\lambda}}$ with an empirical mean $P_n \phi_j\phi_m\sigma^2_{\beta_{\lambda}}$. 
The canonical gradient is then given by $D^*_{\beta_{\lambda}}(O)=\sum_{j\in {\cal R}_{\lambda}(d)}h^*_{\beta_{\lambda}}(j)\phi_j(Y-\bar{Q}_{\beta_{\lambda}}(A,W))$.


\section{Example: Causal effect of binary treatment on survival, conditional on covariates}\label{Chtargetfunction6}
We observe $O=(V,W,A,\tilde{T},\Delta)$, where $V\in [0,1]^{d_1}$ is a $d_1$-dimensional vector of baseline covariates whose treatment effect modification we aim to understand, while $W$ represents an additional vector of measured confounders, and $A$ is the binary treatment of interest. Here $\tilde{T}$ is the follow up time and $\Delta$ is an indicator if this time represents the time to failure of interest or that it just represents a censoring time. In terms of underlying full data $X=(V,W,A,C,T)$ we have $\tilde{T}=\min(T,C)$ and $\Delta=I(T\leq C)$. We assume that $C$ satisfies coarsening at random (CAR): $\lambda_C(t\mid W,A,T)=\lambda_C(t\mid W,A)$ at $t<T$.
Let the model ${\cal M}$ be nonparametric beyond models on $g_A=p_{A\mid W}$ and $\lambda_C$.
The density $p$ can be parametrized in terms of $(Q_W,g_A,\lambda_T,\lambda_C)$, where the conditional hazard $\lambda_{T}$ minimizes the log-likelihood risk under $P$ given by
\[
L(\lambda_T)=-E_P \log \prod_{s< \tilde{T}}(1-\Lambda_T(ds\mid A,V,W))\lambda_T(\tilde{T}\mid A,V,W)^{\Delta}.\]
We can also view the data structure as a longitudinal data structure over time $t$.
Let $N(t)=I(\tilde{T}\leq t,\Delta=1)$ and $C(t)=I(\tilde{T}\leq t,\Delta =0)$, and consider the ordering
$(W,A,(N(t),C(t): t\geq 0))$. Let $F(t)=(W,A,\bar{N}(t),\bar{C}(t))$. Then, due to CAR, $E(dN(t)\mid F(t-))=I(\tilde{T}\geq t)\lambda_T(t\mid V,W,A)$. 
The negative log-likelihood $L(\lambda_T)$ can also be represented as:
\[
L(\lambda_T)=-\int_s \left\{dN(s)\log \lambda(s\mid V,W,A)+(1-dN(s))\log(1-\lambda(s\mid V,W,A))\right\}.
\]
Similarly, we define the negative log-likelihood for $\lambda_C$ as $L(\lambda_C)$. 
The above applies to discrete and continuous survival times. Let's consider the case that $T$ is continuous or finely discrete so that 
we can model $\lambda_T(t\mid V,W,A)=\exp(Q(t,V,W,A))$ with $Q\in D^{(k)}(0,1]^d)$. We will also use notation $\lambda_Q$. 
This model for the conditional failure time hazard now also defines a statistical functional parameter $Q:{\cal M}\rightarrow D^{(k)}([0,1]^d)$ defined by $\lambda_P=\lambda_{Q(P)}$.
 Let $L(Q)=L(\lambda_{Q})$ be the log-likelihood risk defined above.
Let $S_Q(t\mid A,V,W)=\prod_{s\leq t}(1-\Lambda_Q(ds\mid A,V,W))$ be the conditional survival function of $T$ at $t$, given $A,V,W$, under $\lambda_Q$.
We define our target function of interest as $\Psi(P)(v)=E_P(S_Q(t_0\mid A=1,V,W)-S_Q(t_0\mid A=0,V,W)\mid V=v)$.
So $\Psi(P)=\Psi^F(Q(P))$ for a specified $\Psi^F$.
Under a structural causal model so that $T=T_A$ with $T_0,T_1$ being potential outcomes under static treatments $A=0$ and $A=1$, and the assumption that $A$ is independent of $(T_0,T_1)$, given $(V,W)$, 
this statistical target function represents $S_1(t_0\mid V)-S_0(t_0\mid V)$, the additive causal effect of the binary treatment on survival at $t_0$, conditional on $V$. 

{\bf Target function projection approach:} We define a parametric working model $\sum_{j\in {\cal R}_{1,\lambda}}\alpha(j)\phi_{1,j}(V)$ for $\Psi(P)(V)$; determine an orthonormal basis $\phi_{1,j}^*$, $j\in {\cal R}_{1,\lambda}$, w.r.t. an inner product $\langle \cdot,\cdot\rangle_{\psi}$ such as $\int_{[0,1]^{d_1}} h_1(v)h_2(v) dv$. A particular inner product is given by $\langle h_1,h_2\rangle_{\psi}=E_P h_1(V)h_2(V)$, where for the sake of parameter definitions we still treat this inner product as given/fixed, but it means that in practice we have to replace this inner product by its empirical analogue $P_n h_1h_2$. 
 We then define the pathwise differentiable approximation \[
 \Psi_{\lambda}(P)=\Psi_{\lambda}(Q_{W\mid V},Q)=\sum_{j\in {\cal R}_{1,\lambda}}\alpha_{\lambda,P}(j)\phi_{1,j}^*,\] where $\alpha_{\lambda,P}(j)\equiv \langle \Psi(P),\phi_{1,j}^*\rangle_{\psi}$.
 Thus, $\Psi_{\lambda}(P)$ depends on $P$ through the functional $Q(P)$ that defines the conditional failure time hazard, and the conditional distribution of $W$, given $V$. Of course, we could also have defined it as 
 \[
 \Psi_{\Lambda}(P)=\arg\min{\psi\in D^{(k_1)}({\cal R}_{1,\lambda})}\left\lvert \Psi(P)-\psi\right\rvert_{\psi}, 
 \]
 with no need for an orthonormal basis representation. 
 
 We can determine the canonical gradient $D^*_{\alpha_{\lambda}(),P}=(D^*_{\phi_{1,j}^*,P}: j\in {\cal R}_{1,\lambda})$ of $\alpha_{\lambda}()$, which just comes down to finding canonical gradient $D^*_{\phi,P}$ of $\alpha_{\phi}(P)=\langle \Psi(P),\phi\rangle_{\psi}$ for a given function $\phi \in D^{(k_1)}([0,1]^{d_1})$ (i.e., a function $\phi(V)$). 
 
 \begin{lemma}
 Consider the parameter $\alpha_{\phi}(P)=\langle \Psi(P),\phi\rangle_{\psi}$ so that
 $\Psi_{\lambda}(P)=\sum_{j\in {\cal R}_{1,\lambda}}\alpha_{\phi_{1,j}^*}(P)\phi_{1,j}^*$. Therefore the canonical gradient of $\alpha_{\phi}()$ is determined by the canonical gradient of $\alpha_{\phi}()$.
 We have 
 \[
 D^*_{\phi,P}=D^*_{\phi,P_{W\mid V},P}+D^*_{\phi,\Lambda_T(),P},\]
 where
 \begin{eqnarray*}
D^*_{\phi,P_{W\mid V},P}&=&\phi(V)\{S(t_0\mid A=1,V,W)-S(t_0\mid A=0,V,W)\}-\alpha_{\phi}(P)\\
 D^*_{\phi,\Lambda_T(),P}&=&\int H_{t_0,P}(s,A,V,W)\phi(V) (dN(s)-\Lambda_T(ds\mid A,V,W)),\end{eqnarray*}
and
\[
 H_{t_0,P}(s,A,V,W)= \frac{2A-1}{g_A(A\mid V, W)} I(s<t_0)\frac{1}{\bar{G}_C(s-\mid A,V,W)}\frac{S(t_0\mid A,V,W)}{S(s\mid A,V,W)}.
 \]
 \end{lemma}
 {\bf Proof:}
 Note that $\alpha_{\phi}(P)=
E_P E(S(t_0\mid A=1,V,W)-S(t_0\mid A=0,V,W)\mid V)\phi(V)$ depends on $P_V$, $P_{W\mid V}$ and $\Lambda$. 
We have 
\[
D^*_{\phi,Q_V(),P}=\{S(t_0\mid A=1,V)-S(t_0\mid A=0,V) \}\phi(V)-\alpha_{\phi}(P),\] 
 where $S(t_0\mid A=a,V)=E_{W\mid V}(S(t_0\mid A=a,V,W)\mid V)$.
 Let's now determine the $\Lambda_T$-component of $D^*_{\phi,P}$.
 The canonical gradient of $S(t_0\mid A=a,V,W)$ is given by $I(V=v,W=w)/P(v,w) D^*_{S_a(t_0),(v,w),P}$, where 
 $D^*_{S_a(t_0),(v,w),P}$ is the canonical gradient of treatment specific survival $S_a(t_0)$ when one samples from a conditional distribution $(V,W)=(v,w)$. We have the analytic expression for $D^*_{S_a(t_0),(v,w),P}$ which has the form $\int H_{(v,w),a,t_0,P}(t,A,W)(dN(t)-\Lambda_T(t\mid A,W))$ for its $\lambda_T$-component. Here $H_{v,w,a,t_0,P}$ is just the same time dependent clever covariate as in efficient influence curve of $S_a(t_0)$ but where $g(A\mid V,W)=g(A\mid V=v,W=w)$, $S(t\mid A=a,V,W)=S(t\mid A=a,V=v,W=w)$, and $\lambda_C(t\mid A,V,W)=\lambda_C(t\mid A,V=v,W=w)$. 
 Then, the $\lambda_T$-component of $D^*_{\phi,P}$ is given by
 \[
 \begin{array}{l}
 D^*_{\phi,\lambda_T(),P}=\int_v \int_w \frac{I_{v,w}(V,W)}{P(v,w)} D^*_{S_1(t_0)-S_0(t_0),(v,w),\lambda_T(),P}(\tilde{T},\Delta,A,v,w) dP(w\mid v) \phi(v) dP(v)\\
 =D^*_{S_1(t_0)-S_0(t_0),\lambda_T(),(V,W),P}(\tilde{T},\Delta,A,V,W) \phi(V)\\
 = D^*_{S_1(t_0)-S_0(t_0),\lambda_T,P}(O)\phi(V),
 \end{array}
 \]
 where $D^*_{S_1(t_0)-S_0(t_0),\lambda,P}$ is the $\lambda_T$-component of the canonical gradient of $S_1(t_0)-S_0(t_0)$.
 From the literature we know that \[
 D^*_{S_1(t_0)-S_0(t_0),\lambda_T(),P}=\int (H_{1,t_0,P}-H_{0,t_0,P})(s,V,W,A) (dN(s)-\Lambda_T(ds\mid A,V,W)),\]
 where
 \[
 H_{a,t_0,P}(t,A,V,W)= I(A=a)/g_A(a\mid V, W) I(t<t_0)\frac{1}{\bar{G}_C(t-\mid A=a,V,W)}\frac{S(t_0\mid A=a,V,W)}{S(t\mid A=a,V,W)}.
 \]
 So, we can conclude that
 \[
 D^*_{\phi,\lambda_T(),P}=\int (H_{1,t_0,P}-H_{0,t_0,P})(s,A,V,W)\phi(V) (dN(s)-\Lambda_T(ds\mid A,V,W))\]
 is determined by a single time-dependent clever covariate, the same as for targeting $S_1(t_0)-S_0(t_0)$ up till factor $\phi(V)$.
 Then,
 \[
 D^*_{\alpha_{\lambda}(),\Lambda_T(),P}=\int H_{{\bf \phi}_{1,\lambda}^*,t_0,P}(s,A,V,W) (dN(s)-\Lambda_T(ds\mid A,V,W)),\]
 defined by the $n(\lambda)$-dimensional time-dependent clever covariate $H_{{\bf \phi}_{1,\lambda}^*}$.
 
 It remains to determine the $P_{W\mid V}$-component of the canonical gradient $D^*_{\phi,P}$.
 This requires finding the canonical gradient of the following parameter
 \[
\Psi(P_{W\mid V})= \int \phi(v) \int_w \{S(t_0\mid 1,v,w)-S(t_0\mid 0,v,w)\}dP_W(w\mid v) dP_V(v).\]
By $\delta$-method argument, this one is given by 
\[
\begin{array}{l}
\int_v \phi(v) I(V=v)/P_V(v) \{S(t_0\mid A=1,V,W)-S(t_0\mid A=0,V,W)-\Psi(P)(V)\} dP_V(v)\\
=\phi(V)\{S(t_0\mid A=1,V,W)-S(t_0\mid A=0,V,W)-\Psi(P)(V)\}.
\end{array}
\]
This component of the canonical gradient of $D^*_{\phi,P}$ teaches us how to target the regression of $S(t_0\mid A=a,V,W)$ conditional on $V$. However, for this parameter $\alpha_{\phi}$ we can combine the $P_V$ and $P_{W\mid V}$ component and note that it is given by
\[
D^*_{\phi,P_{V,W},P}=\phi(V)\{S(t_0\mid A=1,V,W)-S(t_0\mid A=0,V,W)\}-\alpha_{\phi}(P).\Box\]

So we estimate $\alpha_{\phi}(P)$ with $P_n \{S_n(t_0\mid A=1,V,W)-S_n(t_0\mid A=0,V,W)\}\phi(V)$, and thereby automatically solve $P_n D^*_{\phi,P_{V,W}(),P_{V,W,n},\lambda_n}=0$ for any $\lambda_n$. Thus, there is no need for estimating the regression of $W$, given $V$.
Similarly, if we would not have used an orthonormal basis representation, we would still use the same targeting of the conditional hazard and then run a least squares regression of $S(t_0\mid A=1,W)-S(t_0\mid A=0,W)$ onto the linear model $\sum_{j\in {\cal R}_{1,\lambda}}\alpha(j)\phi_{1,j}$.

{\bf TMLE of $\Psi_{\lambda}(P_0)$:}
 We have
 \[
 D^*_{\alpha_{\lambda}(),\Lambda_T(),P}=\int H_{t_0,P}(s,A,V,W) {\bf \phi}_{1,\lambda}^*(V)( dN(s)-\Lambda_T(ds\mid A,V,W)),\]
 defined by an $n(\lambda)$-dimensional time-dependent clever covariate $H_{t_0,{\bf \phi}_{1,\lambda}^*,P}=H_{t_0,P}{\bf \phi}_{1,\lambda}^*(V)$.
 
So it is now straightforward to construct an HAL-MLE $Q_n$ of $Q_0$; $g_{A,n}$ of $g_{A,0}$, $\lambda_{C,n}$ of $\lambda_{C,0}$, and target $Q_n$ using a $n(\lambda)$-dimensional time-dependent clever covariate $H_{t_0,{\bf \phi}_{1,\lambda}^*,n}$ into an update $Q_{n,\lambda}^*$ solving 
$P_n D^*_{\alpha_{\lambda}(),Q_{n,\lambda}^*,g_{A,n},\lambda_{C,n}}=0$ as a vector equation. Let $P_n^*=(P_{V,W,n},Q_{n,\lambda}^*,g_{A,n},\lambda_{C,n})$.
This then yields the TMLE $\Psi_{\lambda}(P_{n,\lambda}^*)=\alpha_{\lambda}(P_{n,\lambda}^*)^{\top}{\bf \phi}_{1,\lambda}^*$ solving
$P_n D^*_{\Psi_{\lambda}(),v,P_{n,\lambda}^*}=0$ uniformly in $v$ so that it yields a TMLE of the whole target function. 

{\bf Implication of our asymptotic normality results:}
Our results show asymptotic linearity with estimated influence curve $D^*_{\Psi_{\lambda}(),v,P_{n,\lambda}^*}=D^{*,\top}_{{\bf \phi}_{1,\lambda}^*,P_{n,\lambda}^*}{\bf \phi}_{1,\lambda}^*(v)$, and corresponding asymptotic normality of $\Psi_{\lambda}(P_{n,\lambda}^*)-\Psi_{\lambda}(P_0)$ at a rate $(n(\lambda)/n)^{1/2}$. Selecting $\lambda=\lambda_n$ so that $n(\lambda_n)$ is such that the uniform approximation error $\Psi_{\lambda}(P_0)-\Psi(P_0)=O^+(n(\lambda)^{-k_1^*})$ is negligible relative to $(n(\lambda)/n)^{1/2}$ by some $\log n$-factor then also yields
\[
(n/n(\lambda_n))^{1/2}(\Psi_{\lambda}(P_n^*)-\Psi(P_0))(v)/\sigma_n(v)\Rightarrow_d N(0,1),\]
pointwise for each $v$, where $\sigma_n^2(v)$ is the sample variance of $n(\lambda)^{-1/2} D^*_{\Psi_{\lambda}(),v,P_n^*}$. Here $P_n^*=P_{n,\lambda_n}^*$.
This yields pointwise 0.95-confidence intervals $\Psi_{\lambda_n}(P_n^*)(v)\pm 1.96 \sigma_n(v)n(\lambda_n)^{1/2}/n^{1/2}$.
This is equivalent with $\Psi_{\lambda_n}(P_n^*)(v)\pm 1.96\sigma_{n,\lambda_n}(v)/n^{1/2}$, where now
$\sigma^2_{n,\lambda_n}(v)=P_n D^*_{\Psi_{\lambda_n}(),v,P_n^*}$ not normalizing the EIC to have bounded variance. 
Moreover, we can construct simultaneous confidence bands based on the influence curve across $v$, at cost of $ \log n$-factor in width. These could be constructed in the standard manner as done for vectors of target parameters based on the multivariate normal limit distribution with covariance matrix of the vector influence curve. 

 Overall we can conclude that this TMLE of the target function is straightforward, including its inference, and it corresponds with a TMLE for a multivariate target parameter. The main complication relative to TMLE for pathwise differentiable target parameters is the tuning of $\lambda$. As discussed earlier, we could use Lepski's method for selecting of $\lambda$, involving constructing a TMLE for a sequence of $\lambda_j$-values so that we can trade-off change in estimate with change in standard error. 
 
 If $V$ is higher dimensional, it can be challenging to define the appropriate sieve a priori. In that case, we can use a data adaptive sieve TMLE or the Targeted HAL-MLE, where the latter can use a large initial working model $D^{(k)}({\cal R}_{1,N})$, and use the LASSO in the TMLE step targeting $\Psi_{{\cal R}_{1,N}}(P_0)$. 
 In this problem that corresponds with logistic or Cox LASSO with appropriate selection of the $L_1$-norm. One could still decide to use the variance of the EIC of $\Psi_{{\cal R}_{1,\lambda}}()$ for inference, even though that model might be larger than warranted by the rate of convergence. One might suspect that the inference would be conservative. Alternatively, we can use the targeted bootstrap.

\section{Additional Simulation Study Figures}

\begin{figure}[H]
 \centering
 \includegraphics[width=0.9\textwidth]{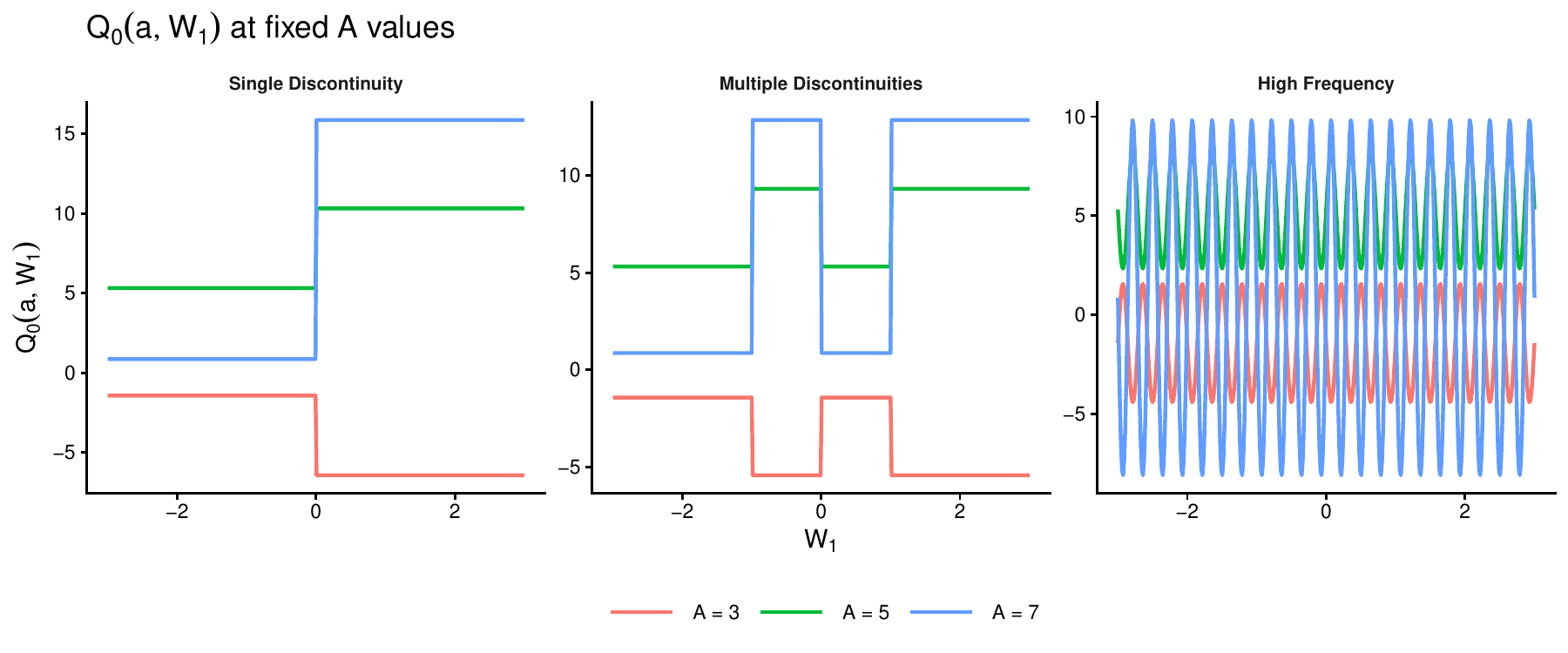}
 \caption{Plots of the functional forms of $\bar Q_0(A,W)$ for various DGP settings at $A=\{3,5,7\}$.}
 \label{fig:q0}
\end{figure}

\begin{figure}[H]
 \centering
 \includegraphics[width=0.4\textwidth]{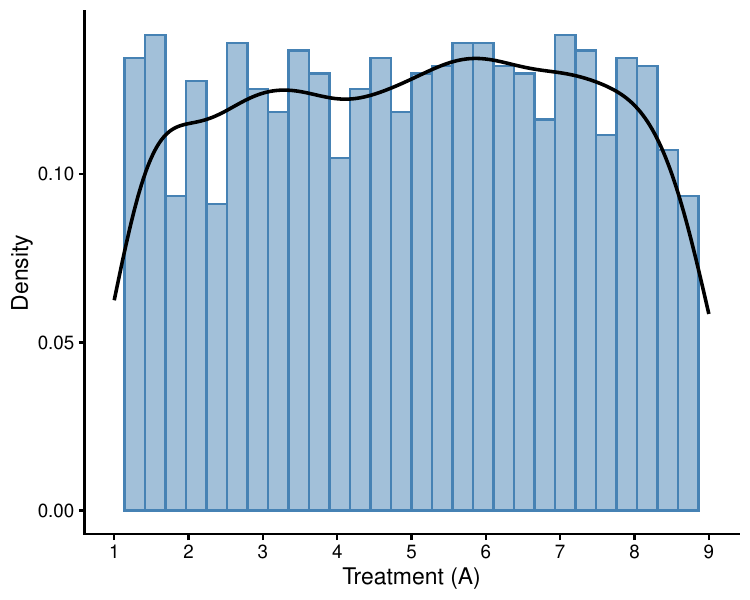}
\includegraphics[width=0.4\textwidth]{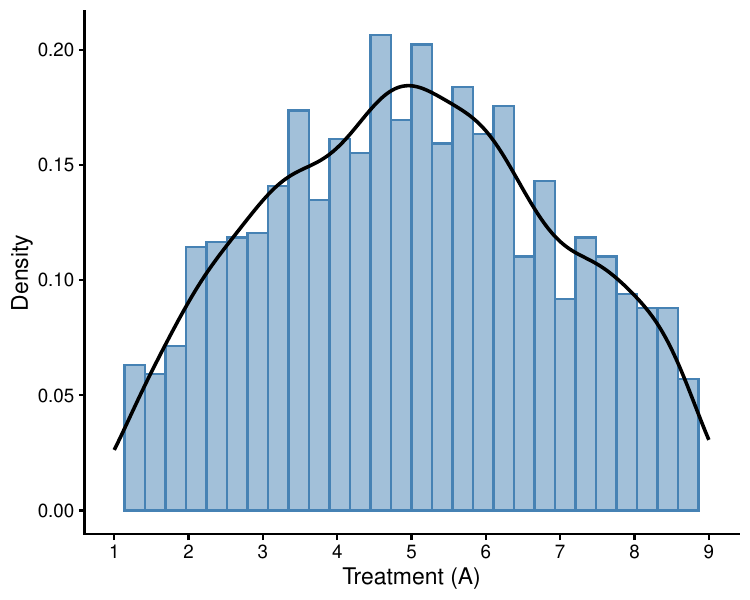}
 \caption{Density of the treatment $A$ for one replicate at $n=2000$ under uniform (left) and normal (right) distributions.}
 \label{fig:treatmentdist}
\end{figure}

\begin{figure}[H]
 \centering
 \includegraphics[width=0.9\textwidth]{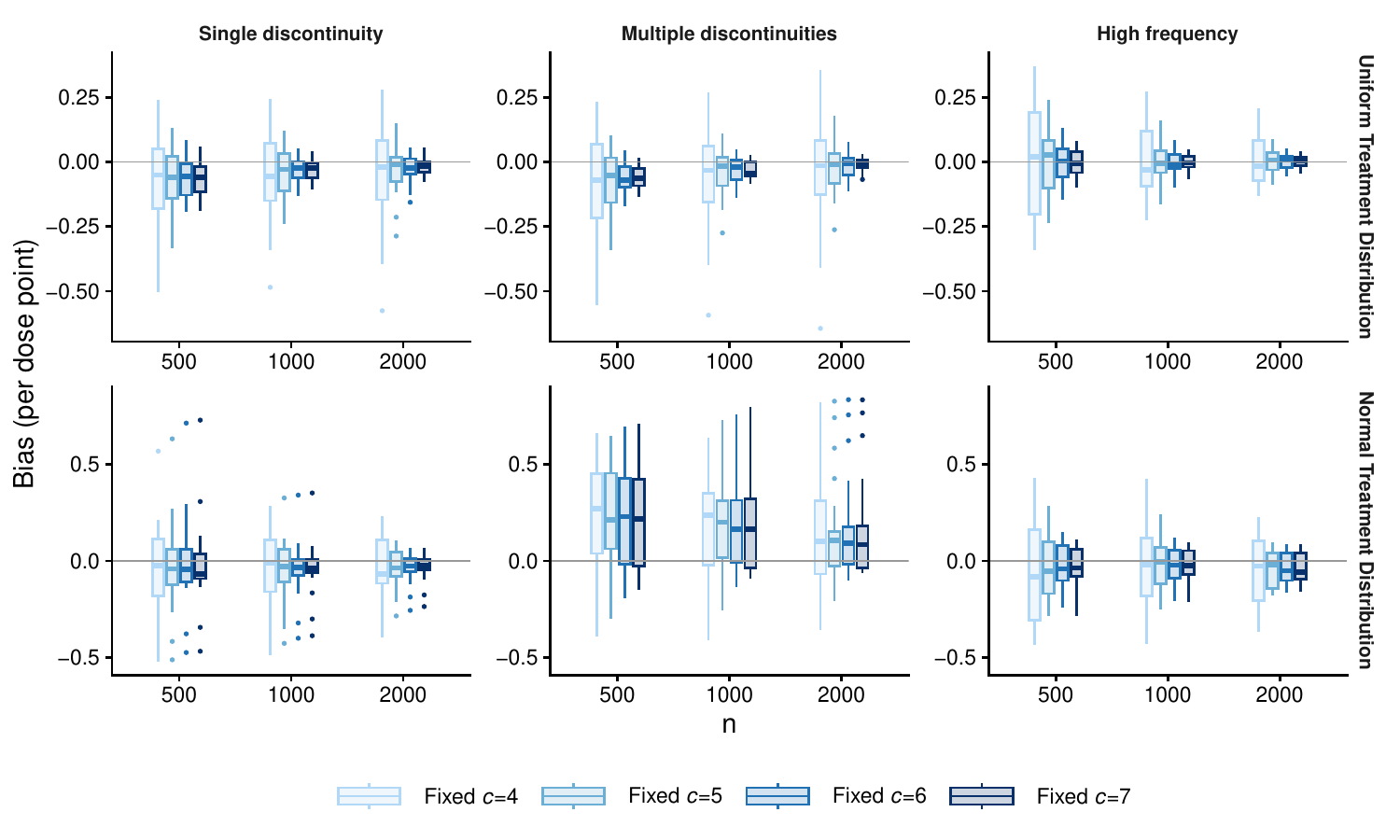}
 \caption{Boxplots of the bias across the $m=25$ dose evaluation points $a_j\in[1,9]$, for each estimator with fixed sieve size, DGP, sample size $n\in\{500,1000,2000\}$, and treatment distribution.
 Each observation is the average bias at one test point. Sieve size is $c\cdot n^{1/5}$.}
 \label{fig:bias}
\end{figure}

\begin{figure}[H]
 \centering
 \includegraphics[width=0.9\textwidth]{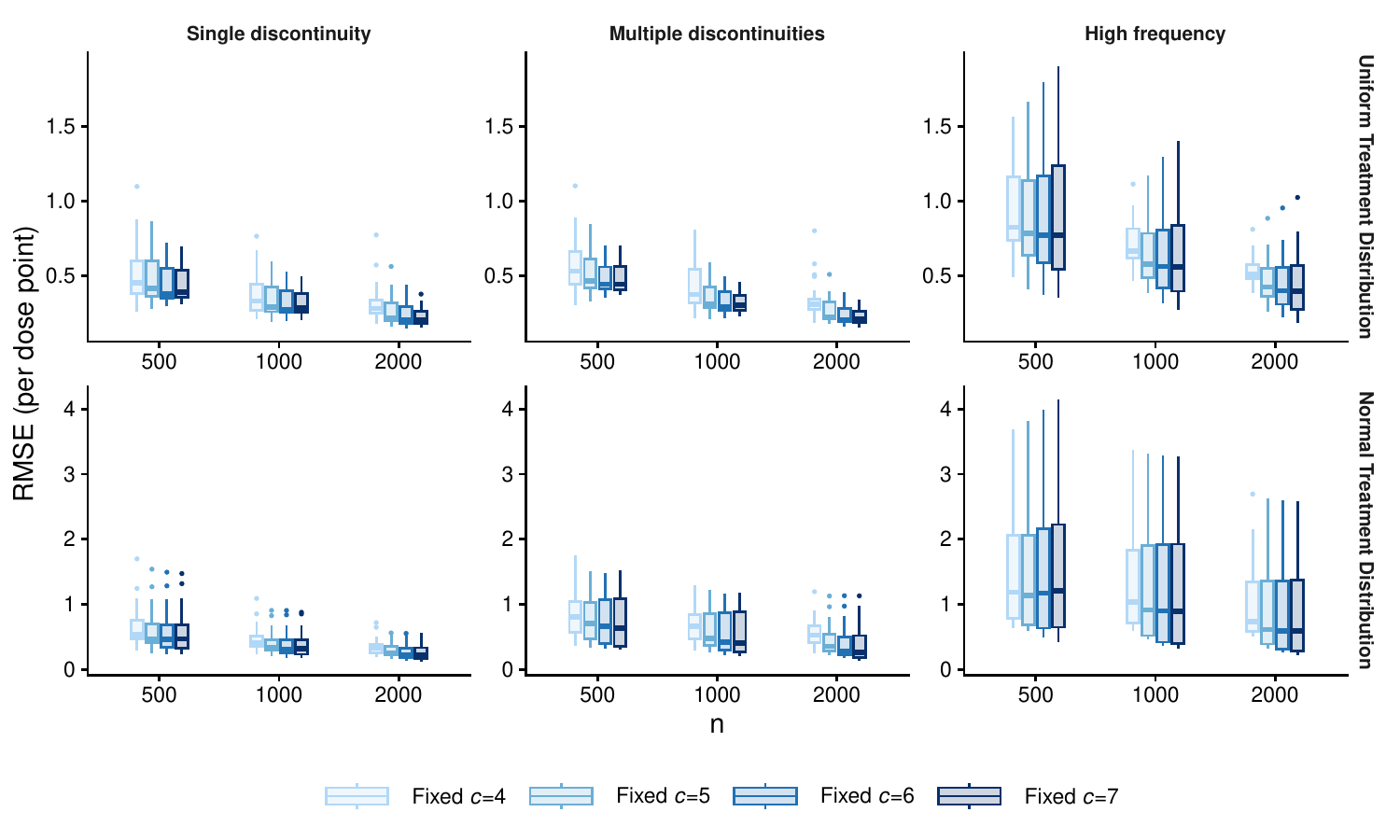}
 \caption{Boxplots of the RMSE across the $m=25$ dose evaluation points $a_j\in[1,9]$, for each estimator with fixed sieve size, DGP, sample size $n\in\{500,1000,2000\}$, and treatment distribution.
 Each observation in a boxplot is the average RMSE at one test point. Sieve size is $c\cdot n^{1/5}$.}
 \label{fig:rmse}
\end{figure}

\begin{figure}[H]
 \centering
 \includegraphics[width=0.9\textwidth]{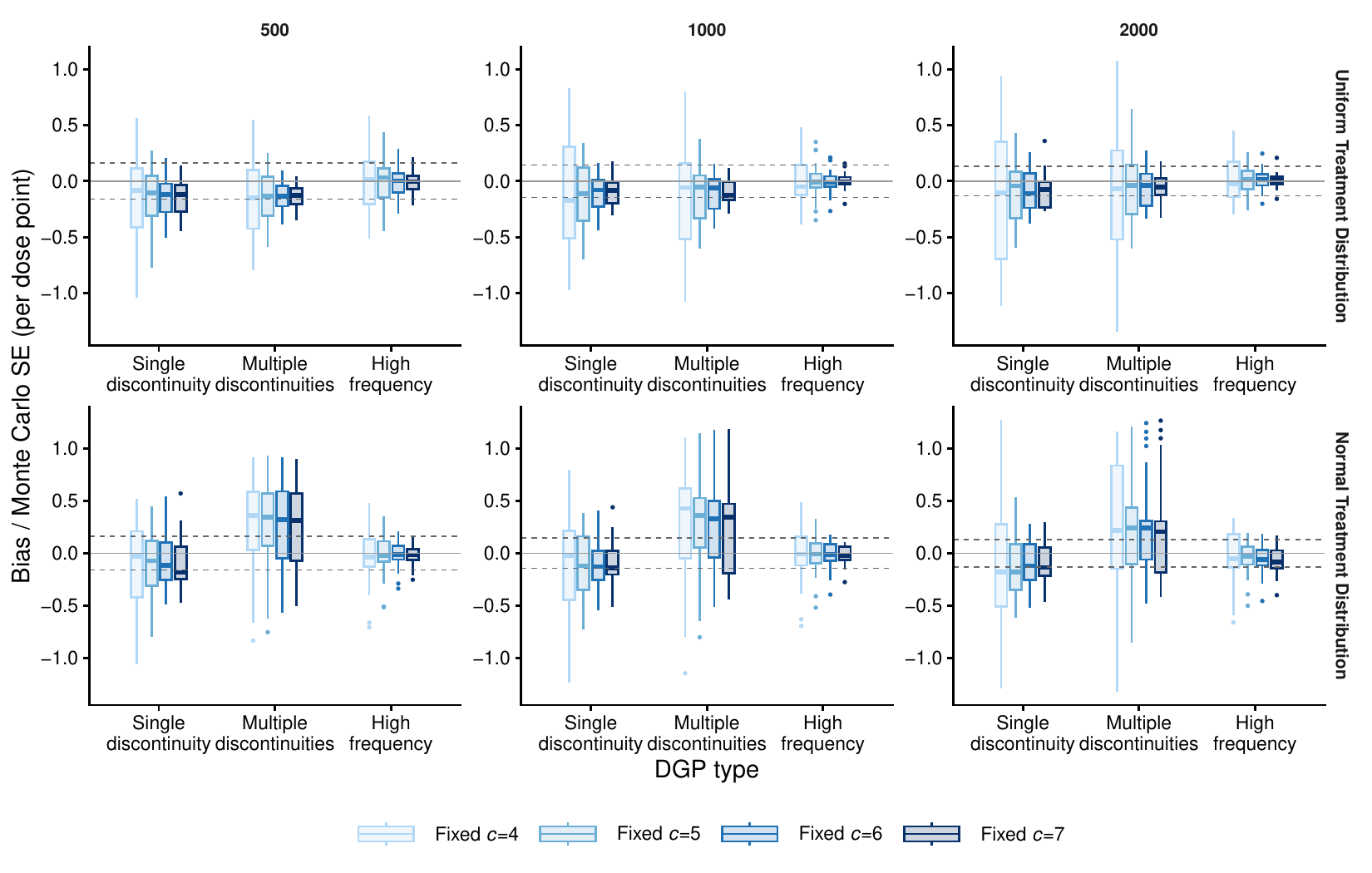}
 \caption{Boxplots of the Bias/MC-SE across the $m=25$ dose evaluation points $a_j\in[1,9]$, for each estimator with fixed sieve size, DGP, sample size $n\in\{500,1000,2000\}$, and treatment distribution.
 Each observation is the average Bias/MC-SE at one test point. Sieve size is $c\cdot n^{1/5}$.
 Dashed lines mark $\pm 1/\log(n)$.}
 \label{fig:ratio}
\end{figure}

\begin{table}[!h]
\centering
\caption{Bias and RMSE of the oracle projection vs.\ the true DRC, averaged over dose points.}
\label{tab:oracleproj}
\centering
\resizebox{\ifdim\width>\linewidth\linewidth\else\width\fi}{!}{
\begin{tabular}[t]{llrcccc}
\toprule
\multicolumn{3}{c}{ } & \multicolumn{2}{c}{Adaptive C=6 Oracle Projection} & \multicolumn{2}{c}{Fixed C=6 Oracle Projection} \\
\cmidrule(l{3pt}r{3pt}){4-5} \cmidrule(l{3pt}r{3pt}){6-7}
DGP & Treatment Distribution & $n$ & Bias & RMSE & Bias & RMSE\\
\midrule
 &  & 500 & \textbf{-0.0016} & \textbf{0.1577} & -0.0022 & 0.1723\\

 &  & 1000 & \textbf{-0.0022} & \textbf{0.1230} & -0.0026 & 0.1337\\

 & \multirow{-3}{*}{\raggedright\arraybackslash Uniform} & 2000 & \textbf{-0.0016} & \textbf{0.0920} & -0.0026 & 0.1147\\

 &  & 500 & \textbf{-0.0002} & \textbf{0.1759} & -0.0013 & 0.1983\\

 &  & 1000 & \textbf{-0.0002} & \textbf{0.1280} & -0.0004 & 0.1460\\

\multirow{-6}{*}[0.5\dimexpr\aboverulesep+\belowrulesep+\cmidrulewidth]{\raggedright\arraybackslash Single discontinuity} & \multirow{-3}{*}{\raggedright\arraybackslash Normal Treatment Distribution} & 2000 & \textbf{-0.0002} & \textbf{0.0912} & -0.0003 & 0.1096\\
\cmidrule{1-7}
 &  & 500 & \textbf{-0.0011} & \textbf{0.1669} & -0.0022 & 0.1879\\

 &  & 1000 & \textbf{-0.0012} & \textbf{0.1284} & -0.0017 & 0.1444\\

 & \multirow{-3}{*}{\raggedright\arraybackslash Uniform Treatment Distribution} & 2000 & \textbf{-0.0005} & \textbf{0.0977} & -0.0012 & 0.1199\\

 &  & 500 & \textbf{-0.0008} & \textbf{0.1881} & -0.0011 & 0.2152\\

 &  & 1000 & \textbf{0.0000} & \textbf{0.1476} & -0.0001 & 0.1658\\

\multirow{-6}{*}[0.5\dimexpr\aboverulesep+\belowrulesep+\cmidrulewidth]{\raggedright\arraybackslash Multiple discontinuities} & \multirow{-3}{*}{\raggedright\arraybackslash Normal Treatment Distribution} & 2000 & \textbf{-0.0001} & \textbf{0.1077} & -0.0004 & 0.1323\\
\cmidrule{1-7}
 &  & 500 & 0.0008 & \textbf{0.2112} & \textbf{0.0007} & 0.2398\\

 &  & 1000 & 0.0000 & \textbf{0.1738} & 0.0000 & 0.1913\\

 & \multirow{-3}{*}{\raggedright\arraybackslash Uniform Treatment Distribution} & 2000 & \textbf{-0.0003} & \textbf{0.1148} & -0.0005 & 0.1391\\

 &  & 500 & \textbf{0.0003} & \textbf{0.2006} & 0.0004 & 0.2286\\

 &  & 1000 & 0.0004 & \textbf{0.1562} & \textbf{0.0003} & 0.1788\\

\multirow{-6}{*}[0.5\dimexpr\aboverulesep+\belowrulesep+\cmidrulewidth]{\raggedright\arraybackslash High frequency} & \multirow{-3}{*}{\raggedright\arraybackslash Normal Treatment Distribution} & 2000 & \textbf{0.0003} & \textbf{0.1060} & 0.0006 & 0.1423\\
\bottomrule
\end{tabular}}
\end{table}

\begin{figure}[H]
 \centering
 \includegraphics[width=0.9\textwidth]{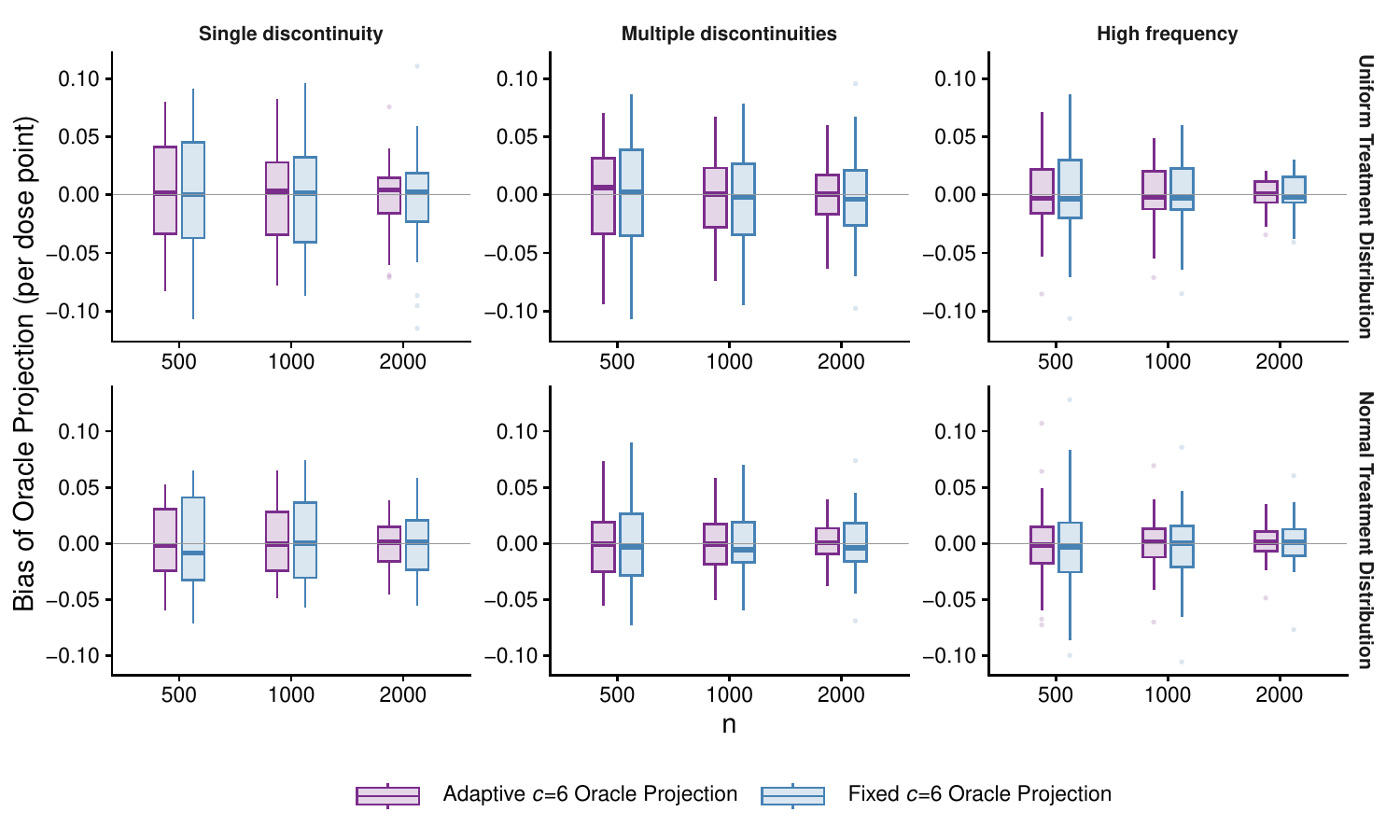}
 \caption{Boxplots of the oracle projection bias across the $m=25$ dose evaluation points, for each DGP, sample size, and treatment distribution.
 Each observation is the oracle projection bias at one test point.  Sieve size is $c\cdot n^{1/5}$.}
 \label{fig:oracle_proj_bias}
\end{figure}

\begin{figure}[H]
 \centering
 \includegraphics[width=0.9\textwidth]{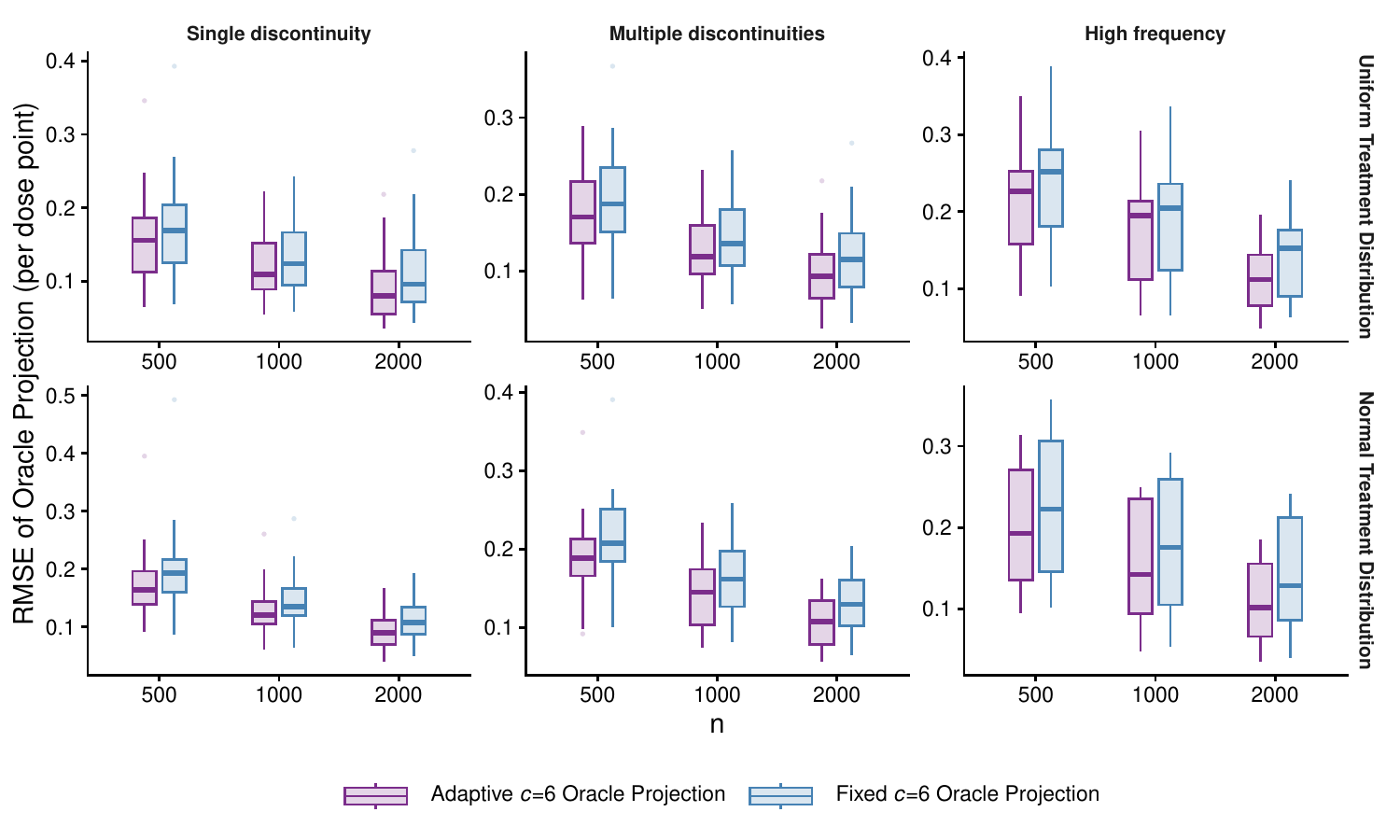}
 \caption{Boxplots of the oracle projection RMSE across the $m=25$ dose evaluation points, for each DGP, sample size, and treatment distribution.  
 Each observation is the RMSE of the Oracle-HAL estimate relative to the truth $\psi_0(a_j)$ at one test point.}
 \label{fig:oracle_proj_rmse}
\end{figure}

\begin{figure}[H]
 \centering
 \includegraphics[width=0.9\textwidth]{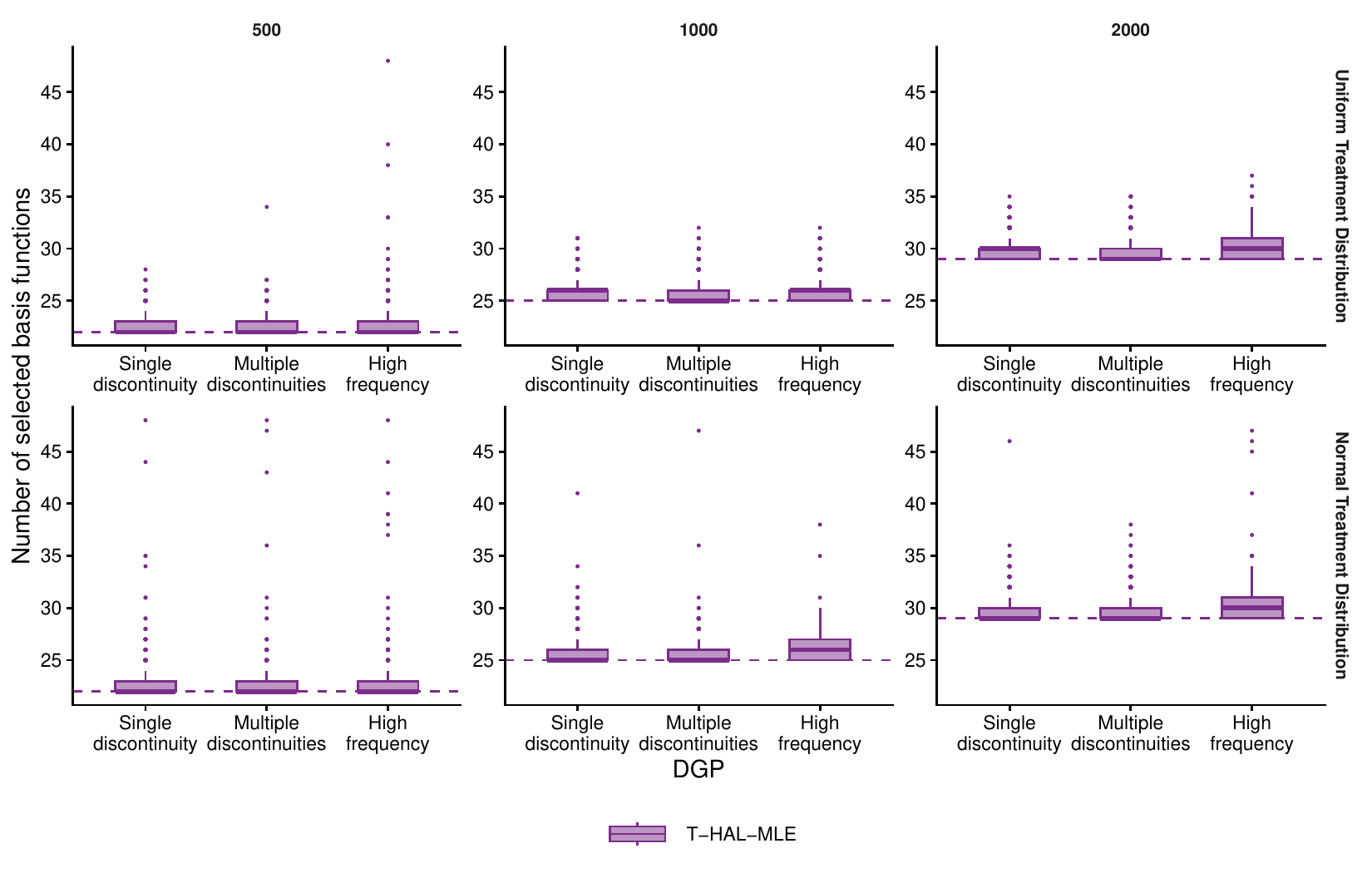}
 \caption{Distribution of the number of basis functions selected by the T-HAL-MLE targeting step across $B=1{,}000$ Monte Carlo replications, for each DGP, sample size, and treatment distribution, under the adaptive selector strategy with $c_1=6$ and $c_2=9$. }
 \label{fig:basis_size}
\end{figure}

\begin{figure}[H]
 \centering
 \includegraphics[width=0.95\textwidth]{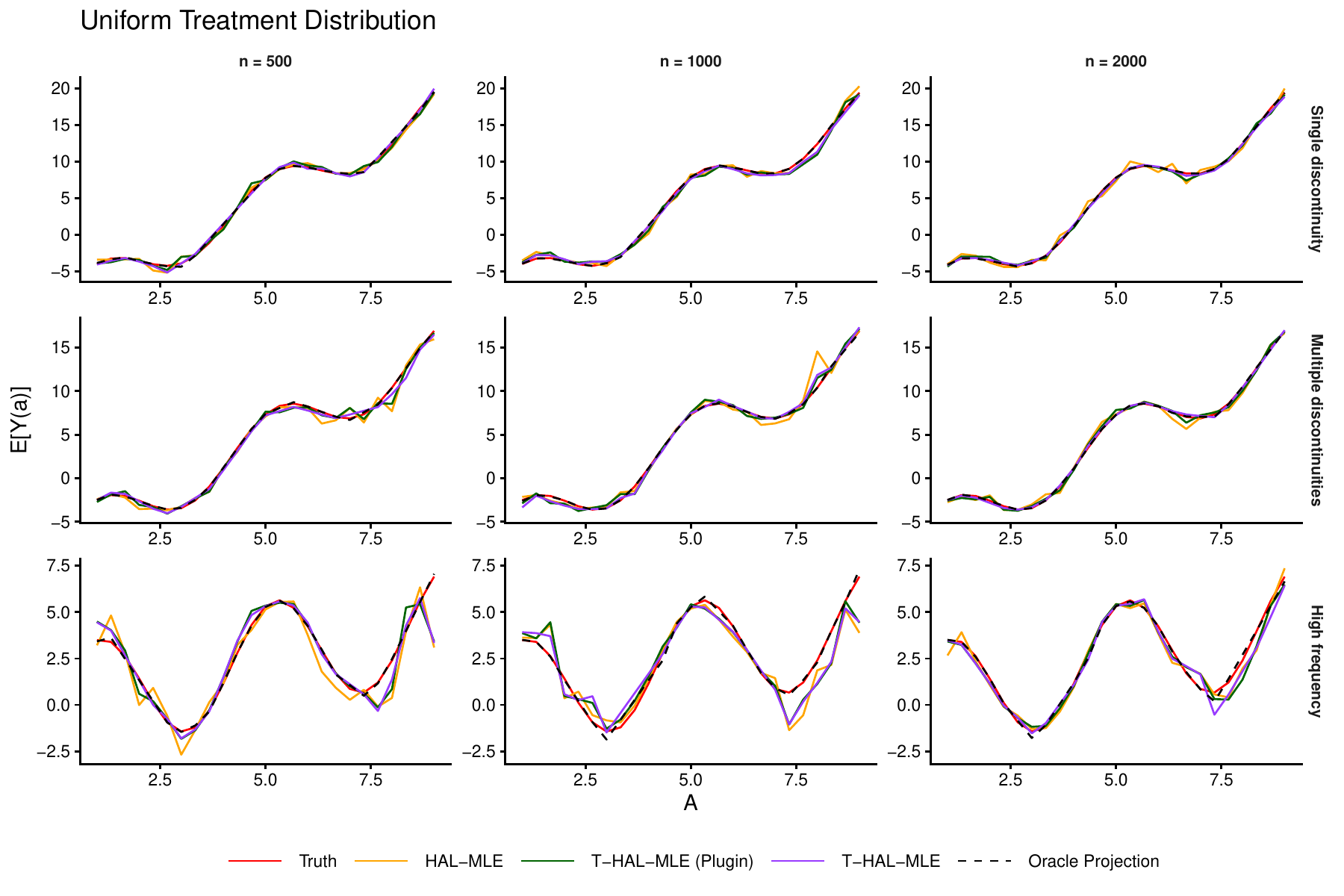}
 \includegraphics[width=0.95\textwidth]{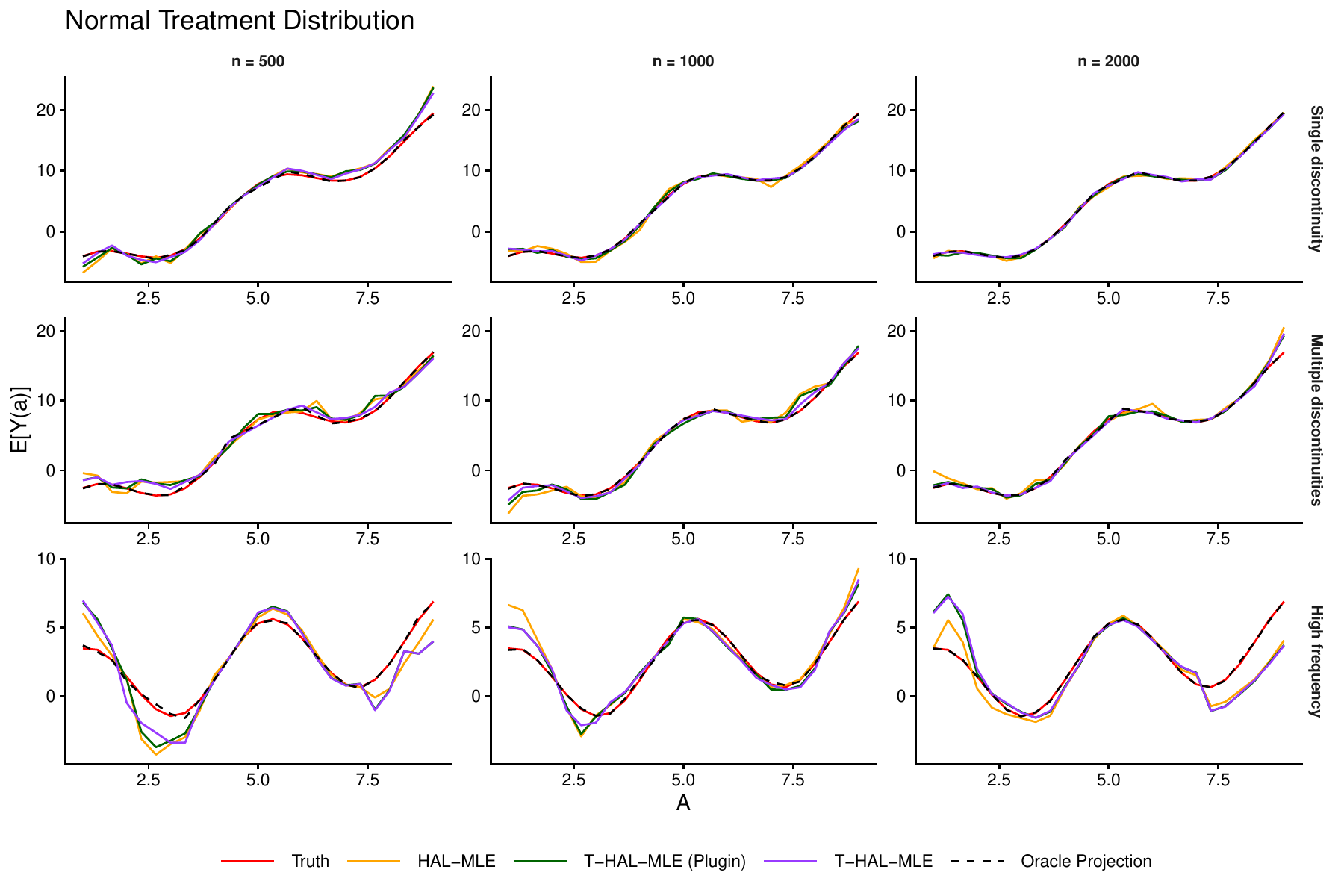}

 \caption{Estimated dose-response curves under uniform (top) and normal (bottom) treatment distributions for one replicate at $n=\{500, 1{,}000, 2000\}$.
 Each row corresponds to one outcome distribution.
 The red line is the true DRC, and the black dashed is the projection of the true DRC onto the basis selected during the LASSO targeting stage, with LASSO penalty chosen through the adaptive selector strategy.}
 \label{fig:drc}
\end{figure}

\end{document}